\documentclass[a4paper]{article}     

\usepackage{amsfonts}
\usepackage{amsmath,amssymb,url,ifthen}
\usepackage{mathtools}
\usepackage{times}
\usepackage{pxfonts}
\usepackage{pifont}

\usepackage{epsfig}
\usepackage{float}
\usepackage{fancybox}
\usepackage{color}
\usepackage{graphicx}
\usepackage{fancyhdr}
\usepackage{longtable}

\usepackage{amssymb,amsmath,amsfonts,amstext,url}
\usepackage{graphicx,color}
\usepackage{array}
\usepackage{pgf,pgfarrows,pgfnodes,pgfautomata,pgfheaps,pgfshade}
\usepackage{epsfig}
\usepackage{multicol}
\usepackage{pstricks,pst-grad}
\usepackage{subfigure}
\usepackage{index}

\newcommand{\algofont}{\mathtt}
\newcommand{\setup}{\algofont{setup}}

\newcommand{\keygen}{\algofont{keygen}}
\newcommand{\encrypt}{\algofont{encrypt}}

\newcommand{\decrypt}{\algofont{decrypt}}
\newcommand{\encap}{\algofont{encap}}
\newcommand{\decap}{\algofont{decap}}
\newcommand{\sign}{\algofont{sign}}
\newcommand{\verify}{\algofont{verify}}
\newcommand{\convert}{\algofont{convert}}
\newcommand{\compute}{\algofont{compute}}
\newcommand{\retrieve}{\algofont{retrieve}}
\newcommand{\proveValidity}{\algofont{proveValidity}}

\newcommand{\verifyconverted}{\algofont{verifyconverted}}
\newcommand{\sigExtract}{\algofont{sigExtract}}
\newcommand{\sigVerify}{\algofont{sigVerify}}
\newcommand{\commit}{\algofont{commit}}
\newcommand{\open}{\algofont{open}}
\newcommand{\confirm}{\algofont{confirm}}
\newcommand{\sconfirm}{\algofont{sconfirm}}
\newcommand{\deny}{\algofont{deny}}

\newcommand{\signcrypt}{\algofont{signcrypt}}
\newcommand{\unsigncrypt}{\algofont{unsigncrypt}}

\newcommand{\entfont}{\mathsf}
\newcommand{\signer}{\entfont{S}}

\newcommand{\Sim}{\entfont{Sim}}
\newcommand{\V}{\entfont{V}}
\newcommand{\C}{\entfont{C}}
\newcommand{\prover}{\entfont{P}}

\newcommand{\A}{{\cal A}}
\newcommand{\R}{{\cal R}}

\newcommand{\varfont}{\mathit}
\newcommand{\param}{\varfont{param}}
\newcommand{\sk}{\varfont{sk}}
\newcommand{\pk}{\varfont{pk}}

\newcommand{\sks}{\sk_\entfont{S}}
\newcommand{\skp}{\sk_\entfont{P}}
\newcommand{\pks}{\pk_\entfont{S}}
\newcommand{\skc}{\sk_\C}
\newcommand{\pkc}{\pk_\C}

\newcommand{\skr}{\sk_{\entfont{R}}}
\newcommand{\pkr}{\pk_{\entfont{R}}}

\newcommand{\coins}{\varfont{coins}}

\newcommand{\SC}{\varfont{SC}}
\newcommand{\cs}{\varfont{CS}}
\newcommand{\out}{\varfont{out}}
\newcommand{\done}{\varfont{done}}
\newcommand{\pok}{\entfont{PoK}}
\newcommand{\kem}{{\mathcal{K}}}
\newcommand{\dem}{{\mathcal{D}}}

\newcommand{\crs}{\mathsf{crs}}
\newcommand{\ta}{\mathsf{TA}}
\newcommand{\nizk}{\mathsf{NIZK}}
\newcommand{\zkp}{\mathsf{ZKP}}
\newcommand{\zk}{\mathsf{ZK}}

\newcommand{\secfont}{\mathsf}

\newcommand{\OW}{\secfont{OW}}
\newcommand{\IND}{\secfont{IND}}

\newcommand{\NM}{\secfont{NM}}

\newcommand{\INV}{\secfont{INV}}

\newcommand{\CMA}{\secfont{CMA}}
\newcommand{\CPA}{\secfont{CPA}}
\newcommand{\PCA}{\secfont{PCA}}
\newcommand{\CCA}{\secfont{CCA}}
\newcommand{\OT}{\secfont{OT}}

\newcommand{\GOAL}{\secfont{goal}}
\newcommand{\ATK}{\secfont{atk}}

\newcommand{\classfont}{\ensuremath{\mathbb}}
\newcommand{\bbbe}{\classfont{E}}
\newcommand{\bbbh}{\classfont{H}}

\newcommand{\bbbs}{\classfont{S}}

\newcommand{\bbbz}{\classfont{Z}}
\newcommand{\bbbg}{\classfont{G}}
\newcommand{\bbbc}{\classfont{C}}

\newenvironment{proof}[1][Proof]{\begin{trivlist}
\item[\hskip \labelsep {\bfseries #1}]}{\end{trivlist}}

\newcommand{\qed}{\nobreak \ifvmode \relax \else
      \ifdim\lastskip<1.5em \hskip-\lastskip
      \hskip1.5em plus0em minus0.5em \fi \nobreak
      \vrule height0.75em width0.5em depth0.25em\fi}

\newtheorem{corollary}{Corollary}
\numberwithin{corollary}{section}

\newtheorem{definition}{Definition}
\numberwithin{definition}{section}

\newtheorem{thm}{Theorem}
\numberwithin{thm}{section}

\newtheorem{lemma}[thm]{Lemma}

\newtheorem{fact}[thm]{Fact}

\newtheorem{remark}{Remark}
\numberwithin{remark}{section}

\numberwithin{equation}{section}

\newenvironment{experiment}[1][]%
  {\begin{quote}\ifx#1\empty\else\fbox{#1}\vskip\medskipamount\nopagebreak\fi%
      \begin{enumerate}}%
  {\end{enumerate}\end{quote}}

\newcommand{\adv}{\mathsf{adv}}
\newcommand{\Adv}{\mathsf{adv}}
\newcommand{\success}{\mathsf{succ}}
\newcommand{\hasard}{\ensuremath{\xleftarrow{R}}}

\newcommand\zgame{\mathsf{Game0}}
\newcommand\ogame{\mathsf{Game1}}

\newcommand{\dash}{\mbox{-}}
\newcommand{\comment}[1]{\texttt{\footnotesize$\triangleright\kern3pt$#1}}

\newcommand{\DL}{\textsf{DL}}
\def\return{\mathsf{return}}

\newgray{lightgray}{0.7}

\newgray{gray1}{0.9}

 \newgray{gray2}{0.8}

\newgray{gray3}{0.7}

\newgray{gray4}{0.2}

\begin{document}

\pagestyle{plain}

\title{The Joint Signature and Encryption Revisited\thanks{This paper builds upon the conference papers \cite{ElAimani2008,ElAimani2009b,ElAimani2009a,ElAimani2010,ElAimani2011,ElAimaniSanders2012}.}}

  \author{Laila El Aimani \thanks{LAPSSII - Cadi Ayyad University Morocco}}

  \date{}

  \maketitle

\begin{abstract}
We study the Sign\_then\_Encrypt, Commit\_then\-\_Encrypt\_and\_Sign, and Encrypt\_then\_Sign paradigms in the context of two cryptographic primitives, namely designated confirmer signatures and signcryption. Our study identifies weaknesses in those paradigms which impose the use of expensive encryption (as a building block) in order to meet a reasonable security level. Next, we propose some optimizations which annihilate the found weaknesses and allow consequently cheap encryption without compromising the overall security. Our optimizations further enjoy verifiability, a property profoundly needed in many real-life applications of the studied primitives.

\textbf{Keywords: }Sign\_then\_Encrypt, Commit\_then\_Encrypt\_and\_Sign,  Encrypt\_then\_Sign, (Public) Verifiability, Designated confirmer signatures, Signcryption, Privacy-preserving mechanisms, Zero knowledge proofs

\end{abstract}

\newpage
\tableofcontents
\newpage

\section{Introduction}
\label{sec:intro}

Cryptographic mechanisms that require both the functionalities of signature and of encryption are becoming nowadays increasingly important. In fact, signatures guarantee the integrity/authenticity of the transmitted data, whereas encryption is needed to ensure either the confidentiality of the signed data or the opacity of the signature. In this document, we study two of these primitives, namely designated confirmer signatures and signcryption.
\paragraph{Designated confirmer signatures} An important feature in
digital signatures is the universal verification, i.e. anyone can verify signatures
issued by a signer given his public key. However, such a property can be
undesirable in some applications and needs to be controlled or limited. A typical example is a software vendor willing to embed
signatures in his products such that only paying customers are entitled to
check the authenticity of these signatures. Undeniable signatures, introduced in
\cite{ChaumvanAntwerpen1989}, provide a good solution to this problem as they
are: (1) only verified with the help of the signer, (2) non transferable,
(3) binding in the sense that a signer cannot deny a signature he has actually
issued. The only drawback of these
signatures is that unavailability of the signer obstructs the entire verification process. To overcome this problem, designated confirmer signatures were
introduced in \cite{Chaum1995}, where the confirmation/denial of a signature is
delegated to a \emph{designated confirmer}. With this solution, the signer can
confirm only signatures he has just generated, whilst the confirmer can
confirm/deny any signature. Finally, a desirable property in designated confirmer signatures is the convertibility of the signatures to ordinary ones. Indeed, such a property turned out to play a central role in fair
payment protocols \cite{BoydFoo1998}.

\paragraph{Signcryption} This primitive was introduced by Zheng \cite{Zheng1997} to simultaneously perform the functions of both signature and encryption in a way that is more efficient than signing and encrypting separately. A typical use-case of this mechanism  is secure email where the sender wants to encrypt his email to guarantee privacy, and at the same time, the receiver needs to ensure that the encrypted email comes from the entity that claims to be its provenance. A further requirement on signcryption is verifiability which consists in the possibility to prove efficiently the validity of a given signcryption, or to prove that a signcryption has indeed been produced on a given message. In fact, verifiability is applicable in filtering out spams in a secure email system; the spam filter should be able to verify the authenticity of the ciphertext without knowing the message. Also, the receiver that decrypts the email might be compelled, for instance to resolve some later disputes, to prove that some sender has (not) produced the email; therefore, it would be desirable to support the prover with efficient means to provide such proofs without having to disclose his private input. Although a number of constructions \cite{BaoDeng1998,ShinLeeShim2002,ChowYiuHuiChow2003,Mao2006,SelviVivekRangan2010} have tackled the notion of verifiability (this notion is often referred to in the literature as public verifiability, and it denotes the possibility to release (by the receiver) some information which allows to publicly verify a signcryption with/out revealing the message in question), most of these schemes do not allow the sender to prove the validity of the created signcryption, nor allow the receiver to prove \emph{without revealing any information, ensuring consequently non-transferability}, to a third party, the (in)validity of a signcryption w.r.t. a given message. It is worth noting that the former need, i.e. allowing the sender to prove the validity of a signcryption without revealing the message, solves completely the spam filtering problem without having the receiver disclose anything; a sender needs only to provide a proof of validity of his signcrypted email to ensure that the latter will be marked as a legitimate email. The proof should be ideally ``non transferable'' so that the spam filter cannot replicate it to a third party, ensuring therefore the privacy of the sender.

\subsection{Related work}
Since the introduction of the aforementioned primitives, many realizations (of these primitives) which achieve different levels of security have been proposed. On a high level, security in these primitives involves basically two properties; privacy and unforgeability. The last property is analogous to unforgeability in digital signatures and it denotes the difficulty to impersonate the signer. Privacy in confirmer signatures (signcryptions) is similar to indistinguishability in public key encryption, and it refers to the difficulty to distinguish confirmer signatures (signcryptions) based on the underlying messages. 
Defining formally those two properties is a fundamental divergence in constructions realizing these primitives as there are many issues which come into play. One consequential difference between security models is whether the adversary is external or internal to the system. The former case corresponds to \emph{outsider security}, e.g. \cite{Dent2005}, whereas the latter denotes \emph{insider security} which protects the system protagonists even when some of their fellows are malicious or have compromised/lost their private keys \cite{CamenischMichels2000,GentryMolnarRamzan2005,WangBaekWongBao2007,AnDodisRabin2002,MatsudaMatsuuraSchuldt2009}. It is naturally possible to mix these notions into one single scheme, i.e. insider privacy and outsider unforgeability \cite{AnDodisRabin2002,ChibaMatsudaSchuldtMatsuura2011}, or outsider privacy and insider unforgeability \cite{BaekSteinfeldZheng2007}. However, the most frequent mix is the latter  as illustrated by the number of works in the literature, e.g. \cite{AnDodisRabin2002,JeongJeongRheeLeeLim2002,BaekSteinfeldZheng2007}; it is also justified by the necessity to protect the signer from anyone trying to impersonate him  including entities in the system. Insider privacy is by contrast needed in very limited applications; the typical example, given in \cite{AnDodisRabin2002}, is when the adversary happens to steal the private key of the signer, but we still wish to protect the  privacy of the recorded signcryptions/confirmer signatures sent by the genuine signer.

Building complicated systems upon simple and basic primitives is customary in
cryptography as it allows to re-use existing work about the primitives, and
it achieves easy-to-understand and easy-to-prove systems. The classical constructions used to build the above mentioned primitives are:

\paragraph{Sign\_then\_Encrypt (StE)} For confirmer signatures, this technique consists in first signing the message, then encrypting the produced signature. The construction was first formally \footnote{The idea without proof was already known, for instance, it was mentioned in \cite{DamgardPedersen1996}.} described in \cite{CamenischMichels2000}, and it suffered the resort to concurrent zero knowledge (ZK) protocols of general NP statements in the confirmation/denial protocol (i.e. proving knowledge of the decryption of a ciphertext, and that this decryption forms a valid signature on the given message). Later, the proposal in \cite{GoldwasserWaisbard2004} circumvented this problem by encrypting the digital signature during the confirmation protocol. With this trick, the authors managed to get rid of concurrent ZK proofs of general NP statements in the confirmation protocol (the denial protocol still suffers the recourse to such
proofs), but at the expense of the security and the length of the resulting signatures. Another construction implementing this principle is given in \cite{Wikstroem2007}; it uses cryptosystems with labels and is analyzed in a more elaborate security model. However, it is supplied with only one efficient instantiation as the confirmation/denial protocols still resort to  concurrent ZK protocols of general NP statements.

For signcryption, this technique consists in similarly signing the message to be signcrypted, however the signcryption corresponds to the encryption of the produced digital signature \emph{in addition to the message}. The construction was first described and analyzed in \cite{AnDodisRabin2002}. It was further extended in \cite{MatsudaMatsuuraSchuldt2009} to support the multi-user setting, i.e. a setting where many senders interact with many receivers, using tag-based encryption. Finally,  there are the recent constructions \cite{ChibaMatsudaSchuldtMatsuura2011} which achieve multi-user insider security using (tag-based) encryption schemes from the hybrid encryption paradigm. It is worth noting that none of these constructions treat verifiability.


\paragraph{Commit\_then\_Encrypt\_and\_Sign (CtEaS)} This technique was first described in the context of signcryption in \cite{AnDodisRabin2002}. It has been essentially introduced to parallel encryption and signature. In fact, signcryption of a message using this technique is obtained by committing to the message, then encrypting the message and the randomness used to form the commitment, \emph{and} signing the commitment. Later, (a variant of) this technique was adopted in  \cite{GentryMolnarRamzan2005} for confirmer signatures, i.e. encryption is performed only on the randomness used to form the commitment, and signature is obtained on this encryption \emph{in addition} to the commitment. This construction was identified to be flawed in \cite{WangBaekWongBao2007}, where the authors propose to use encryption with labels as building blocks in order to repair the flaw. More precisely, a confirmer signature on a message $m$ is obtained by first committing to $m$, then encrypting the randomness used in the commitment w.r.t. the label $m\|\pk$, $\pk$ being the public key of the used signature scheme, \emph{and} finally signing the commitment. Although CtEaS can be used with  \emph{any} signature scheme (StE needs to be used with special signature schemes in order to allow an efficient verifiability), it is still afflicted with the recourse to general ZK proofs, e.g. proving in concurrent ZK the knowledge of the decryption of an IND-CCA encryption that equals a string used for commitment. An efficient instantiation is however achieved for confirmer signatures in \cite{WangBaekWongBao2007} using Camenisch-Shoup's verifiable encryption scheme \cite{CamenischShoup2003} and Pedersen's commitment scheme.

\paragraph{Encrypt\_then\_Sign (EtS)} This technique consists in first encrypting the message, then producing a signature on this encryption. EtS has been introduced for two-user setting signcryption in \cite{AnDodisRabin2002}. It has been later extended in \cite{MatsudaMatsuuraSchuldt2009} to support the multi-user setting  using tag-based encryption. Besides, in  the same work \cite{MatsudaMatsuuraSchuldt2009}, the authors present variations of the paradigm using symmetric primitives and achieve efficient signcryption schemes but the expense of security (outsider unforgeability/privacy) and verifiability.

To summarize the state of the art, StE, CtEaS, and EtS have been studied for the aforementioned primitives in different security models. These studies conclude the need for CCA secure encryption in order to ensure insider privacy. Since \emph{outsider security might be all one needs} for privacy as quoted by the authors in \cite{AnDodisRabin2002}, we propose to  relax the requirement on insider privacy with the hope of weakening the strong assumption (CCA security) on the encryption. The work \cite{MatsudaMatsuuraSchuldt2009} achieves some results in this direction as it rests on CPA secure \emph{symmetric} encryption, but at the expense of verifiability.

It would be nice to study these paradigms in the outsider privacy model, and provide efficient variants which rest on cheap encryption while providing good verifiability properties. This is the main contribution of this paper.
 
\subsection{Contributions and overview of our techniques} 
As stated earlier, the main contribution of the present paper is a thorough study of StE, CtEaS, and EtS in the outsider privacy model. Our study concludes that both StE and CtEaS  require expensive assumptions on the underlying encryption (PCA security) in order to derive signcryption or confirmer signatures with outsider privacy. We do this by first proving the insufficiency of OW-CCA and NM-CPA secure encryption using the celebrated \emph{meta-reduction} tool, then by exhibiting a simple attack if the system is instantiated from certain encryption schemes. These negative results can be explained by an inherent weakness in these constructions that consists in the possibility of creating confirmer signatures or signcryption without the help of the signer. 

Next, we propose ameliorations of the paradigms that annihilate this weakness without compromising the security. We achieve this by binding the digital signature to the resulting signcryption/confirmer signature. Consequently, our optimizations of StE and CtEaS (for both signcryption and confirmer signatures) rest on cheap encryption (CPA secure asymmetric encryption) and support efficiently the verifiability property required for such mechanisms. We actually describe explicitly, and for the first time, the verifiability proofs in case the constructions are instantiated from large classes of encryption, commitment, and signature schemes. 

We have further the following side-results:

\begin{enumerate}
\item Our negative results for StE serve also for providing evidence that a well known undeniable signature \cite{DamgardPedersen1996} is unlikely to provide its conjectured privacy. Moreover, the adjustment we propose to the basic StE paradigm fixes also this scheme  (\cite{DamgardPedersen1996}), and captures further undeniable signatures that were proposed later \cite{LeTrieuKurosawaOgata2009b,SchuldtMatsuura2010}.
\item  We provide practical instantiations of EtS, in the context of both signcryption and confirmer signatures, which efficiently support verifiability. In fact, some of the required verifiability proofs involve non-interactive proofs of correctness of a decryption. We identify several encryption schemes that efficiently implement this feature.

\item We propose a new paradigm for signcryption, Encrypt\_then\_Sign\_then\_Encrypt, which allows efficient verifiability while proffering full outsider privacy (i.e. anonymity of the sender and indistinguishability of the signcryption).

\item Finally, our constructions (of both confirmer signatures and signcryption) can achieve insider privacy while conserving their good verifiability properties if we substitute the required ``normal'' encryption by tag-based encryption combined with secure one-time signatures.

\end{enumerate}

\section{Preliminaries}
\subsection{Cryptographic primitives}
\emph{Notation:} Throughout the text, we will use a dot notation to refer the different components; for instance, $\Gamma.\encrypt()$ refers to the encryption algorithm of public key encryption scheme $\Gamma$, $\Sigma.\pk$ to the public key of signature scheme $\Sigma$, etc.
\subsubsection{Digital signatures}A signature scheme comprises three algorithms, namely the key
generation algorithm $\keygen$, the signing algorithm $\sign$, and the
verification algorithm $\verify$. The standard security notion for a signature scheme is existential
  unforgeability under chosen message attacks (EUF-CMA), which was introduced in
  \cite{GoldwasserMicaliRivest1988}. Informally, this notion refers to the
  hardness of, given a signing oracle, producing a valid pair of message and
  corresponding signature such that the message has not been queried to the signing
  oracle. There exists also  the stronger notion, SEUF-CMA (strong existential unforgeability under chosen message attack),
  which allows the adversary to produce a forgery on a previously queried message,
  however the corresponding signature must not be obtained from the signing oracle.

  \noindent A signature is $(t,\epsilon,q_s)$-(S)EUF-CMA secure, if no adversary, operating in time $t$ and issuing $q_s$ queries to the signing oracle, produces a pair of a (new) message and a valid corresponding signature that was not obtained from the signing oracle, with probability greater than $\epsilon$; the probability is taken over all the random coins.

  \subsubsection{Public key encryption}A public key encryption (PKE) scheme consists of the key generation algorithm
$\keygen$, the encryption algorithm $\encrypt$ and the decryption algorithm
$\decrypt$. The typical \emph{security goals} a PKE scheme should attain
are: one-wayness (OW) which corresponds to the difficulty of inverting a ciphertext, indistinguishability (IND) which refers to the hardness of distinguishing ciphertexts based on the messages they encrypt, and
finally non-malleability (NM) which corresponds to the hardness of deriving from
  a given ciphertext another ciphertext such that the underlying plaintexts
  are meaningfully related. Conversely, the typical \emph{attack models} an adversary against an
encryption scheme is allowed to are: Chosen Plaintext Attack (CPA) where the adversary can encrypt any
  message of his choice, Plaintext
  Checking Attack (PCA) in which the adversary is allowed to
  query an oracle on pairs ($m,c$) and gets answers whether $c$ encrypts $m$ or not, and finally Chosen Ciphertext Attack (CCA) where the adversary is allowed to query a decryption oracle. Pairing the mentioned goals with these attack models yields nine
\emph{security notions}: goal-atk for $\GOAL \in \{ \OW,\IND,\NM \}$ and $ \ATK
\in \{\CPA,\PCA,\CCA\}$. We refer to \cite{BellareDesaiPointchevalRogaway1998}
for the formal definitions of these notions as well as for the relations they satisfy.

\noindent Later in the text, we will need further the INV-CPA notion, i.e. invisibility under a chosen plaintext attack, which denotes the difficulty to distinguish ciphertexts on an adversarially chosen message from random elements in the ciphertext space (public-key variant of the INV-OT notion defined later for Data Encapsulation Mechanisms, i.e. DEMs).

\noindent Similarly, an encryption scheme is $(t,\epsilon,q)$-goal-atk secure, if no adversary operating in time $t$ and issuing $q$ queries to the allowed oracles, succeeds in the game defined by the security notion goal-atk with probability greater than $\epsilon$; the probability is again over all the random coins.


\subsubsection{Key/Data encapsulation mechanisms (KEM/DEMs)}
A KEM comprises three algorithms: (1) the key generation algorithm $\keygen$ which probabilistically  generates a key pair $(\sk,\pk)$, (2) the
encapsulation algorithm $\encap$ which inputs the public key $\pk$ and probabilistically generates a \emph{session key} denoted
  $k$ and its \emph{encapsulation} $c$, (3) and finally the decapsulation algorithm $\decap$ which inputs the private key
  $\sk$ and the element $c$ and computes the decapsulation $k$ of
  $c$, or returns $\perp$ if $c$ is invalid.

\noindent  The typical security goals that a KEM should satisfy are similar to those
defined for encryption schemes. Similarly, when conjoined with the three attack models
CPA, PCA and CCA, they yield nine security notions whose definitions follow
word-for-word from the definitions of the encryption schemes
notions. 

\noindent  Let $\kappa$ be a security parameter. We recall below the formal definition of an $\IND\mbox{-}\CPA$ experiment, conducted by an adversary $\cal A$ against a KEM $\kem$ ($\sf{K}$ denotes the keys space in the experiment below).

{
\begin{center}
\begin{tabular}{l}
\fbox{Experiment $\mathbf{Exp}_{\ensuremath{\cal K},{\mathcal
A}}^{\textsf{ind-cpa-b}}(1^\kappa)$}\\
$(\pk,\sk) \leftarrow \ensuremath{{\cal K}}.\keygen(1^\kappa)$, \\

$\mathcal{I} \leftarrow {\mathcal A}(\pk)$ \\
$(c^\star,k^\star) \leftarrow \kem.\encap_{\pk}()$\\
$\mathsf{~if~} b=0~ \mathsf{then~~} \{\sf{k \hasard K , k^\star \leftarrow k}\}$\\

$d\leftarrow {\mathcal A}(\mathcal{I},c^\star,k^\star)$~\\
Return $d$ \\
\end{tabular}
\end{center}}

\noindent A KEM is $(t,\epsilon)$-$\IND\mbox{-}\CPA$-secure if the advantage defined by 
$$
\mathbf{Adv}_{\ensuremath{{\cal K}},{\mathcal A}}^{\mathsf{ind-cpa}}(\kappa) = \left\vert \Pr\left[\mathbf{Exp}_{\ensuremath{\kem},{\mathcal
A}}^{\mathsf{ind-cpa-b}}(\kappa)=b\right] - \frac{1}{2}\right\vert,
$$

of any adversary $\cal A$, operating in time $t$, in the above game,  is no greater than $\epsilon$. The probability is taken over all the random coins.

A DEM is a secret key encryption scheme given by the same algorithms
forming a public key encryption scheme that are: (1) the key generation algorithm $\keygen$ which produces uniformly distributed keys $k$ on input a given security parameter, (2) the encryption algorithm $\encrypt$ which inputs a key $k$ and a message $m$ and produces a ciphertext $c$, and (3) the decryption algorithm which decrypts ciphertext $c$ using the same key $k$ (used for encryption) to get back the message $m$ or the special rejection symbol $\perp$. 

\noindent We define in the following a security notion for DEMs that we will need later in the text; it is called invisibility under a one-time attack $\INV\mbox{-}\OT$, and it denotes the difficulty to distinguish the encryption of an adversarially chosen message from a random ciphertext ($\sf{C}$ denotes in the experiment below the ciphertext space).

{
\begin{center}
\begin{tabular}{l}
\fbox{Experiment $\mathbf{Exp}_{\ensuremath{{\cal D}},{\mathcal
A}}^{\textsf{inv-ot}-b}(1^\kappa)$}\\
$ k \leftarrow \ensuremath{{\cal D}}.\keygen(1^\kappa)$, \\

$(m^\star,\mathcal{I}) \leftarrow {\mathcal A}(1^\kappa)$ \\

$e^\star \leftarrow \ensuremath{{\cal D}}.\encrypt_{k}(m^\star)$\\
$\mathsf{~if~} b=0~ \mathsf{then~~} \{e \hasard \C, e^\star \leftarrow e\}$\\

$d\leftarrow {\mathcal A}(\mathcal{I},e^\star)$~\\

Return $d$ 
\end{tabular}

\end{center}}
\noindent We define:

$$
\mathbf{Adv}_{\ensuremath{{\cal D}},{\mathcal A}}^{\mathsf{inv-ot}}(1^\kappa) =
\left\vert \Pr\left[\mathbf{Exp}_{\ensuremath{{\cal D}},{\mathcal
A}}^{\mathsf{inv-ot-b}}(1^\kappa)=b\right] - \frac{1}{2}\right\vert,
$$

\noindent A DEM is said to be  $(t,\epsilon)$-$\INV\mbox{-}\OT$ secure if the advantage $\mathbf{Adv}_{\ensuremath{{\cal D}},{\mathcal A}}^{\mathsf{inv-ot}}(\kappa)$ of any adversary $\cal A$, operating in time $t$ and running the experiment above, is no greater than $\epsilon$. The probability is taken over all the random coins.

\noindent \emph{A DEM with injective encryption} is a DEM  where, for a every fixed key, the encryption algorithm $\encrypt$, seen a function of the message, is injective. I.e. for a fixed key, for every message $m$, there exists only one valid ciphertext that decrypts to $m$. Note that such DEMs exist, e.g. the one-time pad, and proffer interesting security properties, e.g. $\INV\mbox{-}\OT$.

Finally, KEMs can be efficiently  combined with DEMs to build secure public key encryption schemes. This technique is called the hybrid encryption paradigm and we refer to \cite{HerranzHofheinzKiltz2006}  for the necessary and sufficient conditions
on the KEMs and the DEMs to obtain a certain security level for the
resulting hybrid encryption scheme.

\subsubsection{Commitment schemes}
A commitment scheme \cite{BrassardChaumCrepeau1988} consists of (1) a key generation algorithm $\keygen$, (2) a commitment algorithm $\commit$, (3) and an opening algorithm $\open$. We require in a commitment scheme the \emph{hiding} and \emph{binding} properties. The former informally denotes the difficulty to infer information about the message from the corresponding commitment, whereas the latter denotes the difficulty to come up with collisions, i.e. find two different messages that map to the same value by the $\commit$ algorithm. A further property might be required, in some applications, for commitments, namely \emph{injectivity}; It denotes that $\commit$ for a fixed message (viewed as a function of the opening value) is injective: two different opening values lead two different commitments.

\subsection{(Non-) Interactive Proofs}
An interactive proof, first introduced in \cite{GoldwasserMicaliRackoff1989}, informally consists of a prover $P$ trying to convince a verifier $V$ that an instance $x$ belongs to a language $L$. $x$ refers to the common input whereas $(P,V)(x)$ denotes the proof instance carried between $P$ and $V$ at the end of which $V$ is (not) convinced with the membership of the alleged instance $x$ to $L$: $(P,V)(x) \in \{\mathsf{Accept, Reject}\}$. $P$ is modeled by a
probabilistic Turing machine whereas $V$ is modeled by a \emph{polynomial}
probabilistic Turing machine. During $(P,V)(x)$, the parties
exchange a sequence of messages called the proof transcript.

\noindent An interactive proof should satisfy  \textbf{completeness} which denotes the property of successfully running the protocol if both parties are honest. A further property required in proof systems is \textbf{soundness} which captures the inability of a cheating prover
$P$ to convince the verifier $V$ with an invalid statement.

\paragraph{(Concurrent) ZK proofs}Let $(P,V)$ be an interactive proof system for some language $L$. We say that $(P,V)$ is \textbf{zero knowledge} if for every $x \in L$, the proof transcript $(P,V)(x)$ can be produced by an efficient algorithm $S$, with no access to the prover, with indistinguishable probability distributions from the real interaction with the genuine prover. A proof is said to provide concurrent ZK if it remains ZK when the prover interacts concurrently with many verifiers (that might potentially collude) on many instances of the proof. It was shown in \cite{DworkNaorSahai2004} that every NP language accepts a concurrent ZK proof system.

\paragraph{Proofs of knowledge} A proof of knowledge is an interactive proof in which the prover succeeds 'convincing' a verifier that he \emph{knows} something. In addition to the \textbf{completeness} property (i.e. property of successfully running the protocol if both parties are honest), a proof of knowledge must further satisfy the \textbf{validity} or \textbf{soundness} property. I.e. let $R$ be the NP-relation for an NP-language $L$:
$$L=\{x\colon \exists~ ~w~ \mathrm{such ~that}~ R(x,w)~ \mathrm{holds}\}$$

\noindent Validity of a  proof of knowledge for $R$ captures the intuition that from any (possibly cheating) prover $\tilde{P}$ that is able to convince the verifier with good enough probability on a statement $x \in L$, there exists an efficient \emph{knowledge extractor} capable of extracting a valid witness for $x$ from $\tilde{P}$ with non negligible probability. This guarantees that no prover that doesn't know the witness can succeed in convincing the verifier.
\noindent Finally, \textbf{zero knowledgeness} of a proof of of knowledge captures the possibility to prove knowledge of the given witness without revealing it. This property is defined, as in interactive proofs, using an efficient simulator, with no access to the prover, capable of producing a proof transcript indistinguishable from the interaction between the genuine prover and the cheating verifier.

\paragraph{$\Sigma$ protocols}A \emph{public-coin protocol} is an interactive proof in which the verifier chooses all
its messages randomly from publicly known sets. A
\emph{three-move protocol} can be written in a canonical form in
which the messages exchanged in the three moves are often called commitment, challenge, and response respectively. The protocol is said to have
the \emph{honest-verifier zero-knowledge property (HVZK)} if there exists an
algorithm that is able, provided
the verifier behaves as prescribed by the protocol, to produce, without the
knowledge of the secret, transcripts that are indistinguishable from those  of the real protocol. The protocol
is said to have the \emph{special soundness property (SpS) property} if there exists an algorithm that is able to extract the secret from two
accepting transcripts of the protocol with the same commitment and different
challenges. Finally, a three-move public-coin protocol with HVZK and
SpS properties is called a \emph{$\Sigma$ protocol}. 
 
\paragraph{Non-interactive proofs} were introduced in \cite{BlumFeldmanMicali1988}. They consist of three entities: a prover, a verifier, and a uniformly selected \emph{common reference string - crs-} (which can be thought of as being selected by a trusted third party). Both verifier and prover can read the reference string. The interaction consists of a single message sent from the prover to the verifier, who is left with the final decision. The zero-knowledge requirement refers to a simulator that outputs pairs that should be indistinguishable from the pairs (crs, prover's message).

\subsection{Cryptographic reductions}
\label{subsec:reductions}

A reduction in cryptology, often denoted $\R$, is informally an algorithm solving some problem given access to an adversary $\A$ against some cryptosystem. To be able to use $\A$, the reduction must simulate $\A$'s environment (instance generation, queries if any...) in a way that is (almost) indistinguishable from the real model. Both the reduction and the adversary are considered probabilistic Turing machines. The advantage of the reduction $\adv(\R)$ is by definition the success probability in solving the given instance of the problem, where the probability is taken over the instance generation and the random coins of both $\R$ and $\A$. Similarly, the advantage of the adversary $\A$, denoted by $\adv(\A)$ refers to the success probability (taken over all the coin tosses) in breaking the cryptosystem.



\paragraph{Key-preserving reductions} These reductions refer to a wide and popular class of reductions which supply the adversary with the same public key as its challenge key. In this text, we restrict this notion to a smaller class of reductions.
\begin{definition}[Key-preserving reductions]
\label{def:key-preserving-reductions}
 Let $\cal A$ be an adversary which solves a problem A that is perfectly reducible to OW-CPA breaking some public key encryption scheme $\Gamma$. Let further $\cal R$ be a reduction breaking some security notion of $\Gamma$ w.r.t. a public key $\pk$ given access to $\cal A$. $\cal R$ is said to be \emph{key-preserving} if it launches $\cal A$ over her own challenge key $\pk$ in addition to some other parameters (chosen freely by her) according to the specification of $\A$.
\end{definition}

\noindent Such reductions were for instance used in \cite{PaillierVillar2006} to prove a separation
between factoring and IND-CCA-breaking some factoring-based encryption schemes in the
standard model.


\section{Convertible Designated Confirmer Signatures (CDCS)}

\subsection{Syntax}
\label{subsec:syntax+model}
A convertible designated confirmer signature (CDCS) scheme consists of the following procedures:
\begin{description}

\item[$\setup(1^\kappa)$] On input a security parameter $\kappa$, probabilistically generate the public parameters $\param$ of
  the scheme. Although not always explicitly mentioned, $\param$ serves
  as an input to all the algorithms/protocols that follow.

\item[$\keygen_{\entfont{E}}(1^\kappa,\param)$] This probabilistic algorithm
  outputs the key pair $(\pk_{\entfont{E}},\sk_{\entfont{E}})$ for the
  entity $\entfont{E}$ in the system; $\entfont{E}$ can either be the
  signer $\entfont{S}$ who issues the confirmer signatures, or the confirmer $\entfont{C}$ who confirms/denies the signatures. 

\item[$\sign_{\sks}(m,\pkc)$] On input $\sks$, $\pkc$ and a message $m$, this probabilistic algorithm outputs a confirmer
  signature $\mu$ on $m$.

\item[$\verify_{\{\coins~\vee~\skc\}}(\mu,m,\pks,\pkc)$] This is an algorithm, run by the signer
  on a \emph{just generated} signature or by the confirmer on \emph{any}
  signature. The input to the algorithm is: the alleged signature $\mu$, the message $m$, $\pks$, $\pkc$, the coins $\coins$ used to produce the signature if the
  algorithm is run by the signer, and $\skc$ if it is run by the confirmer. The output is either
  $1$ if the signature if valid, or $0$ otherwise.

\item[$\sconfirm_{\langle \signer(\coins_\mu),\V \rangle}(\mu,m,\pks,\pkc)$]This is an interactive protocol where the signer $\entfont{S}$ convinces a verifier $\V$ of the validity of a signature he has just generated. The common input comprises the signature and the message in question, in addition to $\pks$ and $\pkc$. The private input of $\entfont{S}$ consists of the random coins used to produce the signature $\mu$ on $m$.

\item[$\confirm\slash\deny_{\langle \C(\skc),\V \rangle}(\mu,m,\pks,\pkc)$]These are interactive protocols
  between the confirmer $\C$ and a verifier $\V$. Their common input consists of $\pks$, $\pkc$, the alleged
  signature $\mu$, and the message $m$. The
  confirmer uses $\skc$ to convince the
  verifier of the validity/invalidity of the signature $\mu$ on $m$. At the
  end, the verifier accepts or rejects the proof.

\item[$\convert_{\skc}(\mu,m,\pks,\pkc)$]This is an algorithm run by the confirmer $\C$
  using $\skc$, in addition to $\pkc$ and $\pks$, on a potential confirmer signature $\mu$ and some message $m$. The result is either
  $\perp$ if $\mu$ is not a valid confirmer signature on $m$, or a string $\sigma$ which is a valid digital signature on $m$ w.r.t. $\pks$.

\item[$\verifyconverted(\sigma,m,\pks)$]This is an algorithm for verifying
  converted signatures. It inputs the converted signature $\sigma$, the message $m$ and $\pks$ and outputs either $0$ or $1$.

\end{description}

\begin{remark}[Notation]
For the sake of simplicity, the public/private keys as well as the private coins will be often omitted from the description of the above algorithms/protocols. Therefore, whenever the context is clear, $\sign$, $\sconfirm$, $\{\confirm,\deny\}$, $\convert$, and $\verifyconverted$ will only involve the message and the corresponding confirmer/converted signature.
\end{remark}

\begin{remark}
 In \cite{GentryMolnarRamzan2005,WangBaekWongBao2007}, the authors give the possibility of obtaining \emph{directly} digital signatures on any given message. We find this unnecessary since it is already enough that a CDCS scheme supports the convertibility feature.
\end{remark}

\subsection{Security model}

A CDCS scheme should meet the following properties:

\begin{description}
\item \emph{Completeness. }Every signature produced by $\sign$ should be validated by the algorithm $\verify$ and correctly converted. Moreover, valid signatures should be correctly confirmed by $\sconfirm$ and $\confirm$, and invalid signatures should be correctly denied by $\deny$ if the entities $\{\entfont{S},\C\}$ follow honestly the protocols.

  \noindent More formally, let $(\pks,\sks)$ and $(\pkc,\skc)$ be the signer's and confirmer's key pairs resp. of a CDCS scheme $\cs$. Let further $m$ be a message from the message space of $\cs$. We consider the following experiment
  \vspace{1 cm}
\begin{center}
\footnotesize
\begin{experiment}[Experiment $\mathbf{Exp}_{\cs}^{\mathsf{completeness}}( m, \pks,\pkc)$]
\item $\mu \gets \cs.\sign_{\sks}(m,\pkc)$;
\item  $\psi \hasard CS.\sf{space}$:
  
  $ ~~~~~~~~~~~~~~~~~~~ \cs.\verify_{\skc}(\psi,m,\pkc,\pks) = 0$;
\item $\out_0 \gets \cs.\verify_{\{\coins_\mu ~\vee~ \skc\}}(\mu,m,\pkc,\pks)$;
\item $\langle \done \mid \out_1 \rangle \gets \cs.\sconfirm_{\langle \signer(\coins_\mu),\V \rangle}(\mu,m,\pks,\pkc)$; 
\item $\langle \done \mid \out_2 \rangle \gets  \cs.\confirm_{\langle \C(\skc),\V \rangle}(\mu,m,\pks,\pkc)$; 
\item $\langle \done \mid \out_3 \rangle  \gets \cs.\deny_{\langle \C(\skc),\V \rangle }(\psi,m,\pks,\pkc)$; 
\item $\sigma \gets \cs.\convert_{\skc}(\mu,m)$;
 \item $\out_4 \gets \cs.\verifyconverted(\sigma,m)$;
\item Return $out_0 \wedge out_1 \wedge out_2 \wedge out_3 \wedge \out_4$.
\end{experiment}
\end{center}
 The scheme $\cs$ complete if, for all signer's and confirmer's key pairs $(\pks,\sks)$  and $(\pkc,\skc)$ resp. , for all messages $m$, the outcome of Experiment $\mathbf{Exp}_\cs^{\mathsf{completeness}}(m,\pks,\pkc)$ is $1$ with high probability, where the probability is taken over all the random choices.

\item \emph{Security for the verifier (soundness). }This property informally means that an
  adversary who compromises the private keys of both the signer and the
  confirmer cannot convince the verifier of the validity (invalidity) of an
  invalid (a valid) confirmer signature. 

 A CDCS scheme $\cs$ is sound if the success probability of any polynomial-time adversary $\cal A$ (returning $1$) in  Experiment $\mathbf{Exp}_{\cs,{\cal A}}^{\mathsf{soundness}}(1^\kappa)$ is negligible; the probability is taken over all the random tosses.

\begin{center}
\small
\begin{experiment}[Experiment $\mathbf{Exp}_{\cs,\A}^{\mathsf{soundness}}(1^\kappa)$]
\item $\param \gets \setup(1^\kappa)$; 
 \item $(\pks,\sks) \gets \cs.\keygen_{\entfont{S}}(1^\kappa)$; $(\pkc,\skc) \gets \cs.\keygen_\C(1^\kappa)$;
\item $(m,\psi) \gets \A(\sks,\skc,\coins_\psi)$:

  $~~~~~~~~~~~~~~~~~~~\cs.\verify_{\{\coins_\psi ~\vee~\skc\}}(\psi,m,\pkc,\pks)=0$; 
\item $(m,\mu) \gets \A(\sks,\skc,\coins_\mu)$:

  $~~~~~~~~~~~~~~~~~~~\cs.\verify_{\{\coins_\mu ~\vee~\skc\}}(\mu,m,\pks,\pkc)=1$; 
\item $\langle \done \mid \out_1 \rangle  \gets  \cs.\sconfirm_{\langle \A(\sks,\skc,\coins_\psi),\V \rangle}(\psi,m,\pks,\pkc)$; 
\item $\langle \done \mid \out_2 \rangle  \gets \cs.\confirm_{ \langle \A(\sks,\skc,\coins_\psi),\V \rangle}(\psi,m,\pks,\pkc)$; 
\item $\langle \done \mid \out_3 \rangle  \gets \deny_{\langle \A(\sks,\skc,\coins_\mu),\V \rangle}(\mu,m,\pks,\pkc)$; 
\item Return $\out_1 \vee \out_2 \vee \out_3$.

\end{experiment}
\end{center}

\item  \emph{Non-transferability. }This property  captures the simulatability of $\sconfirm$, $\confirm$, and $\deny$. It is defined through the following games which involve the adversary, the signer and the confirmer of the CDCS scheme $\cs$, and a simulator (Experiment  $\mathbf{Exp}_{\cs}^{\mathsf{non-transferability}}(1^\kappa)$ ):

\begin{description}
\item \textbf{Game 1:} the adversary $\cal A$ is given the public keys of the signer and
  of the confirmer, namely $\pks$ and $\pkc$ resp. He can then make arbitrary
  queries of type $\{\sign,\sconfirm\}$ to the signer and of type
  $\{\confirm,\deny\}$ and $\convert$ to the confirmer. Eventually, the adversary presents two
  strings $m$ and $\mu$ for which he wishes to carry out, on the common input
  $(m,\mu,\pks,\pkc)$, the protocol $\sconfirm$ with the signer (if $\mu$ has been just generated by the signer on $m$), or
  the protocols $\{\confirm,\deny\}$ with the confirmer. The
  private input of the signer is the randomness used to generate the signature
  $\mu$ (in case $\mu$ is a signature just generated by the signer), whereas the private
  input of the confirmer is his private key $\skc$. The adversary continues
  issuing queries to both the signer and the confirmer until he decides that
  this phase is over and produces an output.
\item \textbf{Game 2:} this game is similar to the previous one with the
  difference of playing a simulator instead of running the real signer or the
  real confirmer when it comes to the interaction of the adversary with the
  signer in $\sconfirm$ or with the confirmer in
  $\{\confirm,\deny\}$ on the common input $(\mu,m,\pks,\pkc)$. The simulator is not given the
  private input of neither the signer nor the confirmer. It is however allowed
  to issue a single oracle call that tells whether $\mu$ is a valid confirmer
  signature on $m$ w.r.t. $\pks$ and $\pkc$. Note that the simulator in this game refers to a probabilistic polynomial-time Turing machine with rewind.
\end{description} 

\begin{center}
\small
\begin{experiment}[Experiment $\mathbf{Exp}_\cs^{\mathsf{non-transferability}}(1^\kappa)$]
\item $\param \gets \cs.\setup(1^\kappa)$; 
 \item $(\pks,\sks) \gets \cs.\keygen_{\entfont{S}}(1^\kappa)$; $(\pkc,\skc) \gets \cs.\keygen_\C(1^\kappa)$; 

\item $(m,\mu) \gets \A^{\mathfrak{S},\mathfrak{Cv},\mathfrak{V}}(\pks,\pkc)$\\
\phantom{$(m,\mu) \leftarrow$} $\left\vert
\begin{array}{l} 
\mathfrak{S} : m_i  \longmapsto \cs.\sign_{\sks}(m_i,\pkc) \\ 
\mathfrak{Cv}: (\mu_i,m_i) \longmapsto \cs.\convert_{\skc}(\mu_i,m_i,\pks,\pkc) \\
\mathfrak{V} : (\mu_i,m_i) \longmapsto \cs.\{\sconfirm,\confirm,\deny\}(\mu_i,m_i,\pks,\pkc) \\
\end{array} \right.$
\item $b \hasard \{0,1\}$;\\
$\algofont{if}~~ b=1~~ \algofont{then}~~ \langle \done \mid \out_0 \rangle  \gets \algofont{prove}_{\langle \prover(\skp),\A \rangle}(\mu,m,\pks,\pkc)$;\\
$\algofont{if}~~ b=0~~ \algofont{then}~~\langle \done \mid \out_1 \rangle  \gets \algofont{prove}_{\langle \Sim,\A \rangle}(\mu,m,\pks,\pkc)$;\\
\phantom{$\algofont{if}~~ b=0~~  \algofont{then}~~ \langle \done \mid $} 
$\left\vert
\begin{array}{l} 
 \sf{if} ~~ \algofont{prove}=\cs.\sconfirm~~  \algofont{then}~~ \prover=\signer ~~\algofont{and}~~ \skp=\coins_\mu   \\ 
 \sf{if}~~\algofont{prove} \in \cs.\{\confirm,\deny\} ~~\algofont{then}~~ \prover=\C ~~\algofont{and}~~ \skp=\skc  \\ 
\end{array} \right.$

\item $b^\star \gets \A^{\mathfrak{S},\mathfrak{Cv},\mathfrak{V}}(\mu,m,\pks,\pkc)$
\item Return $(b=b^\star)$.
\end{experiment}
\end{center}
The confirmer signatures are said to be non-transferable if there exists an efficient simulator such that for all
$(\pks,\pkc)$, the outputs of the adversary in  \textbf{Game 1} and \textbf{Game 2} are indistinguishable. In other words, the adversary should not be able to tell whether he is playing \textbf{Game 1} or \textbf{Game 2}. Note that this definition achieves only the so-called \emph{offline non-transferability}, i.e. the adversary is not supposed to interact concurrently with the prover and an unexpected verifier. We refer to Remark \ref{rmq:securityModel} for the details.

\item \emph{Unforgeability. }It is defined through
the game depicted in Experiment $\mathbf{Exp}_{\cs,\A}^{\mathsf{EUF-CMA}}(1^\kappa)$: the adversary $\A$ gets the signer's public key $\pks$ of a CDCS scheme $\cs$ , and generates the confirmer's key pair ($\skc$,$\pkc$). $\A$ is further allowed to query the signer on
  polynomially many messages, say $q_s$. At the end, $\A$ outputs a pair consisting of a message $m^\star$, that has not been queried yet, and a string $\mu^\star$. $\cal A$ wins the game if $\mu^\star$ is a valid confirmer signature on $m^\star$. 

\begin{center}
\begin{experiment}[Experiment $\mathbf{Exp}_{\cs,\A}^{\mathsf{EUF-CMA}}(1^\kappa)$]
\item $\param \gets \setup(1^\kappa)$; 
 \item $(\pks,\sks) \gets \cs.\keygen_{\entfont{S}}(1^\kappa)$;
\item $(\pkc,\skc) \gets \A(\pks)$;
\item $(m^{\star},\mu^{\star}) \leftarrow \mathcal{A}^{\mathfrak{S}}(\pks,\pkc,\skc)$ \\
\phantom{$(m^{\star},\mu^{\star}) \leftarrow$} $
  \begin{array}{l} 
\mathfrak{S} : m \longmapsto \cs.\sign_{\sks}(m,\pkc) \\ 
\end{array}$ \\
\item \textsf{return} 1 \textsf{if and only if:} \\
\hspace{3mm} - $\verify(\mu^\star,m^\star,\pks,\pkc) = 1$ \\
\hspace{3mm} - $m^\star$ was not queried to $\mathfrak{S}$
\end{experiment}
\end{center}

We say that a CDCS scheme $\cs$ is $(t,\epsilon,q_s)$-EUF-CMA secure if there is no adversary, operating in time
$t$, that wins the above game with probability greater than $\epsilon$, where
the probability is taken over all the random choices.  

\item \emph{Security for the confirmer (invisibility). }Invisibility against a chosen message attack (INV-CMA) is defined through the game between an attacker $\A$ and her challenger $\cal C$ (Experiment $\mathbf{Exp}_{\cs,\A}^{\mathsf{INV-CMA}}(1^\kappa)$): after
$\cal A$ gets the public parameters of the CDCS scheme $\cs$ from $\cal C$, she
starts \textbf{Phase 1} where she queries the $\sign$, $\sconfirm$, $\confirm$, $\deny$, and $\convert$  oracles in an adaptive way. Once $\cal A$ decides that
\textbf{Phase 1} is over, she outputs two messages $m_0^\star, m_1^\star$ as challenge messages. $\cal C$ picks uniformly at random a bit $b \in \{0,1\}$. Then
$\mu^\star$ is generated using the signing oracle on the message $m_b^\star$. Next,
$\cal A$ starts adaptively querying the previous oracles (\textbf{Phase 2}), 
  with the exception of not querying $(\mu^\star,m_i^\star)$, $i=0,1$, to the 
 $\sconfirm$, $\{\confirm,\deny\}$, and $\convert$ oracles. At the end, $\cal A$ outputs a bit $b^\star$. She wins the game if
  $b=b^\star$.
\begin{center}
\footnotesize
\begin{experiment}[Experiment $\mathbf{Exp}_{\cs,\A}^{\mathsf{INV-CMA}}(1^\kappa)$]

\item $\param \gets \setup(1^\kappa)$; 
\item $(\pks,\sks) \gets \cs.\keygen_{\entfont{S}}(1^\kappa)$; $(\pkc,\skc) \gets \cs.\keygen_\C(1^\kappa)$; 

\item $(m_0^{\star},m_1^\star,\mathcal{I}) \leftarrow {\mathcal A}^{\mathfrak{S}, \mathfrak{Cv}, \mathfrak{V}} (\param,\pks,\pkc)$ \\
\phantom{$(m_0^{\star},m_1^{\star},\mathcal{I})$} $\left\vert
\begin{array}{l} 
\mathfrak{S} : m  \longmapsto \cs.\sign_{\sks}(m,\pkc) \\ 
\mathfrak{Cv}: (\mu,m) \longmapsto \cs.\convert_{\skc}(\mu,m,\pks,\pkc) \\
\mathfrak{V} : (\mu,m) \longmapsto \cs.\{\sconfirm,\confirm,\deny\}_{\{\coins_\mu \wedge \skc\}}(\mu,m) \\
\end{array} \right.$ \\

\item $b \hasard \{0,1\}$; $\mu^{\star} \leftarrow \cs.\sign_{\sks}(m^\star_b,\pkc)$\\
\item $b^\star\leftarrow {\mathcal A}^{\mathfrak{S}, \mathfrak{Cv}, \mathfrak{V}}(\text{guess},\mathcal{I},\mu^{\star},\pks,\pkc)$~\\
\phantom{$b^\star \leftarrow$} $\left\vert
\begin{array}{l} 
\mathfrak{S} : m  \longmapsto \cs.\sign_{\sks}(m,\pkc) \\ 
\mathfrak{Cv}: (\mu,m) (\neq (\mu^\star,m_i^\star), i=0,1) \longmapsto \cs.\convert_{\skc}(\mu,m) \\
\mathfrak{V} : (\mu,m) (\neq (\mu^\star,m_i^\star), i=0,1)\longmapsto \\
$\phantom{$(\mu,m) (\neq (\mu^\star,m_i^\star)$} $\cs.\{\sconfirm,\confirm,\deny\}(\mu,m) \\
\end{array} \right.$ 

\item Return $(b=b^\star)$.
\end{experiment}
\end{center}
 We define $\cal A$'s advantage as $ \Adv({\cal A}) = \left | \Pr[b=b^\star] -
  \frac{1}{2} \right|$, where the probability is taken over all the random coins. Finally, a CDCS scheme is
$(t,\epsilon,q_s,q_v,q_{sc})$-INV-CMA secure if no adversary operating in time
$t$, issuing $q_s$ queries to the signing oracle (followed potentially by queries to the $\sconfirm$ oracle), $q_v$ queries to the confirmation/denial oracles and $q_{sc}$ queries to the selective conversion
oracle that wins the above game with advantage greater that $\epsilon$. The probability is taken over all the coin tosses.

\end{description}

We have the following remarks regarding our security model:
\begin{remark}
\label{rmq:securityModel}
\begin{itemize}
\item \textbf{Online vs offline non-transferability}. Our definition of non-transferability is the same adopted in \cite{CamenischMichels2000,GentryMolnarRamzan2005,WangBaekWongBao2007}. In particular, it thrives on the \emph{concurrent} zero knowledgeness of the $\{\sconfirm,\-\confirm,\deny\}$ protocols, and  guarantees only the so-called \emph{offline non-transferability}.\\In fact, non-transferability is not preserved, as remarked by  \cite{LiskovMicali2008}, if the verifier interacts concurrently with the prover and with an unexpected verifier.

 One way to circumvent this shortcoming consists in requiring the mentioned protocols to be \emph{designated verifier proofs \cite{JakobssonSakoImpagliazzo1996}}, i.e. require the verifier to be able to efficiently provide the proofs underlying  $\sconfirm$, $\confirm$, and $\deny$, such that no efficient adversary is able to tell whether he is interacting with the genuine prover or with the verifier. This approach was adhered to for instance in \cite{ChowHaralambiev2011,MonneratVaudenay2011}. We will show that our proposed practical realizations of confirmer signatures satisfy also this stronger notion of \emph{online non-transferability} since the protocols  $\{\sconfirm,\confirm,\deny\}$ can be turned easily into $\Sigma$ protocols which can be in turn transformed efficiently into designated verifier proofs.

\item \textbf{Insider security (for the signer) against \emph{malicious confirmers}}.  We consider the \emph{insider security model against malicious confirmers} in our definition of
  unforgeability. I.e. the adversary is \emph{allowed} to choose his key pair
  $(\sk_C,\pk_C)$. This is justified by the need of preventing the
  confirmer from impersonating the signer by issuing valid signatures on his
  behalf. Hence, our definition of unforgeability, which is the same as the one
  considered by \cite{Wikstroem2007}, implies its similars in \cite{CamenischMichels2000,GentryMolnarRamzan2005,WangBaekWongBao2007}.

\item \textbf{Insider vs outsider security for the confirmer}. Our definition of invisibility, namely INV-CMA, is considered in the \emph{outsider
    security model}. I.e., the adversary does not know the private key of the
  signer. 

Actually, not only outsider security can be enough in many situations as argued earlier in the introduction, but also, there seems to be no tangible extra power that an insider attacker can gain from having access to the signing key. Actually, the insider adversary in the definitions in \cite{GentryMolnarRamzan2005,WangBaekWongBao2007} (constructions from CtEaS) is not allowed to ask the verification/conversion of valid confirmer signatures on the challenge messages (otherwise his task would be trivial: it suffices to replace, in the challenge confirmer signature, the digital signature on the commitment by a new one and ask the resulting confirmer signature for verification/conversion). This restriction involves for instance confirmer signatures that the adversary may have forged on the challenge messages using his signing key.
     
 Let's see how this can translate into a real attack scenario: the insider adversary $\cal A$ has compromised the signer's key and wishes to break the invisibility of an alleged signature $\mu$ on some message $m$. Naturally, the restriction imposed earlier can be achieved by the signer revoking his key and alerting the confirmer not to verify/convert confirmer signatures involving the message $m$ and potentially further messages. The revocation of the signing key implies also not considering signatures that have been issued after the key has been compromised. This leaves $\cal A$ with only verification/conversion queries on  messages where the corresponding signatures have been issued by the genuine signer before the revocation of the signing key. This seems to reduce $\cal A$'s adversarial power down to that of an outsider attacker.

Bottom line is outsider invisibility might be all that one needs and \emph{can have} in practice. Moreover, it allows  the signer to sign the same
  message many times without loss of invisibility, which is profoundly
  needed in licensing software.  

\item \textbf{Invisibility vs non-transferability} Our invisibility notion INV-CMA does not guarantee the
  non-transferability of the signatures. I.e., the confirmer signature might
  convince the recipient that the signer was involved in the signature of some
  message. We refer to the discussion in \cite{GentryMolnarRamzan2005}
  (Section 3) for techniques that can be used by the signer to camouflage the
  presence of valid signatures. We will also propose some constructions (derived from a variant of the StE paradigm) that achieve  a stronger notion of invisibility; for such confirmer signatures, it is difficult to distinguish a valid confirmer signature on some message from a random string sampled from the signature space. 

\end{itemize}
\end{remark}

\subsection{Classical constructions for confirmer signatures}
\label{subsec:mainConstructions}

Consider the following schemes:

\begin{itemize}

\item \textbf{A digital signature scheme $\Sigma$}, given by $\Sigma.\keygen$ which generates a key pair ($\Sigma.\sk$, $\Sigma.\pk$), $\Sigma.\sign$, and $\Sigma.\verify$.

\item \textbf{A public key encryption scheme}, described by \\$\Gamma.\keygen$ that generates the key pair ($\Gamma.\sk$,$\Gamma.\pk$), \\$\Gamma.\encrypt$, and $\Gamma.\decrypt$.

We use the notation  $\Gamma.\encrypt_{\{\Gamma.\pk,\coins\}}(m)$ to refer to the ciphertext obtained from encrypting the message $m$ under the public key $\Gamma.\pk$ using the random coins $\coins$ ($\encrypt$ is a probabilistic algorithm).

\item \textbf{A commitment scheme $\Omega$}, given by the algorithms $\Omega.\commit$ and $\Omega.\open$.

\end{itemize}

Let $m$ be a message. We present now the most popular paradigms used to devise a confirmer signature scheme $\cs$ from the aforementioned primitives. Note that in all those paradigms, ($\Sigma.\pk$,$\Sigma.\sk$) forms the signer's key pair, whereas ($\Gamma.\pk$,$\Gamma.\sk$) forms the confirmer's key pair.

Note that the security analysis of the following constructions is deferred to Section \ref{sec:negative-CDCS} and Section \label{sec:positiveCDCS}.

\subsubsection{The \emph{``sign-then-encrypt'' (StE) paradigm}} 

A CDCS scheme $\cs$ from the StE paradigm is depicted in Figure \ref{fig:StE}. 

\noindent Note that the languages underlying $\sconfirm$, $\confirm$, and  $\deny$ are in  NP and thus accept concurrent zero-knowledge proofs \cite{DworkNaorSahai2004}. This guarantees the completeness, soundness, and (offline) non-transferability of the resulting signatures.

\begin{center}
\begin{figure*}
\[
\begin{array}{|lcl|}
\hline \vspace{.1 cm} 
\bf{[\cs.\setup(1^\kappa)]} &  :& \Sigma.\setup(1^\kappa); ~\Gamma.\setup(1^\kappa) \\ \vspace{.1 cm}

\bf{[\cs.\keygen_\signer(1^\kappa)]}   &: & \Sigma.\keygen(1^\kappa)\\ \vspace{.1 cm}

\bf{[\cs.\keygen_\C(1^\kappa)]}  & :&\Gamma.\keygen(1^\kappa) \\ \vspace{.1 cm}

\bf{[\cs.\sign(m)]}  &:& \Gamma.\encrypt_{\Gamma.\pk}(\Sigma.\sign_{\Sigma.\sk}(m)) \\ \vspace{.1 cm}

\bf{[\cs.\sconfirm(\mu,m)]}  &:& \zk\pok \{(\sigma,\coins_\mu):  \mu = \Gamma.\encrypt_{\{\Gamma.\pk, \coins_\mu\}}(\sigma) \wedge \Sigma.\verify_{\Sigma.\pk}(\sigma,m)=1 \} \\ \vspace{.1 cm}

\bf{[\cs.\confirm(\mu,m)]} &: & \zk\pok \{(\sigma,\Gamma.\sk):  \sigma = \Gamma.\decrypt_{\Gamma.\sk}(\mu) \wedge \Sigma.\verify_{\Sigma.\pk}(\sigma,m)=1 \} \\ \vspace{.1 cm}
 
\bf{[\cs.\deny(\mu,m)]} &:& \zk\pok \{(\sigma,\Gamma.\sk):  \sigma = \Gamma.\decrypt_{\Gamma.\sk}(\mu) \wedge \Sigma.\verify_{\Sigma.\pk}(\sigma,m) = 0 \} \\ \vspace{.1 cm}

\bf{[\cs.\convert(\mu,m)]} &:& {{\Gamma.\decrypt_{\Gamma.\sk}}}(\mu) \\ \vspace{.1 cm}

 \bf{[\cs.\verifyconverted(\sigma,m)]}  &:& \Sigma.\verify_{\Sigma.\pk}(\sigma,m)\\
\hline

\end{array}
\]
\caption{\bf{The StE paradigm}}
\label{fig:StE}
\end{figure*}
\end{center}

\subsubsection{The \emph{``encrypt-then-sign'' (EtS) paradigm}} 

This paradigm was first used in the context of signcryption. We can adapt it to the case of convertible confirmer signatures by requiring a ``trusted authority'' $\ta$ that runs the $\setup$ algorithm and generates a common reference string $\crs$. In fact, signature conversion will involve a non-interactive ZK proof (NIZK), and thus the need for the $\crs$ (generated by a trusted authority) for the simulation of the NIZK proof. We describe in Figure \ref{fig:EtS} confirmer signatures from such a paradigm. 

\noindent Similarly, $\sconfirm$, $\confirm$, and $\deny$ amount to \emph{concurrent zero-knowledge proofs of knowledge} since the underlying languages are in NP. It is worth noting that the aforementioned protocols are only carried out when the signature $\sigma$ on the ciphertext $c$ is valid, otherwise the confirmer signature $\mu=(c,\sigma)$ is obviously deemed invalid w.r.t. m. Finally, $\convert$ outputs (in case of a valid confirmer signature on $m$) a zero knowledge non-interactive (NIZK) proof that $m$ is the decryption of $c$; such a proof is feasible since the underlying statement is in NP (\cite{GoldreichMicaliWigderson1986} and \cite{BlumFeldmanMicali1988}).

\begin{center}
  \begin{figure*}
    \small
\[
\begin{array}{|lcl|}
\hline \vspace{.1 cm} 
\bf{[\cs.\setup(1^\kappa)]} &:&  \Sigma.\setup(1^\kappa);~\Gamma.\setup(1^\kappa) ;~\crs \gets \ta.\setup(1^\kappa)  \\ \vspace{.1 cm}

\bf{[\cs.\keygen_\signer(1^\kappa)]}   &:& \Sigma.\keygen(1^\kappa)\\ \vspace{.1 cm}

\bf{[\cs.\keygen_\C(1^\kappa)]}  &:& \Gamma.\keygen(1^\kappa) \\ \vspace{.1 cm}

\bf{[\cs.\sign(m)]} &: & c \gets \Gamma.\encrypt_{\Gamma.\pk}(m); ~\sigma \gets \Sigma.\sign_{\Sigma.\sk}(c) ; \\ \vspace{.1 cm}

\bf{[\cs.\sconfirm(\{c,\sigma\},m)]}  &:& \zk\pok \{\coins_c:  c = \Gamma.\encrypt_{\{\Gamma.\pk, \coins_c \}}(m) \} \\ \vspace{.1 cm}

\bf{[\cs.\confirm(\{c,\sigma\},m)]} &:& \zk\pok \{\Gamma.\sk:  m = \Gamma.\decrypt_{\Gamma.\sk}(c) \} \\ \vspace{.1 cm}
 
\bf{[\cs.\deny(\{c,\sigma\},m)]} &:& \zk\pok \{\Gamma.\sk:  m \neq \Gamma.\decrypt_{\Gamma.\sk}(c) \} \\ \vspace{.1 cm}

\bf{[ \cs.\convert(\{c,\sigma\},m)]} &:&  \pi \gets \nizk\{m = \Gamma.\decrypt_{\Gamma.\sk}(c) \} ; ~\return ~ \{\pi,c,\sigma\} \\ \vspace{.1 cm}

 \bf{[\cs.\verifyconverted(\{\pi,c,\sigma\},m)]}  &:&  \nizk.\verify(\crs,\pi) ; ~\Sigma.\verify_{\Sigma.\pk}(\sigma,c)\\
\hline

\end{array}
\]
\caption{\bf{The EtS paradigm}}
\label{fig:EtS}
\end{figure*}
\end{center}

\subsubsection{The \emph{``commit-then-encrypt-and-sign'' (CtEaS) paradigm}} This construction has the advantage of performing signature and encryption \emph{in parallel} in contrast to the previous sequential compositions.
It includes among its building blocks, contrarily to the previous constructions, a public key encryption scheme that supports labels. The encryption with labels was introduced in \cite{WangBaekWongBao2007} in order to fix a flaw that afflicted the original proposal in \cite{GentryMolnarRamzan2005}.  More precisely, a confirmer signature on a message $m$ is obtained by first committing to $m$, then encrypting the randomness used in the commitment w.r.t. the label $m\|\Sigma.\pk$ ($\Sigma.\pk$ is the public key of the used digital signature scheme), \emph{and} finally signing the commitment.

However, we remark that this repair will violate the invisibility of the resulting construction. In fact, the standard security definitions for encryption with labels do not require the label of a ciphertext to be hidden (since the label is required as input to the decryption algorithm in order to correctly decrypt the ciphertext). This implies that the signed message will be leaked from the encryption of the randomness used for the commitment.
We can remediate to this problem by using public key encryption without labels to encrypt both the randomness and $m\|\Sigma.\pk$. We describe in Figure \ref{fig:CtEaS} our revised variant of this paradigm.

\noindent Again, $\sconfirm$, $\confirm$, and $\deny$ are only carried out when the signature $\sigma$ on the commitment $c$ is valid, otherwise the confirmer signature $\mu=(c,e,\sigma)$ is clearly invalid w.r.t. m.

\begin{center}
  \begin{figure*}
    \small
\[
\begin{array}{|lcl|}
\hline \vspace{.1 cm} 
\bf{[\cs.\setup(1^\kappa)]} &:&  \Sigma.\setup(1^\kappa);~\Gamma.\setup(1^\kappa); ~\Omega.\setup(1^\kappa)\\ \vspace{.1 cm}

\bf{[\cs.\keygen_\signer(1^\kappa)]}   &:& \Sigma.\keygen(1^\kappa)\\ \vspace{.1 cm}

\bf{[\cs.\keygen_\C(1^\kappa)]}  &:& \Gamma.\keygen(1^\kappa) \\ \vspace{.1 cm}

\bf{[\cs.\sign(m)]}  &:& c \gets \Omega.\commit(m,r); ~e \gets \Gamma.\encrypt_{\Gamma.\pk}                      (r\|m\|\Sigma.\pk);~\sigma \gets \Sigma.\sign_{\Sigma.\sk}(c) ;  \\ \vspace{.1 cm}

\bf{[\cs.\sconfirm(\{c,e,\sigma\},m)]}  &:& \zk\pok \{(r,\coins_c): c = \Omega.\commit(m,r) \wedge e = \Gamma.\encrypt_{\{\Gamma.\pk, \coins_c\}}(r\|m\|\Sigma.\pk) \} \\ \vspace{.1 cm}

\bf{[\cs.\confirm(\{c,e,\sigma\},m)]}  &:& \zk\pok \{(r,\Gamma.\sk): c = \Omega.\commit(m,r) \wedge r\|m\|\Sigma.\pk = \Gamma.\decrypt_{\Gamma.\sk}(e) \} \\ \vspace{.1 cm}
 
\bf{[\cs.\deny(\{c,e,\sigma\},m)]} &:&  \zk\pok \{(r,\Gamma.\sk): c \neq \Omega.\commit(m,r) \wedge r\|m\|\Sigma.\pk = \Gamma.\decrypt_{\Gamma.\sk}(e) \} \\ \vspace{.1 cm}

\bf{[ \cs.\convert(\{c,e,\sigma\},m)]} &:&  r\|m\|\Sigma.\pk \gets \Gamma.\decrypt_{\Gamma.\sk}(c); ~\return ~  \{r,c,\sigma\} \\ \vspace{.1 cm}

 \bf{[\cs.\verifyconverted(\{r,c,\sigma\},m)]}  &:& c \stackrel{?}{=} \Omega.\commit(m,r) ; ~\Sigma.\verify_{\Sigma.\pk}(\sigma,c) \\
\hline

\end{array}
\]
\caption{\bf{The CtEaS paradigm}}
\label{fig:CtEaS}
\end{figure*}
\end{center}

\begin{remark}
 It is possible to require a proof in the $\convert$ algorithms of StE and CtEaS, that the revealed information is indeed a correct decryption of the corresponding encryption; such a proof is again possible to issue (with or without interaction) since the underlying statement is in NP.
\end{remark}

\section{Negative Results for CDCS}
\label{sec:negative-CDCS}

In this section, we show that StE and CtEaS require at least IND-PCA encryption in order to lead to INV-CMA secure confirmer signatures. We proceed as follows.

\indent First, we rule out the OW-CPA, OW-PCA, and IND-CPA notions by remarking that ElGamal's encryption meets all those notions (under different assumptions), but cannot be used in either StE or CtEaS. In fact, the invisibility adversary can create from the challenge signature a new ``equivalent'' signature (by re-encrypting the ElGamal encryption), and query it for conversion or verification to solve the challenge. Actually, this attack applies to any \emph{homomorphic encryption}.

\indent Next, we show the insufficiency of OW-CCA and NM-CPA encryption by means of efficient meta-reductions which forbid the existence of reductions from the invisibility of the resulting confirmer signatures to the OW-CCA or NM-CPA security of the underlying encryption. We first show this impossibility result for a specific kind of reductions, then we extend it to arbitrary reductions assuming further assumptions on the used encryption.

\indent Finally, as an illustration of our techniques, we provide evidence that the well known  Damg\r{a}rd-Pedersen's signature \cite{DamgardPedersen1996} is, contrarily to what is conjectured by the authors, unlikely to be indistinguishable under the DDH assumption.

\subsection{A breach in invisibility using homomorphic encryption}
\label{subsec:invisibilityBreach}
\begin{definition}[Homomorphic encryption]
\label{def:homomorphicEncryption}
A homomorphic public key encryption scheme $\Gamma$ given by $\Gamma.\keygen$, $\Gamma.\encrypt$, and $\Gamma.\decrypt$ has the following properties:

\begin{enumerate}
\item The message space ${\cal M }$ and the ciphertext space ${\cal C}$ are groups w.r.t. some binary operations $\ast_e$ and $\circ_e$ respectively.

\item $\forall (\sk,\pk) \leftarrow \Gamma.\keygen(1^\kappa)$ for any security parameter $\kappa$, $\forall m,m' \in {\cal M}$:
$$
\Gamma.\encrypt_{\pk}(m \ast_e m') = \Gamma.\encrypt_{\pk}(m) \circ_e \Gamma.\encrypt_{pk}(m').
$$

\end{enumerate} 
\end{definition}

Examples of homomorphic encryption \footnote{This encryption is not to confuse with the so-called \emph{fully homomorphic} encryption which preserves the entire ring structure of the plaintexts (supports both addition and multiplication).} in the literature include $\entfont{ElGamal}$ \cite{ElGamal1985}, $\entfont{Paillier}$ \cite{Paillier1999}, and \\$\entfont{Boneh\mbox{-}Boyen\mbox{-}\-Shacham}$ \cite{BonehBoyenShacham2004}. All those schemes are IND-CPA secure (under different assumptions).




\begin{fact}
\label{fact:AttackHomo}
The StE (CtEaS) paradigm cannot lead to INV-CMA secure confirmer signatures when used with homomorphic encryption.
\end{fact}

\begin{proof}
Let $m_0,m_1$ be the challenge messages the invisibility adversary $\cal A$ outputs to his challenger. Let further $\Gamma$, $\Sigma$, and $\Omega$ denote respectively the homomorphic encryption, the digital signature, and the commitment used as building blocks.

\begin{itemize}
\item \emph{StE paradigm.} $\cal A$  receives as a challenge confirmer signature some $\mu_b=\Gamma.\encrypt(\Sigma.\sign(m_b))$, where $b \hasard \{0,1\}$ and is asked to find $b$. To solve his challenge, $\cal A$  obtains another encryption, say $\widetilde{\mu_b}$, of \\$\Sigma.\sign(m_b)$ by multiplying $\mu_b$ with an encryption of the identity element. According to the invisibility experiment, $\cal A$ can query $\widetilde{\mu_b}$ for conversion or verification (w.r.t. either $m_0$ or $m_1$) and the answer to such a query is sufficient for $\cal A$ to conclude.

\item \emph{CtEaS paradigm.} $\cal A$ gets as a challenge confirmer signature some $\mu_b=[c=\Omega.\commit(m_b,r),e,\Sigma.\sign(c)]$ ($b \in \{0,1\}$) where  $e$ is an encryption of $r\|m_b\|\Sigma.\pk$. Similarly, $\cal A$ computes a new confirmer signature on $m_b$ by multiplying $e$ with an encryption of the identity element (of the message space of $\Gamma$). Then, $\cal A$ queries this new signature (w.r.t. either $m_0$ or $m_1$) for conversion/verification, and the answer of the latter is sufficient for $\cal A$ to conclude.
\end{itemize}
\qed\end{proof}

\begin{remark}
Note that the EtS paradigm is resilient to the previous attack since the adversary would need to compute a valid digital signature on the newly computed encryption. This is not plausible in the invisibility game (we consider \emph{outsider} invisibility).
\end{remark}

\begin{corollary}
\label{corollary:ElGamalAttack}
Invisibility in CDCS from StE  and CtEaS cannot rest on OW-CPA, OW-PCA, or IND-CPA encryption.
\end{corollary}

\begin{proof}
ElGamal's encryption \cite{ElGamal1985} is  homomorphic and meets the OW-CPA, OW-PCA, and IND-CPA security notions (under different assumptions). The rest follows from the previous fact.
\qed\end{proof}

\subsection{Impossibility results for key-preserving reductions}
In this paragraph, we prove that NM-CPA and OW-CCA encryption are insufficient for invisible confirmer signatures from StE or CtEaS, if we consider a certain type of reductions. We do this by means of efficient \emph{meta-reductions} that use such reductions (the algorithm reducing NM-CPA (OW-CCA) breaking the underlying encryption scheme to breaking the invisibility of the construction) to break the NM-CPA (OW-CCA) security of the
encryption scheme. Thus, if the encryption scheme is NM-CPA (OW-CCA) secure, the meta reductions forbid the existence of such reductions. In case the encryption scheme is not NM-CPA (OW-CCA) secure, such reductions will be useless. 

Meta-reductions have been successfully used in a number of important
cryptographic results, e.g. the result in \cite{BonehVenkatesan1998} which proves the impossibility
of reducing factoring, in a specefic kind of way, to the RSA problem, or the results in
\cite{PaillierVergnaud2005,Paillier2007} which show that some well known
signatures which are proven secure in the random oracle may not conserve the same
security in the standard model. Such impossibility
results are in general partial as they apply only for certain reductions. Our result is
also partial in a first stage since it requires the reduction
$\cal R$, trying to attack a certain property of an encryption scheme given by the public
key $\Gamma.\pk$, to provide the adversary against the confirmer signature
with the confirmer public key $\Gamma.\pk$. In other terms, our result applies for key-preserving reductions (see Definition \ref{def:key-preserving-reductions}). Our restriction to such a class
of reductions is not unnatural since, to our best knowledge, all the
reductions basing the security of the generic constructions of confirmer
signatures on the security of their underlying components, feed the adversary with the public keys of these components
(signature schemes, encryption schemes, and commitment schemes). Next, we use similar techniques to \cite{PaillierVillar2006} to extend our impossibility results to arbitrary reductions.

\subsubsection{Insufficiency of OW-CCA encryption}
\begin{lemma}
\label{lemma:OW-CCA-CS}
Assume there exists a key-preserving reduction $\cal R$ that converts an INV-CMA adversary $\cal A$ against confirmer signatures from the StE (CtEaS) paradigm to a OW-CCA adversary against the underlying encryption scheme. Then, there exists a meta-reduction $\cal M$ that OW-CCA  breaks the encryption scheme in question.
\end{lemma}
This lemma claims that if the considered encryption is  OW-CCA secure, then, there exists no key-preserving reduction $\cal R$ that reduces OW-CCA breaking it to INV-CMA breaking the construction (from either StE or CtEaS), or if there
exists such an algorithm, then the underlying encryption is not OW-CCA secure,
thus rendering such a reduction useless.
\begin{proof}
Let $\cal R$ be the
key-preserving reduction that reduces OW-CCA breaking
the encryption scheme underlying the construction to INV-CMA breaking the
construction itself. We will construct an algorithm $\cal M$  that uses $\cal R$  to OW-CCA break the same encryption scheme by simulating an execution of the
INV-CMA adversary $\cal A$ against the construction.

Let $\Gamma$ be the encryption scheme $\cal M$ is trying to attack w.r.t.  key $\Gamma.\pk$. $\cal M$ proceeds as follows:

\begin{itemize}
\item \emph{StE paradigm. }Let $c$ be the OW-CCA challenge $\cal M$ is asked to resolve. $\cal M$ launches $\R$ over $\Gamma$ under the same key $\Gamma.\pk$ and the same challenge $c$. Obviously, all decryption queries made by $\R$ can be perfectly answered using $\cal M$'s challenger (since they are different from the challenge $c$). 

$\cal M$ needs now to simulate an INV-CMA adversary $\cal A$ to $\cal R$. To do so, $\cal M$ picks two random messages $m_0$ and $m_1$ from the message space.
To insure that $c$ is not a valid confirmer signature on $m_0$ nor $m_1$, $\cal M$ queries $\R$ for the conversion of both $(c,m_0)$ and $(c,m_1)$ and makes sure that the result to both queries is $\perp$. If this is not the case, then  $\cal M$ will simply abort the INV-CMA game, and output the result of the conversion, say $\sigma$ ($\neq \perp$), to his own OW-CCA challenger. In fact, by definition, $\sigma$ is a valid decryption of $c$ w.r.t. $\Gamma.\pk$.

We assume now that $c$ is not a valid confirmer signature on either $m_0$ or $m_1$. Hence, $\cal M$ outputs $m_0,m_1$ to $\R$ as challenge messages, and receives a challenge $\mu_b$ which is, with enough good probability, a valid confirmer signature on $m_b$ for $b \in \{0,1\}$. $\mu_b$ is according to our assumption different from the challenge ciphertext $c$, and $\cal M$ is requested to find $b$. To solve his challenge, $\cal M$ queries his own OW-CCA challenger for the decryption of $\mu_b$. The result to such a query allows $\cal M$ to find out $b$ with probability one (provided $\cal R$ supplies a correct simulation). 

\item \emph{CtEaS paradigm. }$\cal M$ launches $\cal R$ over $\Gamma$ with the same  key $\Gamma.\pk$ and the same challenge $e$. Thus, all decryption queries made by $\cal R$, which are by definition different from the challenge $e$, can be forwarded to $\cal M$'s own challenger. 

At some point, $\cal M$, acting as an INV-CMA
attacker against the construction, outputs two challenge messages $m_0,m_1$ (chosen randomly from the message space) and gets as response a challenge  $\mu_b=(c_b,e_b,\sigma_b)$ which is, with enough good probability, a valid confirmer signature on $m_b$ for some $b \in \{0,1\}$. $\cal M$ is asked to find $b$.

We first note that $e_b \neq e$. In fact, with overwhelming probability, the challenge $e$  does not encrypt a string whose suffix is  $m_0\|\Sigma.\pk$ or $m_1\|\Sigma.\pk$ (although $\Sigma.\pk$ can be maliciously chosen by $\cal R$, $m_0$ and $m_1$ are independently chosen by $\cal M$ upon receipt of the challenge $e$). Therefore, $\cal M$ requests his  own challenger for the decryption of $e_b$. The answer to such a query will allow $\cal M$ (behaving as an INV-CCA attacker) to perfectly answer his invisibility challenge.

\end{itemize}

To sum up, $\cal M$ is able to perfectly answer the decryption queries made by $\cal R$ (that are by definition different from the OW-CCA challenge). $\cal M$ is further capable of successfully simulating an INV-CMA attacker against the construction (from either StE or CtEaS), provided $\cal R$ supplies a correct simulation. Thus, $\cal R$ is expected to return the answer to the OW-CCA challenge. Upon receipt of this answer, $\cal M$ will forward it to his own challenger.
\qed \end{proof}

\subsubsection{Insufficiency of NM-CPA encryption}
\begin{lemma}
\label{lemma:NM-CPA-CS}
Assume there exists a key-preserving reduction $\cal R$ that converts an INV-CMA adversary $\cal A$ against confirmer signatures from the StE (CtEaS) paradigm to an NM-CPA adversary against the underlying encryption scheme. Then, there exists a meta-reduction $\cal M$ that NM-CPA  breaks the encryption scheme in question.
\end{lemma}

\begin{proof}
Let $\cal R$ be the key-preserving reduction that reduces NM-CPA breaking
the encryption  underlying the construction to INV-CMA breaking the
construction (from StE or CtEaS). We construct an algorithm $\cal M$  that uses $\cal R$  to NM-CPA break the same encryption scheme by simulating an execution of the INV-CMA  adversary $\cal A$ against the construction.

Let $\Gamma$ be the encryption scheme $\cal M$ is trying to attack w.r.t. public key $\Gamma.\pk$. $\cal M$ will launch $\cal R$ over the same public key $\Gamma.\pk$. Next, $\cal M$ will simulate an INV-CMA adversary against the constructions:

\begin{itemize}
\item \emph{StE paradigm. }$\cal M$ (behaving as $\cal A$) queries $\cal R$ on two messages $m_0,m_1$ ($m_0 \neq m_1$) for  confirmer signatures. Let $\mu_0,\mu_1$ be the corresponding confirmer signatures resp. $\cal M$ further queries $(\mu_i,m_i)$, $i \in \{0,1\}$, for conversion. Let $\sigma_0,\sigma_1$ be the corresponding answers respectively. We assume that $\sigma_0 \neq \sigma_1$. If this is not the case, $\cal M$ repeats the experiment until this holds (if all confirmer signatures are encryptions of the same string $\sigma$, then the construction is not secure). At that point, $\cal M$ outputs $D=\{\sigma_0,\sigma_1\}$, to his NM-CPA challenger, as a distribution probability from which the messages will be drawn. He gets a
challenge encryption $\mu^\star$, of either $\sigma_0$ or $\sigma_1$ under $\Gamma.\pk$,
and is asked to produce a ciphertext $\mu'$ whose corresponding plaintext is meaningfully related to the decryption of $\mu^\star$. To solve his task, $\cal M$ queries for instance $(\mu^\star, m_0)$ for conversion. If the result is $\sigma_0$, i.e. $\mu^\star$ is a valid confirmer signature on $m_0$, then  $\cal M$ outputs $\Gamma.\encrypt_{\pk}(\overline{\sigma_0})$ ($\overline{m}$ refers to the bit-complement of $m$) and the relation $R$: $R(m,m')=(m'=\overline{m})$. Otherwise, $\cal M$ outputs $\Gamma.\encrypt_{\pk}(\overline{\sigma_1})$ and the same relation $R$. Finally $\cal M$ aborts the INV-CMA game.   

\item \emph{CtEaS paradigm. }Similarly, $\cal M$ queries $\cal R$ on $m_0,m_1$ ($m_0 \neq m_1$) for confirmer signatures. Let $\mu_0=(c_0,e_0,\sigma_0)$ and $\mu_1=(c_1,e_1,\sigma_1)$ be the corresponding confirmer signatures. $\cal M$ queries again $\mu_0,\mu_1$, along with the corresponding messages, for conversion. Let $r_0$ and $r_1$ be the the randomnesses used to generate the commitments $c_0$ and $c_1$ on $m_0$ and $m_1$ resp. $\cal M$ inputs $ {\cal D}=\{r_0\|m_0\|\Sigma.\pk,r_1\|m_1\|\Sigma.\pk\}$ to his own challenger as a distribution
probability from which the plaintexts will be drawn. $\cal M$ will
receive as a challenge encryption $e^\star$. At that point, $\cal M$ chooses a bit $b \hasard \{0,\}$ , and 
queries $\cal R$ on $\mu^\star=(c_b,e^\star,\sigma_b)$ and the message $m_b$
for conversion. Note that if $e^\star$ encrypts $r_b\|m_b\|\Sigma.\pk$, then $\mu^\star$ is a valid confirmer signature on $m_b$, otherwise it is invalid. Therefore, if the outcome of the query is not $\perp$, then $\cal M$ outputs  $\Gamma.\encrypt_{\pk}(\overline{r_b})$, where $\overline{r_b}$ refers to the bit-complement of $r_b$, and the relation $R$: $R(r,r')=(r'=\overline{r})$. Otherwise, $\cal M$ outputs $\Gamma.\encrypt_{\pk}(\overline{r_{1-b}})$ and the same relation
$R$. Finally $\cal M$ aborts the INV-CMA game.

\end{itemize}

Clearly, $\cal M$ solves correctly his NM-CPA challenge if $\cal R$ provides a correct simulation.
\qed \end{proof}

\subsubsection{Putting all together}

\begin{thm}
\label{thm:negative-key-preserving}
 Consider the security notions obtained
from pairing a security goal $\GOAL \in \{\OW,\IND,\NM\}$ and an attack model $\ATK \in \{\CPA,\PCA,\-\CCA\}$. The encryption scheme underlying the above constructions (from either StE or CtEaS) must be at least IND-PCA
secure, in case the considered reduction is key-preserving, in order to achieve INV-CMA secure confirmer signatures.
\end{thm}

\begin{proof}
Corollary \ref{corollary:ElGamalAttack} rules out OW-CPA, OW-PCA, and IND-CPA encryption. Moreover, Lemma \ref{lemma:OW-CCA-CS} and  Lemma \ref{lemma:NM-CPA-CS}  rule out OW-CCA and NM-CPA encryption resp. The next notion to be considered is IND-PCA. 
\qed\end{proof}

\begin{remark}
Note that the notions OW-CPA, OW-PCA, and IND-CPA are discarded regardless of the used reduction. In fact, we managed to exhibit an encryption scheme (ElGamal's encryption) which meets all those notions, but leads to insecure confirmer signatures when used in the StE or CtEaS paradigms.  
\end{remark}

\begin{remark}
The step of ruling out OW-CPA, OW-PCA, and IND-CPA is necessary although we have proved the insufficiency of stronger notions, namely OW-CCA and NM-CPA. In fact, suppose there is an efficient ``useful'' key-preserving reduction $\cal R$ (i.e. $\R$ solves a presumably hard problem) which reduces OW-PCA breaking a cryptosystem $\Gamma$ underlying a StE or CtEaS construction to INV-CMA breaking the construction itself. Then there exists an efficient key-preserving reduction say ${\cal R}'$ that reduces OW-CCA breaking $\Gamma$ to INV-CMA breaking the construction (OW-CCA is stronger than OW-PCA). This does not contradict Lemma \ref{lemma:OW-CCA-CS} as long as $\Gamma$ is not OW-CCA secure (although it is OW-PCA secure). In other terms, since there are separations between OW-CCA and OW-PCA (same for the other notions), we cannot apply the insufficiency of OW-CCA (NM-CPA) to rule out the weaker notions. 

\noindent This necessity will become more apparent in Section \ref{sec:summary} as we mention how to rule out OW-CCA secure encryption in constructions of verifiably encrypted signatures (VES) from StE, yet there exists many realizations of secure VES realizing StE that use OW-PCA encryption (e.g. ElGamal's encryption in bilinear groups) as a building block, e.g. the VES \cite{BonehGentryLynnShacham2003}.
\end{remark}

\paragraph{On the resort to meta-reductions }It is tempting to envisage stronger techniques than meta-reductions in order to achieve the aforementioned negative results. In fact, meta-reductions give only partial results as they consider a specific class of reductions, e.g. key-preserving reductions.

\noindent For instance, one might try to adapt existing results that separate security notions in encryption, e.g. \cite{BellareDesaiPointchevalRogaway1998}. The problem is that the invisibility adversary in confirmer signatures does not have explicit access to a decryption oracle, i.e. the adversary gets the decryption of a ciphertext only if the latter is part of a valid confirmer signature on some message. Therefore, the separation techniques used in encryption cannot be straightforwardly used in case of confirmer signatures.

\noindent Another possibility consists in building simple counter examples of encryption schemes which are OW-CCA (NM-CPA) secure but lead to insecure confirmer signatures when used in the StE or CtEaS paradigms. Again, it seems difficult to achieve results using this approach without assuming special security properties on the used digital signature scheme, i.e. consider signature schemes that are not strongly unforgeable. 

\noindent The merit of meta-reductions lies in achieving separation results regardless of the used digital signature. 

\noindent We will see in the next subsection how to extend our negative results if the encryption underlying the constructions satisfies further security properties.

\subsection{Extension to arbitrary reductions}
\label{subsec:arbitraryReductions}
To extend the results of the previous paragraph to arbitrary
reductions, we first define the notion of \emph{non-malleability of an encryption scheme key generator} through the following two
games:

\noindent In \textbf{Game 0}, we consider an algorithm $\cal R$ trying to
break an encryption scheme $\Gamma$, w.r.t. a public key $\Gamma.\pk$, in the sense of NM-CPA (or OW-CCA) using an
adversary $\A$ which solves a problem $A$, perfectly reducible to OW-CPA breaking the
encryption scheme $\Gamma$. In this game, $\R$ launches $\A$ over his own
challenge key $\Gamma.\pk$ and some other parameters chosen freely by $\R$ (according to the specifications of $\A$). We
will denote by $\adv_0({\cal R}^{\cal A})$ the success
probability of $\R$ in such a game, where the probability is taken over
the random tapes of both $\R$ and $\A$. We further define
$\success_{\Gamma}^{\zgame}({\A})=\max_{\R}\adv_0({\R}^{\A})$ to be the success in
\textbf{Game 0} of the best reduction $\R$ making the best possible use of
the adversary $\A$. Note that the goal of \textbf{Game 0} is to include
all key-preserving reductions $\R$ from NM-CPA (or OW-CCA)
breaking the encryption scheme in question to solving a problem $A$, which is
reducible to OW-CPA breaking the same encryption scheme.

\noindent In \textbf{Game 1}, we consider the same entities as in \textbf{Game
0}, with the exception of providing $\R$ with, in addition to $\A$, a
OW-CPA oracle (i.e. a decryption oracle corresponding to $\Gamma$) that he can
query w.r.t. any public key $\Gamma.\pk' \neq \Gamma.\pk$, where $\Gamma.\pk$
is the challenge public key of $\R$. Similarly, we define $\adv_1({\R}^{\A})$ to be the success of $\R$ in such a game, and
$\success_{\Gamma}^{\ogame}({\A})=\max_{\R}\adv_1({\R}^{\A})$ the success in
\textbf{Game 1} of the reduction $\R$ making the best possible use of the
adversary $\A$ and of the decryption (OW-CPA) oracle.

\begin{definition}
An encryption scheme $\Gamma$ is said to have a non-malleable key generator if 
$$
\Delta= max_{\A} \left | \success_{\Gamma}^{\ogame}({\A}) - \success_{\Gamma}^{\zgame}({\A})\right |
$$
is negligible in the security parameter.
\end{definition}
This definition informally means that an encryption scheme has a non-malleable key
generator if NM-CPA (or OW-CCA) breaking it w.r.t. a key $\pk$ is no easier
when given access to a decryption (OW-CPA) oracle w.r.t. any public key $\pk'\neq \pk$.

We generalize now our impossibility results to arbitrary reductions as follows.

\begin{thm}
\label{thm:arbitraryreductions}
Theorem \ref{thm:negative-key-preserving} is still valid when considering arbitrary reductions, provided  the encryption scheme underlying the constructions (from either StE or CtEaS) has a non-malleable key
generator.
\end{thm}

To prove this theorem, we first need the following
lemma (similar to Lemma 6 of \cite{PaillierVillar2006})

\begin{lemma}
Let $\cal A$ be an adversary solving a problem $A$, reducible to OW-CPA breaking
an encryption scheme $\Gamma$, and let $\cal R$ be an arbitrary
reduction $\cal R$ that NM-CPA (OW-CCA) breaks $\Gamma$ given
access to $\cal A$. We have $\adv({\R}) \leq \success_{\Gamma}^{\ogame}({\A})$.
\end{lemma}

\begin{proof}
We will construct an algorithm $\cal M$ that plays \textbf{Game 1} with
respect to a perfect
oracle for $\cal A$ and succeeds in breaking the NM-CPA (OW-CCA) security of
$\Gamma$ with the same success probability of $\cal R$. Algorithm $\cal M$ gets a challenge
w.r.t. a public key $\pk$ and launches $\cal R$ over the same challenge and
the same public key. If $\cal R$ calls $\cal A$ on $\pk$, then $\cal M$ will call his own oracle for
$\cal A$. Otherwise, if $\cal R$ calls $\cal A$ on $\pk' \neq \pk$, $\cal M$
will invoke his own decryption oracle for $\pk'$ (OW-CPA oracle) to answer the
queries. In fact, by assumption, the problem A is reducible to OW-CPA breaking $\Gamma$. Finally,
when $\cal R$ outputs the result to $\cal M$, the latter will output the same
result to his own challenger.
\qed\end{proof}
The proof of Theorem \ref{thm:arbitraryreductions} is
similar to that of Theorem 5 in \cite{PaillierVillar2006}:
\begin{proof}
We first remark that the invisibility of the constructions in question is perfectly reducible to OW-CPA breaking the encryption scheme
underlying the construction.

\noindent Next, we note  that the advantage of the meta-reduction $\cal M$ in the proof of Lemma
\ref{lemma:NM-CPA-CS} (Lemma \ref{lemma:OW-CCA-CS}) is at least the same as the advantage of any
key-preserving reduction $\cal R$ reducing the invisibility of a given
confirmer signature to the NM-CPA (OW-CCA) security of its underlying encryption scheme
$\Gamma$. For
instance, this applies to the reduction making the best use of an invisibility
adversary $\cal A$ against the construction. Therefore we have: $\success_{\Gamma}^{\zgame}({\A}) \leq \success(\NM\dash\CPA[\Gamma])$,
 where $\success(\NM\dash\CPA[\Gamma])$ is the success of breaking $\Gamma$ in the
 NP-CPA sense. We also have $\success_{\Gamma}^{\zgame}({\A}) \leq \success(\OW\dash\CCA[\Gamma])$.
\noindent Now, Let $\cal R$ be an arbitrary reduction from NM-CPA (OW-CCA) breaking an encryption scheme $\Gamma$, with a non-malleable key
generator, to INV-CMA breaking the construction (using the same encryption scheme $\Gamma$). We have 
\begin{eqnarray*}
\adv({\R}) & \leq & \success_{\Gamma}^{\ogame}({\A})\\
               & \leq & \success_{\Gamma}^{\zgame}({\A}) + \Delta \\
               & \leq & \min \left\{\success(\NM\dash\CPA[\Gamma],\success(\OW\dash\CCA[\Gamma])\right \} + \Delta
\end{eqnarray*}
since $\Delta$ is negligible, if $\Gamma$ is NM-CPA (OW-CCA) secure, then the advantage of $\cal R$ is also negligible.

\qed\end{proof}

\paragraph{Existence of encryption with a non-malleable key generator}It is not difficult to see that factoring or RSA based encryption schemes are the first candidates to have a non-malleable key generator. In fact, if the public key in these schemes consists only of an RSA modulus $n$, then factorization of other moduli will not help factoring $n$. Examples of such schemes are countless and include OAEP \cite{BellareRogaway1993}, REACT-RSA
\cite{OkamotoPointcheval2001}, PKCS\#1 v2.2 \cite{PKCS1}, Rabin and related systems (Williams \cite{Williams1980},
Blum-Goldwasser \cite{BlumGoldwasser1984}, Chor-Goldreich\cite{ChorGoldreich1984}), the EPOC family, Paillier \cite{Paillier1999}, etc.

\noindent Discrete-log based encryption schemes fail however into this category. Actually, a discrete-log oracle w.r.t. some generator of a given group is sufficient to extract the discrete-log of any element (w.r.t. any element) in this group. Therefore, extension of the above separation results is not straightforward for these schemes; it must use the specific properties of the used encryption scheme. We provide an illustration of such an extension in the next subsection.

\subsection{Analysis of Damg\r{a}rd-Pedersen's \cite{DamgardPedersen1996} undeniable signatures}
\label{app:damgardPedersen}

Let $m \in \{0,1\}^\star$ be an arbitrary message. Damg\r{a}rd-Pedersen's confirmer signature consists of the following procedures:

\begin{description}
\item \textbf{Setup ($\setup$). }On input the security parameter $\kappa$, generate 
a $k$-bit prime $t$ and a prime $p \equiv 1 \bmod t$. Furthermore, select a
collision-resistant hash function $H$ that maps arbitrary-length messages to $\bbbz_t$. 
\item \textbf{Key generation ($\keygen$). } Generate $g$ of
order $t$, $x \in \bbbz_t^\times$, and $h=g^x \bmod p$. Furthermore, select a
generator $\alpha$ of $\bbbz_t^\times$ and $\nu \in \{0,1,\ldots,t-1\}$, and compute $\beta = \alpha^\nu \bmod t$. The public key is
 $\pk=(p,t,g,h,\alpha,\beta)$ and the private key is $(x,\nu)$.
\item \textbf{Signature ($\sign$). }The signer first computes an ElGamal
  signature $(s,r)$ on $m$, i.e. compute $r = g^b \bmod p$ for some $b
  \hasard \bbbz_t^\times$, then compute $s$ as $h(m)=rx+bs\bmod t$. Next, he
  computes an ElGamal encryption  $(E_1=\alpha^\rho,E_2=s\beta^\rho) \bmod t$, for $\rho
  \hasard \bbbz_{t-1}$, of $s$. The undeniable signature on $m$ is the triple $(E_1,E_2,r)$.
\item \textbf{Confirmation/Denial protocol ($\confirm$/$\deny$). }To confirm (deny) a purported
  signature $(E_1,E_2,r)$ on a certain message $m$, the signer issues a $\zk\pok$
 of the language: (see  \cite{DamgardPedersen1996} )
$$
\left\{ s  \colon \DL_{\alpha}(\beta)
= \DL_{E_1}(E_2\cdot s^{-1}) \wedge g^{h(m)}h^{-r} = (\neq) r^s  \right\}
$$
\end{description}

In \cite{DamgardPedersen1996}, the authors prove that the above signatures are unforgeable if the underlying ElGamal signature is also unforgeable, and they conjecture that the signatures meet the following invisibility notion if the problem DDH is hard:

\begin{definition}[Signature indistinguishability]
\label{def:sigInd}
It is defined through the
following game between an attacker $\cal A$ (\emph{a distinguisher}) and his challenger $\cal R$.
\begin{description}
 \item \textbf{Phase 1} after $\cal A$ gets the public key of the scheme $\pk$,
  from $\cal R$, he starts issuing \emph{status requests} and \emph{signature requests}. In a
status request, $\cal A$ produces a pair $(m,z)$, and receives a $1$-bit
answer which is $1$ iff $z$ is a valid undeniable signature on $m$
w.r.t. $\pk$. In a signature request, $\cal A$ produces a message $m$ and
receives an undeniable signature $z$ on it w.r.t. $\pk$.

\item \textbf{Challenge} Once $\cal A$ decides that \textbf{Phase 1} is over,
  he outputs a message $m$ and receives a string $z$ which is either a
  valid undeniable signature on $m$ (w.r.t $\pk$) or a randomly chosen string from the signature space.

\item \textbf{Phase 2} $\cal A$ resumes adaptively making the previous types
  of queries, provided that $m$ does not occur in any request, and that $z$ does
  not occur in any status request. Eventually, $\cal A$ will output a bit.
\end{description}
Let $p_r$, resp. $p_s$ be the probability that $\cal A$ answers $1$ in the
real, resp. the simulated case. Both probabilities are taken over the
random coins of both $\cal A$ and $\cal R$. We say that the signatures are
indistinguishable if $|p_r-p_s|$ is a negligible function in the security parameter.  
\end{definition}

It is clear that the Damg\r{a}rd-Pedersen signatures do not provide the INV-CMA notion according to Subsection \ref{subsec:invisibilityBreach}. In the rest of this section, we provide evidence that the Damg\r{a}rd-Pedersen signatures
are unlikely to meet the above indistinguishability notion under the DDH assumption. 

\begin{lemma}
Assume there exists a key-preserving reduction $\cal R$ that uses an
indistinguishability adversary
$\cal A$ (in the sense of Definition \ref{def:sigInd}) against the above scheme to solve the DDH problem. Then, there exists
an efficient meta-reduction $\cal M$ that solves the DDH problem.
\end{lemma}

\begin{proof}
Let $\cal R$ be the key-preserving reduction that reduces the DDH problem to
distinguishing the Damg\r{a}rd-Pedersen signatures in the sense of Definition
\ref{def:sigInd}. We will construct an algorithm $\cal M$ that uses $\cal R$
to  solve the DDH problem by simulating a distinguisher against the
signatures.

Let $(\alpha,\beta,c_1=\alpha^a,c_2=\beta^b) \in \bbbz_t^\times \times \bbbz_t^\times$ be the DDH instance $\cal
M$ is asked to solve. $\cal M$ acting as a distinguisher of the signature will
make a signature request on an arbitrary message $m$. Let $(E_1,E_2,r)$ be the
answer to such a query. $\cal M$ will make now a status query on $(c_1\cdot
E_1,c_2\cdot E_2,r)$ and the message $m$. $(\alpha,\beta,c_1,c_2)$ is a yes-Diffie-Hellman
instance iff the result of the last query is the confirmation that $(c_1\cdot
E_1,c_2\cdot E_2,r)$ is a signature on $m$.
\qed \end{proof}

\paragraph{Extension to arbitrary reductions. } We cannot employ, in this case, the non-malleability of the key generator technique discussed above. In fact, this would correspond  to
assuming that the DDH problem, w.r.t. a given public key $\pk$, is difficult even when given access to a CDH oracle w.r.t. any $\pk' \neq \pk$, which is untrue.

\noindent  However, we can see that the result still holds true if $\cal R$ feeds the adversary $\cal A$ with the confirmer key $(\alpha',\beta')=(\alpha^\ell,\beta^\ell)$ for some $\ell$ not necessarily known to $\cal M$. In fact, $\cal M$ (or $\cal A$) checks that $(\alpha',\beta',\alpha,\beta)$ is a Diffie-Hellman tuple by first making a signature request on some message, then making a status request on the same message and on the product of the corresponding confirmer signature and the tuple $(\alpha,\beta,1)$ (the answer to such a status request should be the execution of the confirmation oracle). Next, $\cal A$ checks his DDH instance  $(\alpha,\beta,c_1,c_2)$ by using the same technique, namely first make a signature request on some message, followed by a status request on the same message and the product of the returned signature with the tuple $(c_1,c_2,1)$. The answer to this query is sufficient for $\cal M$ to conclude.

Finally, Damg\r{a}rd-Pedersen's undeniable signatures can be repaired so as to provide invisibility by producing  ElGamal's signature on the message to be signed concatenated with the used encapsulation $E_1$. In Section \ref{subsec:efficientStE-CDCS}, we will see that this repair is a special instance of our new StE paradigm for CDCS.

\section{Positive Results for CDCS}
\label{sec:positiveCDCS}
The above negative results are due, to a large extent, to the \emph{strong forgeability} of the confirmer signatures from StE or CtEaS. I.e. a polynomial-time adversary is able to produce, given a valid
confirmer signature on some message, another valid confirmer signature on the same message
without the help of the signer; the attacker requests the conversion of the given confirmer signature and then obtain a new confirmer signature on the same message by simply re-encrypting the response (note that a conversion query is not necessary if the used encryption scheme is homomorphic according to Subsection \ref{subsec:invisibilityBreach}). Therefore, any
reduction $\cal R$ from the invisibility of the construction to the security of the underlying encryption scheme will need more than a list of records maintaining the queried messages along with the corresponding confirmer and
digital signatures. Thus the insufficiency of notions like IND-CPA. In
\cite{CamenischMichels2000,GentryMolnarRamzan2005,WangBaekWongBao2007}, the authors stipulate that the given reduction
would need a decryption oracle (of the encryption scheme) in order to
handle the queries made by the INV-CMA attacker $\cal A$, which makes the
invisibility of the constructions rest on the IND-CCA security of the
encryption scheme. In this section, we remark that the queries made by $\cal A$ are not
completely uncontrolled by $\cal R$; they are encryptions of some
data already released by $\cal R$, provided the digital signature scheme is
strongly unforgeable, and thus known to her. Therefore, a
plaintext checking oracle suffices to handle those queries. 

The rest of this section will be organized as follows. In Subsection \ref{subsec:IND-PCA-sufficient}, we show that StE and CtEaS can thrive on IND-PCA encryption provided the used signature scheme is SEUF-CMA secure. Next, in Subsections \ref{subsec:efficientStE-CDCS} \& \ref{subsec:efficientCtEaS-CDCS}, we propose efficient variants of secure StE and CtEaS respectively which rest on IND-CPA encryption.

\subsection{Sufficiency of IND-PCA encryption}
\label{subsec:IND-PCA-sufficient}

\begin{thm}[StE paradigm]
\label{thm:INDPCASufficient-StE}
Given $(t,q_s,q_v,q_{sc}) \in \mathbb{N}^4$ and $(\varepsilon,\epsilon') \in [0,1]^2$, confirmer signatures from StE are ($t,\epsilon,q_s,q_v,q_{sc}$)-INV-CMA secure if the
underlying digital signature is $(t,\epsilon',q_s)$-SEUF-CMA secure and the underlying encryption scheme is
($t+q_sq_{sc}(q_{sc}+q_v),\epsilon\cdot (1-\epsilon')^{(q_{sc}+q_v)},\-q_{sc}(q_{sc} + q_v)$)-IND-PCA secure.
\end{thm}
\begin{proof}
Let $\cal A$ be an attacker that ($t,\epsilon,q_s,q_v,q_{sc}$)-INV-CMA breaks
a confirmer signature from StE, believed to use a $(t,\epsilon',q_s)$-SEUF-CMA
signature scheme. We construct an algorithm $\cal R$  that IND-PCA breaks the underlying
encryption:

\begin{description}
\item  \textbf{[$\keygen$]} $\cal R$ gets the public parameters
  of the target encryption scheme from her challenger, that are $\Gamma.\pk$ ,
  $\Gamma.\encrypt$, and $\Gamma.\decrypt$. Then, she chooses a secure signature scheme $\Sigma$ with parameters $\Sigma.\pk$, $\Sigma.\sk$,
  $\Sigma.\sign$, and $\Sigma.\verify$. 
\item \textbf{[$\sign$ queries]} For a signature query on a message
  $m$, $\cal R$ proceeds exactly as the standard algorithm using $\Sigma.\sk$ and $\Gamma.\pk$. $\cal R$ further maintains \emph{internally} in a list $\cal L$ the queried messages along with the corresponding confirmer signatures and the intermediate values namely the digital signatures (on the message) and the random nonce used to produce the confirmer signatures. It is clear that this simulation is indistinguishable from the standard $\sign$ algorithm.

\item \textbf{[$\sconfirm$ queries]} $\cal R$ executes the standard $\sconfirm$ protocol on a just generated signature using the randomness used to produce the confirmer signature in question. 

\item \textbf{[$\convert$ queries]} For a putative confirmer signature
  $\mu$ on $m$, $\cal R$ will look up the list $\cal L$. We note
  that each record of $\cal L$ comprises  four components, namely, (1) $m_i:$ the queried
  message (2) ${\sigma_i}:$ the digital signature on $m_i$ (3) $\mu_i = \Gamma.\encrypt_{\Gamma.\pk}(\sigma_i):$ the confirmer signature on $m_i$ (4) $r_i:$ the randomness used to encrypt $\sigma_i$ in $\mu_i$. 

If no record having as first component
  the message $m$ appears in $\cal L$, then $\cal R$ will output $\perp$. 

Otherwise, let $t$ be
  the number of records having as first component the message $m$. $\cal R$ will invoke the plaintext checking
  oracle (PCA) furnished by her own challenger on $(\sigma_i,\mu)$, for $1
  \leq i \leq t$, where $\sigma_i$ corresponds to the second component of such records. If the PCA
  oracle identifies $\mu$ as a valid encryption of some $\sigma_i$, $1 \leq i \leq
  t$, then $\cal R$ will return $\sigma_i$, otherwise she will return
  $\perp$.  

This simulation differs from the real one when the signature $\mu$ is
valid and was not obtained from the signing oracle. We note that the only ways to
create a valid confirmer signature without the help of $\cal R$ consist in  either encrypting
a digital signature obtained from the conversion
oracle or  coming up
with a new fresh pair of message and corresponding signature $(m,\mu)$. $\cal
R$ can handle the first case using her PCA oracle and list of
records $\cal L$. In the second case, we can distinguish two sub-cases: either
$m$ has not been queried to the signing oracle in which case the pair
$(m,\mu)$ corresponds to an existential forgery on the confirmer signature
scheme and thus to an existential forgery on the underlying digital scheme
according to \cite[Theorem 1]{CamenischMichels2000}, or $m$ has been queried to the
signing oracle but $\Gamma.\decrypt(\mu)$ is not an output of the selective
conversion oracle, which corresponds to a strong existential forgery on the
underlying digital signature. Therefore, the probability that this scenario does not happen is at least
$(1-\epsilon')^{q_{sc}}$ because the underlying digital signature scheme is
  $(t,\epsilon',q_s)$-SEUF-CMA secure by assumption.

\item \textbf{[$\{\confirm,\deny\}$ queries]} $\cal R$ will proceed exactly as in
  the selective conversion with the exception of simulating the denial
  protocol instead of returning $\perp$, or the confirmation protocol instead
  of returning the converted digital signature (the $\{\confirm,\deny\}$ protocols are concurrent zero-knowledge proofs, and thus they are simulatable). This simulation does not deviate from the standard execution of the protocols by at least $(1-\epsilon')^{q_v}$.

\item \textbf{[Challenge phase]} Eventually, $\cal A$ outputs two challenge messages $m_0$ and $m_1$. $\cal R$ will
then compute two signatures $\sigma_0$ and $\sigma_1$ on $m_0$  and $m_1$
respectively, which she gives to her own IND-PCA challenger. $\cal R$ receives then
the challenge $\mu^\star$, as the encryption of either $\sigma_0$ or
$\sigma_1$, which she will forward to $\cal A$.

\item \textbf{[Post challenge phase]} $\cal A$ will continue issuing queries to the
signing, confirmation/denial and selective conversion oracles and $\cal R$
can answer as previously. Note that in this phase, $\cal A$ is not allowed to
query the selective conversion or the confirmation/denial oracles on $(m_i,\mu^\star)$, $i=0,1$. Also, $\cal R$ is
not allowed to query her PCA oracle on $(\mu^\star,\sigma_i)$, $i=0,1$. 
Again, the probability to not deviate from the real invisibility game is at least $(1-\epsilon')^{q_{sc}+q_v}$.

\item \textbf{[Final output]} When $\cal A$ outputs his answer $b \in
\{0,1\}$, $\cal R$ will forward this very answer to her own challenger. Therefore
$\cal R$ will IND-PCA break the underlying encryption scheme with advantage at
least $\epsilon\cdot (1-\epsilon')^{(q_v+q_{sc})}$, in time at most $t +
q_sq_{sc}(q_v+q_{sc})$ after at most $q_{sc}(q_{sc} + q_v)$ queries to the PCA
oracle. 

\end{description}
\end{proof}

\begin{thm}[CtEaS paradigm]
\label{thm:INDPCASufficient-CtEaS}
Given $(t,q_s,q_v,q_{sc}) \in \mathbb{N}^4$ and $(\varepsilon,\epsilon') \in
[0,1]^2$, confirmer signatures from CtEaS are ($t,\epsilon,q_s,q_v,q_{sc}$)-INV-CMA
secure if they use a $(t,\epsilon',q_s)$-SEUF-CMA secure digital signature,
a statistically hiding and $(t,\epsilon_b)$ binding commitment, and a 
($t+q_sq_{sc}(q_{sc}+q_v),\frac{\epsilon}{2}\cdot
[(1-\epsilon')\cdot(1-\epsilon_b)]^{(q_{sc}+q_v)},q_{sc}(q_{sc} + q_v)$)-IND-PCA secure
encryption scheme.
\end{thm}

\begin{proof}

Let $\cal A$ be an attacker against the CtEaS construction. We construct an
attacker $\cal R$ against the underlying encryption:

\begin{description}
\item \textbf{[$\setup$ and $\keygen$].} $\cal R$ gets the parameters of the encryption scheme $\Gamma$ from her
challenger. Then she chooses a $(t,\epsilon',q_s)$-SEUF-CMA digital signature $\Sigma$ (along with a key
pair ($\Sigma.\pk,\Sigma.\sk$)) and a secure commitment $\Omega$.

\item \textbf{[$\sign$ and $\sconfirm$ queries].} $\R$ proceeds exactly like the standard algorithm/protocol, with the exception of maintaining, in case of $\sign$, in a list $\cal L$ the queried messages, the corresponding confirmer signatures, and the intermediate values used to produce these, for instance the random strings used to produce the commitments.


\item \textbf{[$\convert$ and $\{\confirm,\deny\}$ queries]} To convert an alleged signature $\mu_i=(c_i,e_i,\sigma_i)$  on a message $m_i$, $\R$ checks the validity of $\sigma_i$ on $c_i$; if it is invalid, then $\R$ proceeds as prescribed by the standard algorithm. Otherwise, $\R$ checks the list $\cal L$ for records corresponding to the queried message $m_i$ and where $c_i$ has been used as a commitment on $m_i$. If $e_i$ is found in one of these records as encryption of some $r_i$ concatenated with $m_i\|\Sigma.\pk$ ($r_i$ is the opening value of $c_i$), then $\R$ proceeds as dictated by the standard algorithm. Otherwise, $\R$ queries her PCA oracle on $e_i$ and on each opening value of $c_i$ found in these records (concatenated always with $m_i\|\Sigma.\pk$). $\cal R$ returns the opening value giving rise to a 'yes' response (by the PCA oracle), if any, otherwise she returns $\perp$.

\noindent Verification ($\{\confirm,\deny\}$) queries are handled similarly with the exception of simulating the denial protocol instead of returning $\perp$, and the confirmation protocol instead of converting the signature.

\noindent This simulation differs from the standard procedure when $\mu_i$ is valid, but $m_i$ has not been queried before, or $c_i$ has not been used to generate commitments on $m_i$. The first case corresponds to an existential forgery on the construction which translates into breaking the binding property of the commitment scheme if $c_i$ has been used a commitment on some message $m_j \neq m_i$, or to breaking the existential unforgeability of the underlying digital signature otherwise. The second case corresponds to an existential forgery on the underlying signature scheme. Both cases do not happen with probability at least $[(1-\epsilon')\cdot(1-\epsilon_b)]^{q_v+q_{sc}}$.

\item \textbf{[Challenge phase]} At some point, $\cal A$ outputs two messages $m_0,m_1$ to $\cal R$. The
latter chooses a random string $r$ from the corresponding space. $\cal R$ 
outputs to her challenger  the strings
$r\|m_0\|\Sigma.\pk$ and $r\|m_1\|\Sigma.\pk$. She receives then a ciphertext $e_{b}$, encryption of $r\|m_b\|\Sigma.\pk$, for some $b \in \{0,1\}$. 
To answer her challenger, $\cal R$ chooses a bit $b' \hasard \{0,1\}$, computes a commitment $c_{b'}$ on the message $m_{b'}$ using the string $r$. Then, $\cal R$  outputs $\mu=(c_{b'},e_{b},\Sigma.\sign_{\Sigma.\sk}(c_{b'}))$ as a challenge signature to $\cal A$. 

Two cases: either $\mu$ is a valid confirmer signature on $m_{b'}$ (if
$b=b'$), or it is not a valid signature on either $m_0$ or $m_1$. However, since the used commitment is statistically hiding, i.e. $c_{b'}$ reveals no information about $m_{b'}$, then $\mu$ is conform to a challenge signature in a real INV-CMA game.


\item \textbf{[Post challenge phase]} $\cal A$ continues to issue queries and
  $\cal R$ continues to handle them as before. Note that at this stage, $\cal R$ cannot request her PCA oracle on
$(e_{b},r\|m_i\|\Sigma.\pk)$, $i \in \{0,1\}$. $\cal R$
would need to make such a query if she gets a
verification (conversion) query on a signature $(c_i,e_b,\sigma_i) \neq \mu$ and the message
$m_i \in  \{m_0,m_1\}$. $\cal R$ will respond to such a query by running the denial protocol
(output $\perp$). This simulation differs from the real algorithm when
$(c_i,e_b,\sigma_i)$ is valid on $m_i$. Again, such a scenario won't happen with
probability at least $(1-\epsilon')^{q_v+q_{sc}}$, because the query would
form a strong existential forgery on the digital signature scheme underlying
the construction.

\item \textbf{[Final output]} The rest of the proof follows in a
  straightforward way. Let $b_a$ be the bit output by  $\cal A$. $\cal R$ will output $b'$ to her challenger in case $b'=b_a$ and $1-b'$ otherwise.

The advantage of $\cal A$ in such an attack is defined
by $\epsilon = \adv({\cal A}) = \left|\Pr[b_a=b'|b=b'] - \frac{1}{2}\right|$

We assume again without loss of generality that $\epsilon=\Pr[b_a=b'|b=b'] -
\frac{1}{2}$. The advantage of $\cal R$ by definition the product $p_{\sf sim}\cdot p_{\sf chal}$, where $p_{\sf sim}$ is the probability of providing a simulation indistinguishable from that in a real attack; it is equal to $[(1-\epsilon_b)\cdot(1-\epsilon')]^{q_v+q_{sc}}$. Whereas $p_{\sf chal}$ is the probability that $\cal R$ solves her challenge provided the simulation is correct:

\begin{eqnarray*}
\small
p_{\sf chal}  &=& \left [\Pr[b'=b_a,b=b' ] + \Pr[b' \neq b_a,b \neq b'] - \frac{1}{2} \right ]\\
            &=&  \frac{1}{2}\left [\Pr[b'=b_a|b=b'] + \Pr[b' \neq b_a|b \neq b'] - 1 \right ]\\
               &=& \left [\frac{1}{2}(\epsilon +
                 \frac{1}{2} + \frac{1}{2} - 1) \right ]\\
    &=& \frac{\epsilon}{2}\\
\end{eqnarray*}

Actually, $\Pr[b \neq b']=\Pr[b= b'] =
\frac{1}{2}$ as $b'\hasard\{0,1\}$. Moreover, if $b\neq
b'$, then the probability that $\cal A$ answers $1-b'$  is $\frac{1}{2}$ (since $\mu$ is invalid on either $m_0$ or $m_1$).
\end{description}
\qed \end{proof}

\subsection{An efficient variant of StE}
\label{subsec:efficientStE-CDCS}

One attempt to circumvent the problem of \emph{strong forgeability} of
constructions obtained from the plain StE paradigm
can be achieved by binding the digital signature to its encryption. In this way, from a
digital signature $\sigma$ and a message $m$, an adversary cannot create a new
confirmer signature on $m$ by just re-encrypting $\sigma$. In fact, $\sigma$ forms a digital signature on $m$ and some data, say $c$, which uniquely defines the confirmer signature on $m$. Moreover, this data $c$ has to be
public in order to issue the $\{\sconfirm,\confirm,\deny\}$ protocols. 

In this subsection, we propose a realization of this idea using hybrid encryption
(the KEM/DEM paradigm). We also allow more flexibility without compromising the overall
security by encrypting only one part of the signature and leaving
out the other part, provided it does not reveal information about the key nor
about the message.

\subsubsection{The construction}
\label{subsec:kemConstruction}
Let $\Sigma$ be a digital signature scheme given by $\Sigma.\keygen$, which generates 
a key pair ($\Sigma.\sk$, $\Sigma.\pk$), $\Sigma.\sign$, and $\Sigma.\verify$. Let further $\kem$ be a KEM given
by $\kem.\keygen$, which generates a key pair ($\kem.\pk$,
$\kem.\sk$), $\kem.\encap$, and $\kem.\decap$. Finally, we consider a DEM $\dem$ given
by $\dem.\encrypt$ and $\dem.\decrypt$.

We assume that any digital signature $\sigma$, generated
using $\Sigma$ on an arbitrary message $m$, can be efficiently transformed in a
reversible way to a pair $(s,r)$ where $r$ reveals no information about $m$
nor about $(\Sigma.\sk,\Sigma.\pk)$. I.e. there exists an
algorithm that inputs a message $m$ and a key pair $(\Sigma.\sk,\Sigma.\pk)$
and outputs a string statistically indistinguishable from
$r$, where the probability is taken over the messages and the key pairs considered by $\Sigma$. This technical detail will improve the efficiency of the construction as it will not necessitate
encrypting the entire signature $\sigma$, but only the message-key-dependent part,
namely $s$. Finally, we assume that $s$ belongs to the message space of $\dem$. 

We further assume that the encapsulations generated by $\kem$ are exactly $\kappa$-bit long, where $\kappa$ is a security parameter. This can be for example realized by padding with zeros, on the left of the most significant bit of the
given encapsulation, until the resulting string has length $\kappa$. Moreover, the
operator $\|$ denotes the usual concatenation operation between two bit-strings. As a result, the first bit of $m$ will always be at the $(\kappa+1)$-th position in $c\|m$, where $c$ is a given encapsulation.

\vspace{.5 cm}
 The construction of confirmer signatures from $\Sigma$, $\kem$, and $\dem$ is given as follows.
\begin{description}
\item \textbf{Key generation ($\keygen$). }Set the signer's key pair to
  $(\Sigma.\sk,\Sigma.\pk)$ and the confirmer's key pair to $(\kem.\sk,\kem.\pk)$.
\item \textbf{Signature ($\sign$). }Fix a key $k$ together with its encapsulation
  $c$. Then, compute a (digital) signature $\sigma = \Sigma.\sign_{\Sigma.\sk}(c\|m)=(s,r)$ on $c\|m$, and output $\mu= (c,\dem.\encrypt_{k}(s),r)$.
\item \textbf{Verification ($\{\sconfirm,\confirm,\deny\}$). }To confirm (deny) a purported
  signature $\mu=(c,e,r)$, issued on a certain message $m$, the prover (either the signer or the confirmer) provides a concurrent zero-knowledge proof of knowledge of the decryption of $(c,e)$, which together with $r$ forms a valid (invalid) signature on $c\|m$. Providing such a proof is possible since the underlying statement is an NP language \cite{DworkNaorSahai2004}.
\item \textbf{Selective conversion ($\convert$). }To convert a signature
  $\mu=(c,e,r)$ issued on a certain message $m$, the confirmer first checks
  its validity. If it is invalid, he outputs $\perp$, otherwise, he computes
  $k=\kem.\decap_{\kem.\sk}(c)$ and outputs
  $(\dem.\decrypt_{k}(e),r)$.
\end{description}
     
\begin{thm}
\label{thm:newStE-forgery-CDCS}
Given $(t,q_s) \in \mathbb{N}^2$ and $\varepsilon \in [0,1]$, the above construction is ($t,\epsilon,q_s$)-EUF-CMA secure if 
the underlying digital signature scheme is ($t,\epsilon,q_s$)-EUF-CMA secure.
\end{thm}
\begin{proof}
Let $\cal A$ be an attacker that ($t,\epsilon,q_s)$-EUF-CMA breaks the above construction. The algorithm $\cal R$ ($t,\epsilon,q_s$)-EUF-CMA breaks the underlying digital signature scheme $\Sigma$ as follows:
\begin{description}
\item \textbf{[Key generation]} $\cal R$ gets the parameters of $\Sigma$ from her challenger. Then she chooses an appropriate
  KEM $\kem$ and DEM $\dem$ and asks $\cal A$ to provide her with the
  confirmer key pair $(\kem.\sk,\kem.\pk)$. Finally, $\cal R$ fixes the above parameters as a setting for the confirmer signature scheme $\cal A$ is trying to attack. 
\item \textbf{[Signature queries]} For a signature query on a message $m$,
  $\cal R$ computes an encapsulation $c$ together with its
  decapsulation $k$ (using $\Gamma.\pk$). Then, she requests her challenger for a digital signature $\sigma=(s,r)$
  on $c\|m$. Finally, she encrypts $s$ in $\dem.\encrypt_k(s)$ and outputs the
  confirmer signature $(c,\dem.\encrypt_k(s),r)$. 
\item \textbf{[Final Output]} Once $\cal A$ outputs his forgery
  $\mu^\star=(c^\star,e^\star,r^\star)$ on $m^\star$. $\cal R$
  will compute the decapsulation of $c^\star$, say $k$. If $\mu^\star$ is valid then by
definition $(\dem.\decrypt_k(e^\star),r^\star)$ is a valid digital
signature on $c^\star\|m^\star$. Thus, $\cal R$ outputs
$(\dem.\decrypt_k(e^\star),r^\star)$ and $c^\star\|m^\star$ as a
valid existential forgery on $\Sigma$. In fact, if, during a query made
by $\cal A$ on a message $m^i$, $\cal R$ is compelled to query her
own challenger for a digital signature on $c^\star\|m^\star =
c^i\|m^i$, then $m^\star = m^i$ (by construction),
which contradicts the fact that $(\mu^\star,m^\star)$ is an existential
forgery output by $\cal A$.
\end{description}

\end{proof}

The following remark is vital for the invisibility of the resulting confirmer
signatures.
\begin{remark}
\label{rmq:strongforgery}
The previous theorem shows that existential unforgeability of the underlying
digital signature scheme suffices to ensure existential unforgeability of the
resulting construction. Actually, one can also show that this requirement on
the digital signature (EUF-CMA security) guarantees that no adversary,
against the construction, can come up with a valid confirmer signature
$\mu=(c,e,r)$ ($c$ is the encapsulation used to generate the confirmer
signature $\mu$) on a message $m$ that has been queried before to the signing
oracle but where $c$ was never used to generate answers (confirmer signatures) to the signature queries.

\noindent To prove this claim, we construct from such an adversary, say $\cal A$, an EUF-CMA
adversary $\cal R$ against the underlying digital signature scheme, which runs in the
same time and has the same advantage as $\cal A$. In fact, $\cal R$ will simulate $\cal
A$'s environment in the same way described in the proof of Theorem
\ref{thm:newStE-forgery-CDCS}. When $\cal A$ outputs his forgery
$\mu^\star=(c^\star,e^\star,r^\star)$ on a message $m_i$ that has been previously
queried to the signing oracle, $\cal R$ decrypts $(c^\star,e^\star)$ in
$s^\star$, which by definition forms, together with $r^\star$, a valid digital signature on
$c^\star\|m_i$. Since by assumption $c^\star$ was never used to generate
confirmer signatures on the queried messages, $\cal R$ never invoked her own
challenger for a digital signature on $c^\star\|m_i$. Therefore, $(s^\star,r^\star)$ will form a valid
existential forgery on the underlying digital signature scheme.
\end{remark}

In the rest of this subsection, we show that the new StE paradigm achieves a stronger notion (than INV-CMA) of invisibility that we denote SINV-CMA. This notion was first introduced in \cite{GalbraithMao2003}, and it captures the difficulty to distinguish confirmer signatures on an adversarially chosen message from random elements in the confirmer signature space. Again, the difference with the previously mentioned notions lies in the challenge phase where the SINV-CMA attacker outputs a message $m^\star$ and receives in return an element $\mu^\star$ which is either a confirmer signature on  $m^\star$ or a random element from the confirmer signature space. There is again the natural restriction of not querying the challenge pair to the $\sconfirm$, $\{\confirm,\deny\}$, and $\convert$ oracles. Similarly, a CDCS scheme is
$(t,\epsilon,q_s,q_v,q_{sc})$-SINV-CMA secure if no adversary operating in time
$t$, issuing $q_s$ queries to the signing oracle (followed potentially by queries to the $\sconfirm$ oracle), $q_v$ queries to the confirmation/denial oracles and $q_{sc}$ queries to the selective conversion
oracle that wins the game defined in Experiment  $\mathbf{Exp}^{\mathsf{SINV\mbox{-}\CMA}}(1^\kappa)$ (in Figure \ref{fig:sinv-cma4CS}) with advantage greater that $\epsilon$. The probability is taken over all the coin tosses. 

\begin{figure}
\label{fig:sinv-cma4CS}
\begin{footnotesize}
\begin{experiment}[Experiment $\mathbf{Exp}_{\cs,\A}^{\mathsf{SINV\mbox{-}\CMA}}(1^\kappa)$]

\item $\param \gets \cs.\setup(1^\kappa)$; 
\item $(\pks,\sks) \gets \cs.\keygen_{\entfont{S}}(1^\kappa)$; \\$(\pkc,\skc) \gets \cs.\keygen_\C(1^\kappa)$; 

\item $(m^\star,\mathcal{I}) \leftarrow {\mathcal A}^{\mathfrak{S}, \mathfrak{Cv}, \mathfrak{V}} (\text{find},\pks,\pkc)$ \\
\phantom{$(m^{\star},\mathcal{I}) $} $\left\vert
\begin{array}{l} 
\mathfrak{S} : m  \longmapsto \cs.\sign_{\sks}(m,\pkc) \\ 
\mathfrak{Cv}: (\mu,m) \longmapsto \cs.\convert_{\skc}(\mu,m) \\
\mathfrak{V} : (\mu,m) \longmapsto \cs.\{\sconfirm,\confirm,\deny\}(\mu,m) \\
\end{array} \right.$ \\
\item $\mu_1^{\star} \leftarrow \cs.\sign_{\sks}(m^\star,\pkc)$
\item $\mu_0^{\star} \hasard \cs.\sf{space}$; $b \hasard \{0,1\}$
\item $b^\star\leftarrow {\mathcal A}^{\mathfrak{S}, \mathfrak{Cv}, \mathfrak{V}}(\textsf{guess},\mathcal{I},\mu_b^{\star},\pks,\pkc)$~\\
\phantom{$b^\star$} $\left\vert
\begin{array}{l} 
\mathfrak{S} : m  \longmapsto \cs.\sign_{\{\sks,\pkc\}}(m) \\ 
\mathfrak{Cv}: (\mu,m) (\neq (\mu_b^\star,m^\star)) \longmapsto \cs.\convert_{\skc}(\mu,m) \\
\mathfrak{V} : (\mu,m) (\neq (\mu_b^\star,m^\star))\longmapsto \cs.\{\sconfirm,\confirm,\deny\}(\mu,m) \\
\end{array} \right.$ 

\item If $(b=b^\star)$ return $1$ else return $0$.

\end{experiment}
\end{footnotesize}
\caption{Strong invisibility in confirmer signatures}
\end{figure}

\begin{thm}
\label{thm:KEM-invisibility}
Given $(t,q_s,q_v,q_{sc}) \in \mathbb{N}^4$ and $(\varepsilon,\epsilon') \in
[0,1]^2$, the construction proposed above is ($t,\epsilon,q_s,q_v,q_{sc}$)-SINV-CMA
secure if it uses a $(t,\epsilon',q_s)$-SEUF-CMA secure digital signature, an ($t,\epsilon"$)-INV-OT secure DEM with injective encryption, and a ($t+q_s(q_v+q_{sc}),\epsilon\cdot (1-\epsilon")\cdot(1-\epsilon')^{q_v+q_{sc}}$)-IND-CPA secure KEM.
\end{thm}
\begin{proof}
Let $\cal A$ be an attacker that ($t,\epsilon,q_s,q_v,q_{sc}$)-SINV-CMA breaks our
construction. We construct an algorithm $\cal R$ that breaks the underlying KEM as follows.


\begin{description}
\item \textbf{[$\keygen$]} $\cal R$ gets the parameters
  of the KEM $\kem$ from her challenger. Then, she
  chooses an appropriate ($t,\epsilon"$)-INV-OT secure DEM $\cal D$ together with a
  $(t,\epsilon',q_s)$-SEUF-CMA secure signature scheme $\Sigma$. $\R$ further generates a key pair $(\Sigma.\sk,\Sigma.\pk)$ for $\Sigma$ and sets it as the signer's key pair.
\item \textbf{[$\sign$ and $\sconfirm$ queries]} For a signature query, $\R$ proceeds like the standard algorithm. She further maintains a list $\cal L$ of the encapsulations $c$ and keys $k$ used to generate the confirmer signatures. 
 \item \textbf{[$\{\confirm,\deny\}$ queries]} For a verification query on a signature $\mu=(c,e,r)$ and a message $m$, $\R$  looks up the list $\cal L$ for the decapsulation of $c$, which once found, allows $\R$ to check the validity of the signature and therefore simulate correctly the suitable protocol (confirmation or denial). If $c$ has not been used to generate confirmer signatures, then $\R$ will run the denial protocol. Note that $\R$ is able to perfectly simulate to confirmation/denial protocol since these  are by definition concurrent zero-knowledge proofs of knowledge.

This simulation differs from the real one when the signature  $\mu=(c,e,r)$ on
  $m$ is valid, but $c$ does not appear in any record of $\cal L$. We
  distinguish two cases: either $m$ was never queried to the signing
  oracle, then $(m,\mu)$ would correspond to an existential forgery on the
  confirmer signature (and thus to an existential forgery on $\Sigma$), or $m$ has been
  previously queried to the signing oracle in which case  $(m,\mu)$ would
  correspond to an existential forgery on $\Sigma$ thanks to Remark \ref{rmq:strongforgery}. Hence,
  the probability that both scenarios do not happen is at least $ (1-\epsilon')^{q_{v}}$ ($\Sigma$ is $(t,\epsilon',q_s)$-SEUF-CMA secure).
 \item \textbf{[$\convert$ queries]} Conversion queries are treated like verification queries with the exception of converting the signature instead of running $\confirm$, and issuing $\perp$ instead of running $\deny$. Similarly, this
   simulation does not differ from the real execution of the algorithm with probability at
   least $ (1-\epsilon')^{q_{sc}}$.
\item  \textbf{[Challenge]} Eventually, $\cal A$ outputs a challenge message $m^{\star}$. $\cal R$ uses her challenge $(c^\star,k^\star)$ to compute a digital signature
$(s^\star,r^\star)$ on $c^\star\|m^\star$. Then, she encrypts $s^\star$ in $e^\star$ using
$\dem.\encrypt_{k^\star}$  and outputs $\mu^\star=(c^\star,e^\star,r^\star)$ to $\cal
A$. Therefore, $\mu^\star$ is either a valid confirmer signature on $m^\star$ or an element
indistinguishable from a random element in the confirmer signature space ($k^\star$ is random and $\dem$ is INV-OT secure, moreover $r^\star$ is information-theoretically independent from $m$ and $\Sigma.\pk$). If $\mu^\star$,
 in the latter case, is a random element in the confirmer signature
space, then this complies with the scenario of a real attack. Otherwise, if
$\mu^\star$ is \emph{only indistinguishable from random}, then if the
advantage of $\cal A$ is non-negligibly different from the advantage of an invisibility adversary in a real attack, then $\cal A$ can be easily used to ($t,\epsilon''$)-INV-OT break $\dem$. To sum up, under the INV-OT assumption of $\dem$,
the challenge confirmer signature $\mu^\star$ is either a valid confirmer
signature on $m^\star$ or a random element in the confirmer signature space.

\item  \textbf{[Post challenge phase]} $\cal A$ will continue issuing queries to the previous oracles, and $\cal R$
can answer as previously. Note that in this phase, $\cal A$ might request the
verification or conversion of a confirmer signature
$(c^\star,-,-)$ on a message $m_i \neq m^\star$. According to the previous analysis, such a signature is invalid  w.r.t. $m_i$ with probability $ (1-\epsilon')^{q_{sc}+q_v}$. 

In case the verification/conversion query involves $m^\star$ and $c^\star$, then let $(c^\star,\tilde{e},\tilde{r}) \neq \mu^\star$ be the queried signature. We have $(\tilde{e},\tilde{r}) \neq (e^\star,r^\star)$. Two cases manifest. Either $r^\star = \tilde{r}$, in which case $(c^\star,\tilde{e},\tilde{r})$ is invalid w.r.t. $m^\star$ since otherwise $e^\star$ and $\tilde{e}$ will be two different ciphertexts that decrypt to $s^\star$, which is impossible since $\dem$ is by assumption a DEM with injective encryption. Or $r^\star \neq \tilde{r}$; therefore for $(c^\star,\tilde{e},\tilde{r})$ to be valid w.r.t. $m^\star$, $(\dem.\decrypt(\tilde{e}),\tilde{r})$ ($\neq (s^\star,r^\star)$) must be a valid digital signature on $c^\star\|m^\star$. Therefore, this latter scenario does not happen with probability at least $(1-\epsilon')$ since $\Sigma$ is SEUF-CMA secure.

Bottom line is, whenever a verification/conversion query involves $c^\star$ or an encapsulation $c$ that is not in the list $\cal L$, $\cal R$ will issue the denial protocol in case of a verification query, or the
symbol $\perp$ in case of a conversion query. The probability that the simulation does not differ from the real
execution is at least $ (1-\epsilon')^{q_{sc}+q_v}$.

\item \textbf{[Final output]} When $\cal A$ outputs his answer $b \in
\{0,1\}$, $\cal R$ will forward this answer to her own challenger. Therefore $\cal R$ will $(t+q_s(q_v+q_{sc}),\epsilon\cdot (1-\epsilon")\cdot(1-\epsilon')^{q_v+q_{sc}})$-IND-CPA break $\kem$.
\end{description}

 \end{proof}

\subsection{An efficient variant of CtEaS}
\label{subsec:efficientCtEaS-CDCS}
We remediate to the strong forgeability problem of CtEaS by using the same trick applied to StE, namely bind the used digital signature to the corresponding confirmer signature. This is achieved by producing a digital
signature on both the commitment and the encryption of the random string used to generate it. In this way, the attack discussed in Subsection \ref{subsec:invisibilityBreach} no longer applies, since an adversary will need to produce a digital signature on the commitment and the re-encryption of the
random string used in it. Note that such a fix already appears in the construction of \cite{GentryMolnarRamzan2005}, however, it was not exploitable as the invisibility was considered in the insider model.

\subsubsection{Commit\_then\_Encrypt\_then\_Sign: CtEtS}
Let $\Sigma$ be a signature scheme given by $\Sigma.\keygen$, that generates
$(\Sigma.\pk,\Sigma.\sk)$, $\Sigma.\sign$, and $\Sigma.\verify$. Let further
$\Gamma$ denote an encryption scheme given by $\Gamma.\keygen$, that generates
$(\Gamma.\pk,\Gamma.\sk)$, $\Gamma.\encrypt$, and $\Gamma.\decrypt$. Finally let
$\Omega$ denote a commitment scheme given by $\Omega.\commit$ and
$\Omega.\open$. We assume that $\Gamma$ produces ciphertexts of length exactly
a certain $\kappa$. As a result, the first bit of $c$ will always be at
the $(\kappa+1)$-th position in $e\|c$, where $e$ is an encryption produced by $\Gamma$. The construction from the CtEtS paradigm is as follows.
\begin{description}
\item \textbf{Key generation ($\keygen$). }The signer's key pair is $(\Sigma.\pk,\Sigma.\sk)$ and the confirmer's key pair is $(\Gamma.\pk,\Gamma.\sk)$.

\item \textbf{Signature ($\sign$). } To sign a message $m$, first produce a commitment $c$ on $m$ using a random string $r$, encrypt this string in $e$, and then produce a digital signature $\sigma=\Sigma.\sign_{\Sigma.\sk}(e\|c)$. The confirmer signature on $m$ is $\mu=(c,e,\sigma)$.

\item \textbf{Verification ($\{\sconfirm,\confirm,\deny\}$). }To verify  a signature on a given message, the prover (either the signer or the confirmer) provides a concurrent zero-knowledge proof of the underlying (NP) statement. 

\item \textbf{Selective conversion ($\convert$). } Conversion of a valid signature $\mu=(c,e,\sigma)$ is done by decrypting $e$.
\end{description}

It is clear that this new construction looses the parallelism of the original one, i.e. encryption and signature can no longer be carried out in parallel, however, it has the advantage of resting on cheaper encryption as we will show in the following. Moreover, it still preserves the merit of the CtEaS paradigm, namely possibility to instantiate with \emph{any} digital signature scheme \footnote{Practical realizations from the StE paradigm need to use digital signatures from a special class that we specify in Definition \ref{def:signatureClass}}.

\begin{thm}
\label{thm:newCtEaS-forgery-CDCS}
Given $(t,q_s) \in \mathbb{N}^2$ and $(\varepsilon,\epsilon') \in
[0,1]^2$, the construction depicted above is
$(t,\epsilon\cdot (1-\epsilon')^{q_s},q_s)$-EUF-CMA secure if it uses a $(t,\epsilon')$-binding commitment and a $(t,\epsilon,q_s)$-EUF-CMA secure
digital signature scheme.
\end{thm}
\begin{proof}

Let $\cal A$ be an EUF-CMA attacker against the construction. We construct
an EUF-CMA attacker $\cal R$ against the underlying digital signature scheme as follows.

$\cal R$ gets the parameters of the digital signature from her attacker, and
chooses suitable encryption and commitment schemes. After she gets the confirmer's key pair from $\A$, $\R$ can perfectly simulate signature queries using her own challenger. At some point, $\cal A$ will output a forgery
$\mu^\star=(c^\star,e^\star,\sigma^\star)$ on some message $m^\star$, which
was never queried before for signature. By definition, $\sigma^\star$ is a valid
digital signature on $e^\star\|c^\star$. Suppose there exists $1 \leq i \leq q_s$ such that $e^\star\|c^\star = e_i\|c_i$
where $\mu_i=(c_i,e_i,\sigma_i)$ was the output confirmer signature on the
query $m_i$. Due to the special way the strings $e_i\|c_i$ are created, this implies $(e_i,c_i)=(e^\star,c^\star)$. With probability at least $(1-\epsilon')$, we have $m_i = m^\star$ (the commitment is $(t,\epsilon')$-binding), which is a contradiction. Therefore,  $e^\star\|c^\star$ was not queried by $\cal R$ for a digital signature with probability at least $(1-\epsilon')^{q_s}$. $\cal R$ outputs to her challenger the EUF-CMA forgery  $\sigma^\star$ and $e^\star\|c^\star$. 
\qed \end{proof}
For the invisibility proof, we first need this lemma:
\begin{lemma}
\label{lemma:commitEncrypt}
Let $\Omega$ and $\Gamma$ be a commitment and a public key encryption
scheme respectively. We consider the following game between an adversary
$\cal A$ and his challenger $\cal R$:
\begin{enumerate}
\item $\cal R$ invokes the algorithms $\Gamma.\keygen(1^\kappa)$ to generate
  $(\pk,\sk)$, where $\kappa$ is a security parameter.
\item $\cal A$ outputs two messages $m_{0}$ and $m_{1}$ such that
  $m_{0} \neq m_{1}$ to his
  challenger.
\item ${\cal R}$ generates two  nonces $r_{0}$ and $r_{1}$ such
  that $r_{0} \neq r_{1}$. Next, he chooses two bits $b,b' \hasard \{0,1\}$ uniformly at
  random. Finally, he outputs to $\cal A$
  $c_{b}=\Omega.\commit(m_{b},r_{1-b'})$ and
  $e_{b'}=\Gamma.\encrypt_{\pk}(r_{b'})$.
\item $\cal A$ outputs a bit $b_{a}$ representing his guess of $c_{b}$
  not being the commitment of $m_{b}$ using the nonce
  $\Gamma.\decrypt(e_{b'})$. $\cal A$ wins the game if $b_{a} \neq b$,
  and we define his advantage 
$$
\adv({\cal A}) = \left|\Pr[b\neq b_{a}] - \frac{1}{2}\right|,
$$
where the probability is taken over the random tosses of both $\cal A$
and $\cal R$.
\end{enumerate}
If $\Omega$ is injective, $(t,\epsilon_b)$-binding, and $(t,\epsilon_{h})$-hiding, then $\adv({\cal A})$ in the above game
is at most $\frac{\epsilon_{h}}{1-\epsilon_b}$.
\end{lemma}
\begin{proof}
Let $\epsilon$ be the advantage of $\cal A$ in the game above. We will
construct an adversary $\cal R$ which breaks the hiding property of
the used commitment with advantage $\epsilon\cdot(1-\epsilon_b)$.
\begin{itemize}
\item $\cal R$ gets from $\cal A$ the messages $m_{0},m_{1}$, and
  forwards them to her own challenger.
\item $\cal R$ receives from her challenger the commitment
  $c_{b}=\Omega.\commit(m_{b},r)$ for some $b \hasard \{0,1\}$ and
  some nonce $r$.
\item $\cal R$ generates a nonce $r'$ and outputs to $\cal A$
$c_{b}$ and $e=\Gamma.\encrypt_{\pk}(r')$.
\item When $\cal A$ outputs a bit $b_{a}$, $\cal R$ outputs to her
  challenger $1-b_{a}$. 
\end{itemize}
If $\cal A$ can by some means get hold of $r'$, then he can compute
$c_{i}=\Omega.\commit(m_{i},r')$, $i=0,1$. Since $\Omega$ is
injective and binding, then $c_{b} \neq \Omega.\commit(m_{b},r') $
and $c_{b} \neq \Omega.\commit(m_{1-b},r')$ respectively, i.e. $c_{b} \notin\{c_{0},c_{1}\}$. Thus, $\A$ will get no information on the message underlying $c_b$ even if he manages to invert $e$.

We have by definition: 
\begin{eqnarray*}
\adv({\cal R}) &=& (1-\epsilon_b)\left|\Pr[1-b_{a} = b] - \frac{1}{2}\right|\\
 &=&  (1-\epsilon_b) \left|\Pr[b_{a}\neq b] - \frac{1}{2} \right|\\
&=& \epsilon\cdot(1-\epsilon_b)
\end{eqnarray*}
\qed\end{proof}
\begin{remark}
Note that the above lemma holds true regardless of the used
encryption $\Gamma$. For instance, it can be used with encryption schemes that do not require any kind of security. 
\end{remark}

\begin{thm}
\label{thm:newCtEaS-invisibility-CDCS}
Given $(t,q_s,q_v,q_{sc}) \in \mathbb{N}^4$ and $(\varepsilon,\epsilon') \in
[0,1]^2$, the construction depicted above is
$(t,\epsilon,q_s,q_v,q_{sc})$-INV-CMA secure if it uses a
$(t,\epsilon',q_s)$-SEUF-CMA secure digital signature, an injective,
$(t,\epsilon_b)$-binding, and ($t,\epsilon_h$)-hiding commitment, and a
$(t+q_s(q_v+q_{sc}),\frac{1}{2}(\epsilon+\frac{\epsilon_h}{1-\epsilon_b})(1-\frac{e_h}{1-\epsilon_b})\cdot[(1-\epsilon_b)\cdot(1-\epsilon')]^{q_v+q_{sc}})$-IND-CPA secure encryption scheme.
\end{thm}
\begin{proof}(Sketch)

Let $\Sigma$, $\Gamma$, and $\Omega$ be the signature, encryption, and commitment schemes resp. underlying the construction. Let further $\R$ be the reduction using the invisibility attacker $\A$ in order to break $\Gamma$.

$\cal R$ gets the public key of $\Gamma$ from her challenger. She further generates the parameters of $\Sigma$ (for instance ($\Sigma.\sk$,$\Sigma.\pk$)) and of $\Omega$.
\begin{description}
\item \textbf{Pre-challenge phase.} 
Simulation of $\sign$ and $\sconfirm$ queries is done as dictated by the standard algorithm/\-protocol, with the exception of maintaining a list $\cal L$ of the strings used to produce commitments on the queried messages in addition to their encryptions. 

\noindent For a verification (conversion) query, $\R$ looks up the list $\cal L$ for the decryption of the first component of the signature; if it is found, $\R$ simulates the confirmation protocol (issues the converted signature in case of a conversion query), otherwise she simulates the denial protocol (issues the symbol $\perp$ in case of a conversion). The difference between this simulation and the real execution of the
algorithm manifests when a queried signature, say $(c_i,e_i,\sigma_i)$, is valid, on the queried message $m_i$, but $e_i$ is not present in the list. We
distinguish two cases, either the underlying message $m_i$ has been queried
previously on not. In the latter case, such a signature would correspond to an
existential forgery on the construction, thus, to an existential forgery on $\Sigma$ or to breaking the binding property of $\Omega$. In the former case, let $(c_j,e_j,\sigma_j)$ be
the output signature to $\cal A$ on the message $m_i$. We have
$e_i\|c_i \neq e_j\|c_j$ since $e_i \neq e_j$, and both $e_i$ and $e_j$ are the
$n$-bit prefixes of $e_i\|c_i$ and $e_j\|c_j$ resp. We conclude that the adversary would have to compute a digital signature on a string for which he had never  obtained a signature. Thus, the query would lead to an existential forgery on $\Sigma$. 

\noindent Bottom line is, the probability that the provided simulation does not deviate from the real execution is at
least $[(1-\epsilon')\cdot(1-\epsilon_b)]^{q_v+q_{sc}}$.

\item \textbf{Challenge phase.} At some point, $\cal A$ outputs two messages $m_0,m_1$ to $\cal R$. The latter chooses  two different random strings $r_0$ and $r_1$ and hands them to her challenger. $\cal R$ receives then a ciphertext $e_{b'}$, encryption of $r_{b'}$, for some $b' \in \{0,1\}$. To answer her challenger, $\cal R$  computes a commitment $c_b$ on the message $m_b$ for some $b \hasard \{0,1\}$ using the string $r_{b}$, then outputs $\mu=(c_{b},e_{b'},\Sigma.\sign_{\Sigma.\sk}(e_{b'}\|c_b))$ as a challenge confirmer signature to $\cal A$. Two cases: either $\mu$ is a valid confirmer signature on $m_b$ (if
$b'=b$), or it is not a valid signature on either $m_0$ or $m_1$. $\cal A$ cannot tell the difference between the provided challenge and that in a real attack with probability at least $1-\frac{\epsilon_h}{1-\epsilon_b}$ according to Lemma \ref{lemma:commitEncrypt}.

\item \textbf{Post-challenge phase.} $\cal A$ continues to issue queries and
  $\cal R$ continues to handle them as before. Note that in this phase, $\cal R$ might get a
verification (conversion) query on a signature $(c_b,e_b',-) \neq \mu$ and the message
$m_b$. $\cal R$ will respond to such a query by running the denial protocol
(output $\perp$). This simulation differs from the real algorithm when
$(c_b,e_b',-)$ is valid on $m_b$. Again, such a scenario won't happen with
probability at least  $[(1-\epsilon')\cdot(1-\epsilon_b)]^{q_v+q_{sc}}$, because the query would
form a strong existential forgery on $\Sigma$.

\item \textbf{Final output.} Let $b_a$ be the bit output by  $\cal A$. $\cal R$ will output $b$
 to her challenger in case $b=b_a$ and $1-b$ otherwise.

\end{description}
\noindent The advantage of $\cal A$ in such an attack is defined
by $\epsilon = \adv({\cal A}) = \left|\Pr[b_a=b|b'=b] - \frac{1}{2}\right|$.

\noindent We assume again without loss of generality that $\epsilon=\Pr[b_a=b|b'=b] -
\frac{1}{2}$. The advantage
of $\cal R$ is then given by the product $p_{\sf sim}\cdot p_{\sf chal}$, where $p_{\sf sim}$ is the probability of providing a simulation indistinguishable from that in a real attack; it is equal to $(1-\frac{e_h}{1-\epsilon_b})\cdot[(1-\epsilon_b)\cdot(1-\epsilon')]^{q_v+q_{sc}}$. Whereas $p_{\sf chal}$ is the probability that $\cal R$ solves her challenge provided the simulation is correct:

\begin{eqnarray*}
p_{\sf chal} &=& \left [\Pr[b=b_a,b'=b ] +
  \Pr[b \neq b_a,b' \neq b] - \frac{1}{2} \right ]\\
               &=&  \frac{1}{2}\left [\Pr[b=b_a|b'=b
                 ] + \Pr[b \neq b_a|b' \neq b] - 1 \right ]\\
               &=& \frac{1}{2}\left [(\epsilon +
                 \frac{1}{2}) + (\frac{\epsilon_{h}}{1-\epsilon_b}+\frac{1}{2}) - 1 \right ]\\
&=& \frac{1}{2}(\epsilon+\frac{\epsilon_h}{1-\epsilon_b})
\end{eqnarray*}

In fact, $\Pr[b' \neq b]=\Pr[b'= b] =
\frac{1}{2}$ as $b \hasard\{0,1\}$. Moreover, if $b'\neq
b$, then the probability that $\cal A$ answers $1-b$  is $\frac{1}{2}$ greater
than the advantage of the adversary in the game defined in Lemma
\ref{lemma:commitEncrypt}, namely $\frac{\epsilon_{h}}{1-\epsilon_b}$ .

\qed \end{proof}

\section{Practical Realizations of CDCS}
\label{sec:realizations-CDCS}
In this section, we provide practical realizations of confirmer signatures from StE, CtEtS, and EtS. We first introduce some classes of basic primitives that constitute important building blocks for the mentioned constructions. Then, we proceed to the description of our concrete instantiations of the paradigms.


\subsection{The class $\bbbs$ of signatures}
\begin{definition}
\label{def:signatureClass}
$\bbbs$ is the set of all digital signatures for which there exists a pair
of efficient algorithms, $\convert$ and $\retrieve$, where $\convert$ inputs a public
key $\pk$, a message $m$, and a valid signature $\sigma$ on $m$
(according to $\pk$) and
outputs the pair $(s,r)$ such that:
\begin{enumerate}
\item $r$ reveals no information about $m$ nor about $\pk$, i.e. there exists an algorithm $\mathsf{simulate}$ such that for every public key $\pk$ from the key space
  and for every message $m$ from the message space, the output $\mathsf{simulate}(\pk,m)$ is statistically
  indistinguishable from $r$.
\item there exists an algorithm $\compute$ that on the input $\pk$, the
  message $m$ and $r$, computes a description of a function $f: (\bbbg,\ast) \rightarrow (\bbbh,\circ_s)$:
\begin{itemize}
\item where $(\bbbg,\ast)$ is a group and $\bbbh$ is a set equipped with the binary operation $\circ_s$ ,
\item $\forall S,S' \in \bbbg$: $f(S \ast S')=f(S) \circ_s f(S')$.

\end{itemize}
and an $I \in \bbbh$, such that $f(s) = I$.
\end{enumerate}
and an algorithm $\retrieve$  that inputs $\pk$, $m$ and the correctly
converted pair $(s,r)$ and retrieves\footnote{Note that the $\retrieve$ algorithm suffices to ensure the non-triviality of the map $f$; given a pair $(s,r)$ satisfying the conditions described in the definition, one can efficiently recover the original signature on the message.} the signature $\sigma$ on $m$.
\end{definition}

The class $\bbbs$ differs from the class $\bbbc$, introduced in \cite{ShahandashtiSafavi-Naini2008}, in the condition required for the function $f$. In fact, in our description of $\bbbs$, the
function $f$ should satisfy a homomorphic property, whereas in the class $\bbbc$,
$f$ should only possess an efficient zero knowledge protocol for proving
knowledge of a preimage of a value in its range. We show in Theorem
\ref{thm:PKsignatures} that signatures in $\bbbs$ accept also efficient
ZK proofs for proving knowledge of preimages, and thus belong
to the class $\bbbc$. Conversely, one can claim that signatures in
$\bbbc$ are also in $\bbbs$, at least from a practical point of view, since it is
not known in general how to achieve efficient ZK protocols for proving knowledge of preimages of $f$ without having the latter
item satisfy some homomorphic properties. It is worth noting that similar to
the classes $\bbbs$ and $\bbbc$ is the class of signatures introduced in
\cite{GoldwasserWaisbard2004}, where the condition of having an efficient
 ZK protocol for proving knowledge of preimages is weakened to having
only a \emph{witness hiding} proof of knowledge. Again, although this is a
weaker assumption on $f$, all illustrations of signatures in this wider
class happen to be also in $\bbbc$ and $\bbbs$. Our resort to specify the homomorphic
property on $f$ will be justified later when describing the
confirmation/denial protocols of the resulting construction. In fact, these
protocols are concurrent composition of proofs and therefore need
a careful study as it is known that zero knowledge is not closed under concurrent
composition. Besides, the class $\bbbs$ encompasses most proposals that were
suggested so far, e.g. \cite{BellareRogaway1996,Schnorr1991,GennaroHaleviRabin1999,BonehLynnShacham2004,PointchevalStern2000,CramerShoup2000,CamenischLysyanskaya2002,CamenischLysyanskaya2004,BonehBoyen2004,ZhangSafaviNainiSusilo2004,Waters2005}. The reason why $\bbbs$ includes most digital signature schemes lies in
the fact that a signature verification consists in applying a function $f$ to
the ``vital'' part of the signature in question, then comparing the result to
an expression computed from the message underlying the signature, the
``auxiliary'' or ``simulatable'' part of the signature, and finally the public
parameters of the signature scheme. The function $f$ need not be one-way,
however the signature scheme would be trivially forgeable if it is not the case. Moreover, it ($f$) consists most of the time of an arithmetic operation (e.g. exponentiation, raising
to a power, pairing computation) which satisfies an easy homomorphic property.

\begin{figure*}
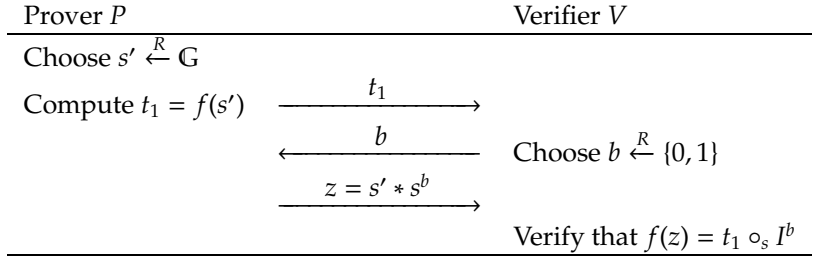

\begin{center}
\begin{tabular}{lcl}
Prover $P$& & Verifier $V$ \\
\hline

 Choose $s' \hasard \bbbg$ & & \\
 Compute $ t_1=f(s')$  & $\xrightarrow{\makebox[25mm]{$t_1$}}$ & \\
 
&  $\xleftarrow{\makebox[25mm]{$b$}}$  &  Choose $b \hasard \{0,1\}$\\
 
 &  $\xrightarrow{\makebox[25mm]{$z=s'\ast s^b$}}$  &  \\
 & & Verify that $f(z)=t_1\circ_s I^b$ \\
\hline
\end{tabular}

\caption{\label{fig:PKsignature} Proof of membership to the language
  $\{s \colon f(s) = I\}$ $\textbf{Common input: } I$ and $\textbf{Private
  input}: s$ }
\end{center}
\end{figure*}

\begin{thm}
\label{thm:PKsignatures}
The protocol depicted in Figure \ref{fig:PKsignature} is an efficient zero
knowledge protocol for
proving knowledge of preimages of the function $f$ described in Definition \ref{def:signatureClass}.

\end{thm}

The proof is straightforward using the standard techniques.\qed

\subsection{The class $\bbbe$ of encryption schemes}

\begin{definition}
\label{def:cryptosystemClass-EoS}
$\bbbe$ is the set of public key encryption schemes $\Gamma$ that have the following
properties:
\begin{enumerate}
\item The message space is a group ${\cal M }=(\bbbg,\ast)$ and the ciphertext
  space $\cal C$ is a set equipped with a binary operation $\circ_e$.
\item Let $m \in {\cal M}$ be a message and $c$ its encryption with respect to
  a key $\pk$. On the common input $m$, $c$, and $\pk$, there exists an efficient zero knowledge proof $\zkp$ of $m$ being the
  decryption of $c$  with respect to $\pk$. The private input of the prover is either
  the private key $\sk$, corresponding to $\pk$, or the randomness used to
  encrypt $m$ in $c$.
\item $\forall m,m' \in {\cal M}$, $\forall \pk \colon$ 
$$\Gamma.\encrypt_{\pk}(m \ast m') = \Gamma.\encrypt_{\pk}(m) \circ_e \Gamma.\encrypt_{pk}(m')$$
Moreover, given the randomnesses used to encrypt $m$ in $\Gamma.\encrypt_{\pk}(m)$ and
$m'$ in $\Gamma.\encrypt_{pk}(m')$, one can deduce (using only the public parameters) the randomness used to
encrypt $m \ast m'$ in \\$\Gamma.\encrypt_{\pk}(m) \circ_e \Gamma.\encrypt_{pk}(m')$.
\end{enumerate} 
\end{definition}

Examples of encryption schemes include for instance ElGamal's encryption \cite{ElGamal1985}, Paillier's encryption \cite{Paillier1999}, or the Boneh-Boyen-Shacham scheme \cite{BonehBoyenShacham2004}. In fact, these schemes are homomorphic and possess an efficient proof of correctness of a decryption, namely the proof of equality of two discrete logarithms in case of \cite{ElGamal1985,BonehBoyenShacham2004} and the proof of knowledge of an $N$-th root in case of  \cite{Paillier1999}. Note that both  ElGamal's and  Boneh-Boyen-Shacham's encryptions are derived from the KEM/DEM paradigm and are therefore suitable for use in the new StE paradigm.
\begin{figure*}
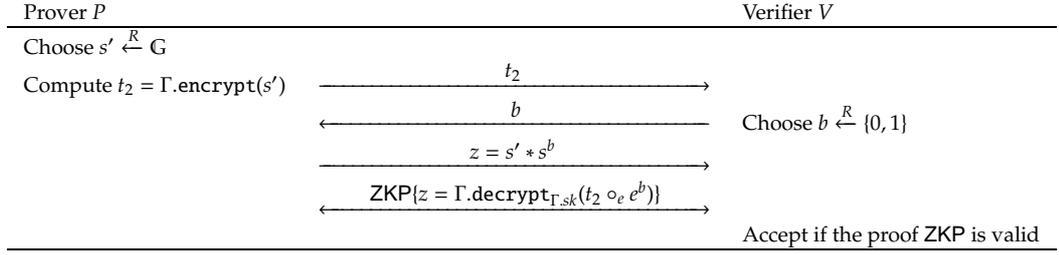

\footnotesize
\begin{center}
\begin{tabular}{lcl}
Prover $P$& & Verifier $V$ \\
\hline

 Choose $s' \hasard \bbbg$ & & \\
 Compute $ t_2= \Gamma.\encrypt(s')$  & $\xrightarrow{\makebox[50mm]{$t_2$}}$ & \\
 
&  $\xleftarrow{\makebox[50mm]{$b$}}$  &  Choose $b \hasard \{0,1\} $\\
 
 &  $\xrightarrow{\makebox[50mm]{$z=s' \ast s^b$}}$  &  \\
&  $\xleftrightarrow{\makebox[50mm]{$\zkp\{z=\Gamma.\decrypt_{\Gamma.\sk}(t_2\circ_e e^b)\}$}}$  &  \\

 & & Accept if the proof $\zkp$ is valid \\
\hline
\end{tabular}

\caption{\label{fig:PKcryptosystemEoS} Proof system of membership 
  to  the language $\{ s \colon s = \Gamma.\decrypt_{\Gamma.\sk}(e)\}$ $\textbf{Common input: } (e,\Gamma.\pk)$ and $\textbf{Private input: }$ $s$ and $\Gamma.\sk$ or randomness encrypting $s$ in $e$}

\end{center}
\end{figure*}

\begin{thm}
\label{thm:PKcryptosystemsEoS}
Let $\Gamma $ be an encryption scheme from the above class $\bbbe$. Let further $c$ be an encryption
of some message $m$ under some public key $\pk$. The protocol depicted in Figure
\ref{fig:PKcryptosystemEoS} is a zero knowledge protocol for proving knowledge of the decryption of $c$. 
\end{thm}
\begin{proof}

Completeness is straightforward. 

Validity (knowledge extractability) is also easy. In fact, suppose a malicious prover $\tilde{P}$ can successfully answer two different challenges $0$ and $1$ (challenge space is $\{0,1\}$) for the same commitment value $t_2$:
$$
z_1 = \Gamma.\decrypt(t_2) ~\wedge ~ z_2 = \Gamma.\decrypt(t_2\circ_e e)
$$
Since $\circ_e$ induces a group law in the ciphertext space of $\Gamma$, we have:
$z_1^{-1} = \Gamma.\decrypt(t_2^{-1})$. It follows that $\tilde{P}$ can compute a decryption of $e$ as $z_1^{-1}\ast z_2 = \Gamma.\decrypt(e)$. We conclude that the soundness error probability of the protocol is at most $1/2$ (we assume that $\zkp$ has negligible soundness error). We will see in Subsection \ref{subsec:soundnessError} how to reduce the soundness error without necessarily repeating the protocol many times.

For the zero-knowledgeness, we describe the following simulator:

\begin{enumerate}
\item Generate uniformly a random challenge $b' \hasard \{0,1\}$. Choose
  a random $z \hasard \bbbg$, compute $t_2=\Gamma.\encrypt_{\Gamma.\pk}(z) \circ_e e^{-b'}$ and sends it to the verifier.
\item Get $b$ from the verifier.
\item If $b= b'$, the simulator sends back $z$ and simulates the proof
  $\mathsf{ZKP}$ for $z$ being the decryption of $t_2\circ_e e^b$ (this proof
  is simulatable since it is zero knowledge by assumption). Otherwise, it goes to Step 2 (\emph{rewinds} the verifier).
\end{enumerate}

The prover's first message is always an encryption of a random value, and so is the first message of the simulator. Since $b'$ is chosen uniformly at random
from $\{0,1\}$, then, the probability that the simulator does not rewind the verifier
is $1/2$, and thus the simulator runs in polynomial time in the security parameter. Finally, the
distribution of the answers (last messages) of the prover and of the simulator is the
same. We conclude that the proof is perfectly zero knowledge.

\end{proof}

\subsection{The class $\mathbb{C}$ of commitments}
\begin{definition}
\label{def:commitClass}
$\mathbb{C}$ is the set of all commitment schemes for which there exists an algorithm
$\compute$ that  inputs the commitment  key
  $\pk$, the message $m$  and the commitment $c$ on $m$, and computes a
  description of a map $f: (\bbbg,\ast) \rightarrow (\bbbh,\circ_c)$ where:
\begin{itemize}
\item $(\bbbg,\ast)$ is a group and $\bbbh$ is a set equipped with the binary operation $\circ_c$ ,
\item $\forall r,r' \in \bbbg$: $f(r \ast r')=f(r) \circ_c f(r')$.

\end{itemize}
and an $I \in \bbbh$, such that $f(r) = I$, where $r$ is the opening value of $c$
w.r.t. $m$.
\end{definition}

It is easy to check that Pedersen's commitment scheme is in this
class. Actually, most commitment schemes have this built-in property
because it is often the case that the committer wants to prove efficiently that a
commitment is produced on some message. This is possible if the function $f$
is homomorphic as shown in Figure \ref{fig:commit}.

\begin{figure*}
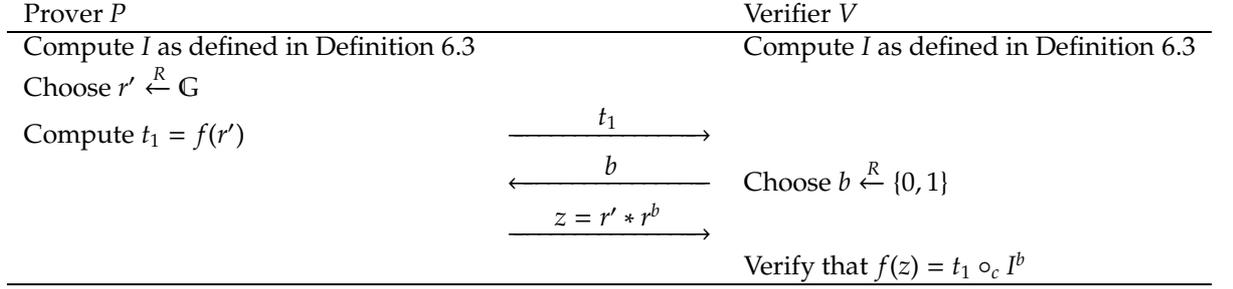

\begin{center}
\begin{tabular}{lcl}
Prover $P$& & Verifier $V$ \\
\hline
Compute $I$ as defined in Definition \ref{def:commitClass}& & Compute $I$ as defined in Definition \ref{def:commitClass}\\
 Choose $r' \hasard \bbbg$ & & \\
 Compute $ t_1=f(r')$  & $\xrightarrow{\makebox[25mm]{$t_1$}}$ & \\
 
&  $\xleftarrow{\makebox[25mm]{$b$}}$  &  Choose $b \hasard \{0,1\} $\\
 
 &  $\xrightarrow{\makebox[25mm]{$z=r'\ast r^b$}}$  &  \\
 & & Verify that $f(z)=t_1\circ_c I^b$ \\
\hline
\end{tabular}

\caption{\label{fig:commit} Proof of membership to the language
  $\{r \colon c = \commit(m,r)\}$ $\textbf{Common input: } (c,m)$ and $\textbf{Private
  input}: r$.}
\end{center}
\end{figure*}

\begin{thm}
\label{thm:commitClass}
The protocol depicted in Figure \ref{fig:commit} is an efficient zero knowledge protocol for
proving knowledge of preimages of the function $f$ described in Definition \ref{def:commitClass}.
\qed \end{thm}

\subsection{Practical realizations from StE}
We combine a secure signature scheme $ \Sigma \in \bbbs$ and a
secure encryption scheme $\Gamma \in \bbbe$, which is \emph{derived from the KEM/DEM paradigm}, in the way described in Subsection \ref{subsec:efficientStE-CDCS}. Namely we first compute an encapsulation $c$
together with its corresponding key $k$. Then we compute a signature $\sigma$ on
 $c$ concatenated with the message to be signed. Finally convert $\sigma$ to
$(s,r)$ using the $\convert$ algorithm described in Definition \ref{def:signatureClass} and encrypt
$s$ in $e={\cal D}.\encrypt_k(s)$ using $k$. The resulting confirmer signature is $(c,e,r)$.
We describe in Figure \ref{fig:conf-deny} the
confirmation/denial protocols corresponding to the resulting
construction. Note that the confirmation protocol can be also run by the
signer who wishes to confirm the validity of a just generated confirmer signature using the randomness used to generate it.

\begin{figure*}
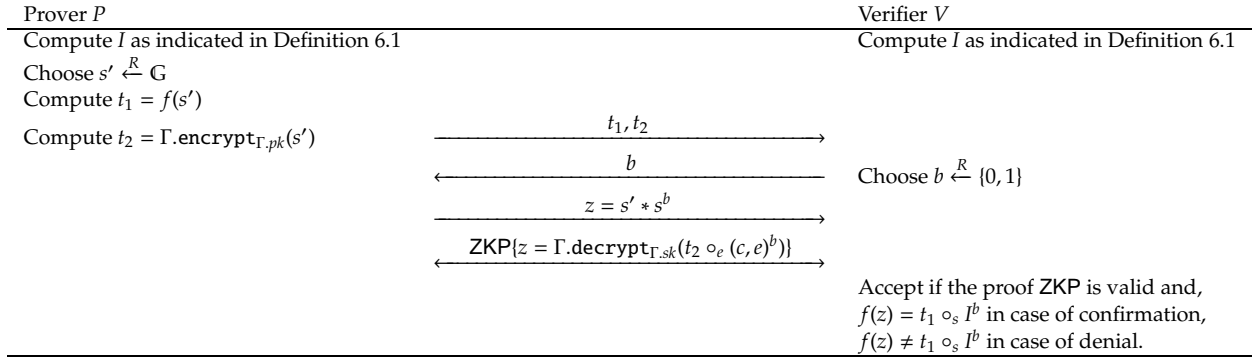

\begin{footnotesize}
\begin{center}
\begin{tabular}{lcl}
Prover $P$& & Verifier $V$ \\
\hline
Compute $I$ as indicated in Definition \ref{def:signatureClass}& & Compute $I$ as indicated in Definition \ref{def:signatureClass}\\

Choose $s' \hasard \bbbg$& & \\
Compute $ t_1=f(s')$& & \\
Compute $ t_2= \Gamma.\encrypt_{\Gamma.\pk}(s')$ & $\xrightarrow{\makebox[50mm]{$t_1,t_2$}}$ & \\
 
&$\xleftarrow{\makebox[50mm]{$b$}}$ &Choose $b \hasard \{0,1\}$\\
 
 &$\xrightarrow{\makebox[50mm]{$z=s' \ast s^b$}}$ &  \\
&$\xleftrightarrow{\makebox[50mm]{$\zkp\{z=\Gamma.\decrypt_{\Gamma.\sk}(t_2 \circ_e (c,e)^b)\}$}}$&  \\

& &Accept if the proof $\zkp$ is valid and, \\
& &$f(z) = t_1\circ_s I^b$ in case of confirmation, \\
& &$f(z) \neq t_1\circ_s I^b$ in case of denial. \\
\hline
\end{tabular}

\caption{\label{fig:conf-deny} Confirmation/Denial protocol for the new StE.
  $\pok\{s \colon s = \Gamma.\decrypt(c,e) \wedge
  \Sigma.\verify(\retrieve(s,r),m\|c) = (\neq) 1\}$ $\textbf{Common input: } (c,e,r,\Sigma.\pk,\Gamma.\pk)$ and $\textbf{Private
  input: } \Gamma.\sk$ or randomness encrypting $s$ in $(c,e)$}
\end{center}
\end{footnotesize}
\end{figure*}

\begin{thm}
\label{thm:conf}
Let $\Sigma$ and $\Gamma $ be signature and encryption schemes from the classes $\bbbs$ and $\bbbe$ resp.
The confirmation protocol (run by either the signer on a just generated
signature or by the confirmer on any signature) described in Figure
\ref{fig:conf-deny} is a zero knowledge proof of knowledge.
\end{thm}
\begin{proof}
The confirmation protocol in Figure \ref{fig:conf-deny} is a parallel
composition of the proofs depicted in Figures \ref{fig:PKsignature} and
\ref{fig:PKcryptosystemEoS}. Therefore completeness and soundness (knowledge extractability) follow  from the completeness and soundness of the underlying proofs
(see \cite{Goldreich2001}). Finally, the ZK simulator is the parallel composition of the ZK simulators for the mentioned protocols.
\qed \end{proof}

\begin{thm}
\label{thm:deny}
The denial protocol described in Figure \ref{fig:conf-deny}, for $\Sigma \in \bbbs$ and $\Gamma \in \bbbe$, is a proof of knowledge with computational zero knowledge if $\Gamma$ is IND-CPA-secure.
\end{thm}
\begin{proof}
Using the standard techniques, we prove that the denial protocol depicted in Figure
\ref{fig:conf-deny} is complete and sound. Similarly, we provide the following simulator to prove the ZK property.

 \begin{enumerate}
\item Generate $b' \in_R \{0,1\}$. Choose $z \in_R \bbbg$ and a random $t_1 \in_R f(\bbbg)$ and $t_2=\Gamma.\encrypt_{\Gamma.\pk}(z) \circ_e (c,e)^{-b'}$.
\item Get $b$ from the verifier. If $b= b'$, it sends $z$ and simulates the
  proof  $\mathsf{ZKP}$ of $z$ being the decryption of $t_2\circ_e (e,s_k)^b$. If $b \neq b'$, it goes to Step 1.
\end{enumerate}
The prover's first message is an encryption of a random value $s' \in_R \bbbg$,
in addition to $f(s')$. The simulator's first
message is an encryption of a random value $z\ast {s^{-b'}}$ and the element $t_1 \in_R f(\bbbg)$, which is \emph{independent of $z$}. Distinguishing
those two cases is at least as hard as breaking the IND-CPA security of the
underlying encryption scheme. Therefore, under the IND-CPA security of the
encryption scheme, the simulator's and prover's first message distributions are
indistinguishable. Moreover, the simulator runs in expected polynomial time, since
the number of rewinds is $2$. Finally, the distributions of the prover's and the simulator's messages in the
last round are again, by the same argument, indistinguishable under the
IND-CPA security of the encryption scheme.
\qed \end{proof}

\paragraph{Concurrent Zero knowledgeness}
If the proof  $\zkp$ underlying the above protocols is a public-coin Honest-Verifier Zero Knowledge (HVZK) protocol, then there are a number of efficient transformations that turns the above confirmation/denial protocols into proofs that are  concurrent zero knowledge, e.g.  \cite{MicciancioPetrank2002}. For instance, if  $\zkp$ is a Sigma protocol, then the aforementioned confirmation/denial protocols can be efficiently turned into concurrent ZK proofs according to  \cite{Damgard2000}; this transformation preserves the number of rounds while it incurs a tiny overhead in the computational complexity (computation of a commitment on a message). Note that although the transformation \cite{Damgard2000} is in the auxiliary string model, such a scenario is easy to achieve in a public-key setting; for example certificates computed by a PKI on public keys are possible candidates to auxiliary strings to the players.

\paragraph{Performance of the new StE}
Our variant of StE improves the plain paradigm \cite{CamenischMichels2000} as it weakens the assumption on the underlying
encryption from IND-CCA to IND-CPA. This
impacts positively the efficiency of the construction from many sides. In fact,
the resulting signature is shorter and its generation/verification cost is smaller. An illustration is given by ElGamal's encryption and its IND-CCA
variant, namely Cramer-Shoup's encryption where the ciphertexts are at least
twice longer than ElGamal's ciphertexts. Also, there is a multiplicative
factor of at least two in favor of ElGamal's encryption/decryption
cost. Moreover, the confirmation/denial protocols are rendered more efficient
by allowing homomorphic encryption schemes as shown earlier in this subsection, e.g. \cite{ElGamal1985,BonehBoyenShacham2004} \footnote{Both schemes are IND-CPA secure and are derived from the KEM/DEM paradigm. Moreover, the underlying KEM and DEM present interesting homomorphic properties that make them belong the class $\bbbe$ of encryption schemes. We refer to the discussion after Definition \ref{def:cryptosystemClass-EoS} for the details}. Such encryption schemes were not
possible to use before since a homomorphic scheme can never attain the IND-CCA security. Besides,
even when the IND-CCA encryption scheme is decryption verifiable, e.g.
Cramer-Shoup, the involved protocols are much more expensive than
those corresponding to their IND-CPA variant. 

\noindent The construction achieves also better performances than the
proposal of \cite{GoldwasserWaisbard2004}, where the confirmer signature
comprises $k$ commitments and $2k$ IND-CCA encryptions, where $k$
is the number of rounds used in the confirmation protocol. Moreover, the
denial protocol presented in \cite{GoldwasserWaisbard2004} suffers the resort
to proofs of general NP statements (where the considered encryption is
IND-CCA). The same remark applies to the construction of \cite{Wikstroem2007}
where both the confirmation and denial protocols rely on proofs of general NP statements.

Finally, we remark that our new StE, first introduced in \cite{ElAimani2008}, captures many efficient realizations of
confirmer/undeniable signatures, e.g. \cite{LeTrieuKurosawaOgata2009b,SchuldtMatsuura2010}. It also serves for analyzing some early schemes that had a speculative security: the Damg\r{a}rd-Pedersen \cite{DamgardPedersen1996}  undeniable signatures. In fact, we showed in Subsection \ref{app:damgardPedersen}, that these signatures are unlikely to be invisible, and we proposed a fix so that
they meet the required security notion; interestingly, this repair turns out to be a special
instantiation of the new StE paradigm. Actually, even the confirmation/denial protocols
provided in \cite{DamgardPedersen1996} happen to be a special case of
the confirmation/denial protocols depicted in Figure \ref{fig:conf-deny}.

\subsection{Practical realizations from CtEtS}
\label{subsec:NewCtEaS-instantiations}
The CtEtS has the merit of supporting \emph{any} digital signature scheme as a building block. In fact, efficient as the StE paradigm is, it still applies only to a restricted class of signatures. For instance, StE does not seem to be plausible with the PSS signature scheme \cite{BellareRogaway1996}.

CtEtS does not involve a proof of knowledge of a signature in its confirmation/denial protocols. In fact, confirmation (denial) of a signature on a certain message consists in proving knowledge of the decryption of a given ciphertext, and that this decryption is (is not) the opening value of a given commitment on the message in question. More specifically, the confirmation/denial protocols for CtEtS, when the encryption $\Gamma$ belongs to the class $\bbbe$ and the commitment $\Omega$ belongs to the class $\mathbb{C}$, are depicted in Figure \ref{fig:conf-denySOC}.

\begin{figure*}
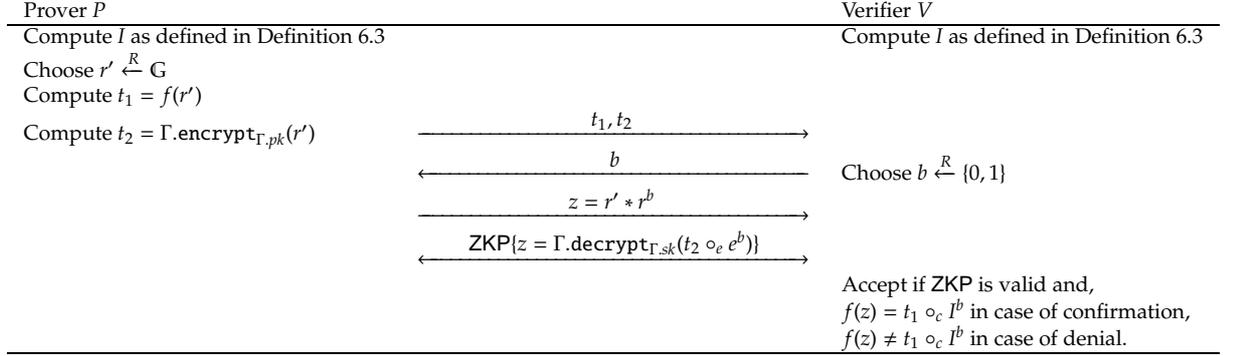

\begin{center}
\footnotesize
\begin{tabular}{lcl}
Prover $P$& & Verifier $V$ \\
\hline

Compute $I$ as defined in Definition \ref{def:commitClass}& & Compute $I$ as defined in Definition \ref{def:commitClass}\\
Choose $r' \hasard \bbbg$& & \\
Compute $t_1=f(r')$& & \\
Compute $t_2= \Gamma.\encrypt_{\Gamma.\pk}(r')$ & $\xrightarrow{\makebox[50mm]{$t_1,t_2$}}$ & \\
 
&$\xleftarrow{\makebox[50mm]{$b$}}$ &Choose $b \hasard \{0,1\} $\\
 
 &$\xrightarrow{\makebox[50mm]{$z=r' \ast r^b$}}$ &  \\
&$\xleftrightarrow{\makebox[50mm]{$\mathsf{ZKP}\{z=\Gamma.\decrypt_{\Gamma.\sk}(t_2\circ_e e^b)\}$}}$&  \\

& &Accept if $\mathsf{ZKP}$ is valid and, \\
& &$f(z) = t_1 \circ_c I^b$ in case of confirmation, \\
& &$f(z) \neq t_1 \circ_c I^b$ in case of denial. \\
\hline
\end{tabular}

\caption{\label{fig:conf-denySOC} Confirmation/Denial protocol for the new CtEtS paradigm.
  $\pok \{r \colon  r = \Gamma.\decrypt(e) \wedge
  c = (\neq) \Omega.\commit(m,r) \}$ $\textbf{Common input: }
  (e,c,m,\Gamma.\pk,\Omega.\pk)$ and $\textbf{Private
  input: } \Gamma.\sk$ or randomness encrypting $r$ in $e$.}

\end{center}
\end{figure*}

\begin{thm}
\label{thm:confSOC}
Let $\Omega$ and $\Gamma $ be commitment  and encryption schemes from the classes $\bbbc$ and $\bbbe$ resp. The confirmation protocol  depicted in Figure
\ref{fig:conf-denySOC} is a zero knowledge proof of knowledge.
\end{thm}

\begin{thm}
\label{thm:denySOC}
The denial protocol depicted in Figure  \ref{fig:conf-denySOC}, for $\Omega \in \bbbc$ and $\Gamma \in \bbbe$, is a proof of knowledge with computational zero knowledge if $\Gamma$ is IND-CPA-secure.
\end{thm}

The proofs are similar
to those of Theorem \ref{thm:conf} and Theorem \ref{thm:deny} respectively. \qed 
Similar to the new StE paradigm, our new CtEtS achieves better performances than the original technique (short signature, small
generation, verification, and conversion cost) while applying to any signature scheme. Moreover, it accepts many efficient instantiations (if the used commitment and encryption belong to the already mentioned classes) as its confirmation/denial protocols no longer relies on general proofs of NP statements. 

\subsection{The EtS paradigm}
\label{subsec:EtS-CDCS}

EtS can be seen as a special instance of CtEtS since IND-CPA encryption can be easily used to get
statistically binding and computationally hiding commitments. Therefore, one can first commit to the message to be signed  using the the encryption scheme, then sign the resulting ciphertext. The confirmer signature is composed of the ciphertext and of its signature. In fact, there will be
no need to encrypt the string used to produce the ciphertext (commitment) since the private key of the encryption scheme is sufficient to check the validity of a ciphertext w.r.t. a given message. Finally, selective conversion is achieved by releasing a \emph{Non-Interactive} ZK (NIZK) proof that the ciphertext (first part of the confirmer signature) decrypts to the message in question. Note that in this paradigm, the setting should include a \emph{trusted authority} that generates the common reference string (CRS); again, this is plausible in a public key setting as PKIs can successfully play this role.

\noindent Similar to CtEtS, unforgeability of EtS rests on the unforgeability of the underlying signature, whereas invisibility rests on the strong unforgeability (SEUF-CMA) of the signature and on the indistinguishability (IND-CPA) of the encryption. We discuss in the rest of this subsection how to achieve practical realizations from this technique.

\subsubsection{Confirmation/denial protocols} Confirmation in EtS amounts to a proof of correctness of a decryption (i.e. a given ciphertext encrypts a given message). This is in general easy since in most encryption
schemes, one can define, given a ciphertext $c$ and its underlying plaintext $m$, two
homomorphic one-way maps $f$  and $g$, and two quantities $I$ and $J$
such that $f(r)=I$ and  $g(\sk)=J$, where $r$ is the randomness used to
encrypt $m$ in $c$, and $\sk$ is the private key of the encryption scheme. Examples of such encryptions include \cite{ElGamal1985,BonehBoyenShacham2004,Paillier1999,CramerShoup2003,CamenischShoup2003}. The confirmation protocol in this case will be reduced to a proof of knowledge of a preimage of $J$ ($I$) by the function $g$ ($f$), for which we provided an efficient proof
in Figure \ref{fig:PKsignature}.

Concerning the denial protocol, it is not always straightforward. In most
discrete-logarithm-based encryptions, this protocol amounts to a proof of
inequality of discrete logarithms as in
\cite{ElGamal1985,BonehBoyenShacham2004,CramerShoup2003}. In case the
encryption scheme belongs to the class $\bbbe$ defined earlier, Figure \ref{fig:deny-EtS} provides an efficient proof that
$c$ encrypts some $\tilde{m}$ which is different from $m$. In the protocol provided in this
figure, $f$ denotes an arbitrary \emph{homomorphic} map:
$$
f(m\ast m')=f(m)\circ f(m')
$$

\begin{figure*}
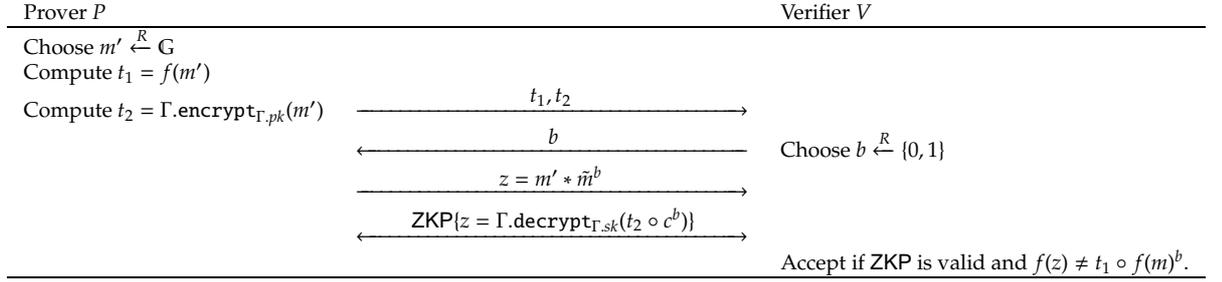

\footnotesize
\begin{center}

\begin{tabular}{lcl}
Prover $P$& & Verifier $V$ \\
\hline

Choose $m' \hasard \bbbg$& & \\
Compute $ t_1=f(m')$& & \\
 Compute $ t_2= \Gamma.\encrypt_{\Gamma.\pk}(m')$ & $\xrightarrow{\makebox[50mm]{$t_1,t_2$}}$ & \\
 
&$\xleftarrow{\makebox[50mm]{$b$}}$ &Choose $b \hasard \{0,1\}$\\
 
 &$\xrightarrow{\makebox[50mm]{$z=m' \ast \tilde{m}^b$}}$ &  \\
&$\xleftrightarrow{\makebox[50mm]{$\zkp\{z=\Gamma.\decrypt_{\Gamma.\sk}(t_2\circ c^b)\}$}}$&  \\

 & &Accept if $\zkp$ is valid and $f(z) \neq t_1 \circ f(m)^b$.\\
\hline
\end{tabular}

\caption{\label{fig:deny-EtS} Denial protocol in EtS. $\pok\{\tilde{m} \colon \tilde{m} = \Gamma.\decrypt(c) \wedge
   \tilde{m} \neq m  \}$ \textbf{Common input: }
  $(m,c,\Gamma.\pk)$ and \textbf{Private
  input: } $\tilde{m}$ and $\Gamma.\sk$ or randomness encrypting $\tilde{m}$ in $c$}

\end{center}
\end{figure*}

Similarly, the above denial protocol can be shown to be a proof of knowledge with computational ZK, if $\Gamma$ is IND-CPA secure.

\subsubsection{Selective conversion}
\label{subsubsection:conversion-CDCS}
Selective conversion in confirmer signatures from EtS consists in providing the non-interactive variant of the confirmation protocol. We note in this paragraph few solutions to achieve this goal:
\begin{description}

\item \emph{The case of fully decryptable encryption schemes} I.e.
encryption schemes where decryption leads to the randomness used to produce the
ciphertext. In this case, selective conversion can simply be achieved by releasing the randomness used to generate the ciphertext. Examples of encryption schemes from this class include \cite{Paillier1999}'s encryption: the scheme operates on messages in $\bbbz_N$, where $N=pq$ is a
safe RSA modulus. Encryption of a message $m$ is done by picking a random $r
\in_R \bbbz_N^{\times}$ and then computing the ciphertext $c=r^N(1+mN) \bmod
N^2$. Decryption of a ciphertext $c$ is done by raising it to $\lambda=\mathrm{lcm}(p-1,q-1)$ to find $m$. It is easy to see that recovering $r$, once $m$ is computed, amounts to an extraction of the $N$-th root of $\frac{c}{1+mN}$.

\item \emph{Damg{\aa}rd et al. \cite{DamgardFazioNicolosi2006}'s
    solution. }This solution transforms a 3-move interactive ZK protocol $P$ with linear
  answer to a non-interactive ZK one (NIZK) using a homomorphic
  encryption scheme in a registered key model, i.e. in a model where the verifier
  registers his key. The authors in
  \cite{DamgardFazioNicolosi2006} proposed an efficient illustration using
  Paillier's encryption and the proof of equality of two discrete
  logarithms. We conclude that with such a technique, EtS accepts an efficient instantiation if the considered
  encryption allows proving the correctness of a decryption using the proof of
  equality of two discrete logarithms, e.g. \cite{ElGamal1985,BonehBoyenShacham2004,CramerShoup2003}.
\item \emph{Groth and Sahai \cite{GrothSahai2008}'s solution. }This technique is applicable in general for encryption schemes where the encryption/decryption algorithms perform only group or pairing (if bilinear groups are involved) operations on the randomness or the private key resp. 

\item  \emph{Lindell \cite{Lindell2014}'s solution. } This Fiat-Shamir like transform turns any Sigma protocol for a relation $R$  into a NIZK proof for the associated language $L_R$, in the common reference string model (without any random oracle). The concrete computational complexity of the transform is slightly higher than the original Fiat-Shamir transform. 

\end{description}

\subsection{Reducing the soundness error}
 \label{subsec:soundnessError}
The protocols presented earlier in this section consist of a proof of knowledge of preimages, by some \emph{homomorphic} map, which incidentally satisfy a relation efficiently provable via a zero knowledge proof $\zkp$. 

In this subsection, we show how to reduce the soundness error of these protocols without necessarily repeating them. We will focus on the part of the protocol proving knowledge of the preimage; actually we assume $\zkp$ has a negligible soundness error since it can itself implement the optimizations we propose if it is a proof of knowledge for group homomorphisms.

Let $f \colon (\bbbg,\ast) \rightarrow (\bbbh,\circ)$ be the homomorphic map underlying the proof of knowledge. Let further $I$ be the value for which we want to prove knowledge of a preimage. We consider a challenge space $\cal C$ that satisfies, for some known values $\ell \in \bbbz$ and $u \in \bbbg$, the following \cite{Maurer2015}:
\begin{enumerate}
\item $\gcd(b_1-b_2,\ell)=1$ for all $b_1,b_2 \in {\cal C}$ (with $b_1\neq b_2$),
\item $f(u) = I^\ell $. 
\end{enumerate}

Note that the above conditions are easily met in groups with known prime order $\ell$, i.e. discrete-log based groups. 

The protocol below is an efficient zero knowledge proof of knowledge of a preimage of $I$, if ${\cal C}$ is polynomially bounded.


\begin{center}
\small
\begin{tabular}{lcl}
Prover $P$& & Verifier $V$ \\
\hline

 Choose $s' \hasard \bbbg$ & & \\
 Compute $ t=f(s')$  & $\xrightarrow{\makebox[15mm]{$t$}}$ & \\
 
&  $\xleftarrow{\makebox[15mm]{$b$}}$  &  Choose $b \hasard {\cal C \subseteq \mathbb{N}} $\\
 
 &  $\xrightarrow{\makebox[15mm]{$z=s'\ast s^b$}}$  &  \\

 & & Accept if $f(z)=t\circ I^b$ \\
\hline
\end{tabular}
\end{center}

Completeness is straightforward. 

For knowledge extractability, we consider two accepting transcripts for the same commitment value $t$ and different challenges $b_1,b_2$ ($b_2 \geq b_1$).  Let $z_1,z_2$ the responses of the prover in the last round.

\noindent We have $f(z_1) = t \circ I^{b_1}$ and  $f(z_2) = t \circ I^{b_2}$. Therefore $f(z_1^{-1}\ast z_2) = I^{b_2-b_1}$. 

\noindent We compute values $x,y$ by the Extended Euclidean Algorithm to get. $x\ell + y(b_2-b_1) = 1$. It follows that $I= I^{x\ell}I^{y(b_2-b_1)}$. Thus $I=f(u^x \ast (z_1^{-1}\ast z_2)^y)$. In other words, a preimage of $I$ can be computed as $u^x \ast (z_1^{-1}\ast z_2)^y$.

Finally, the ZK simulator is similar to that of the original protocol with the exception of drawing the challenge $b'$, in the first stage of the protocol, from $\cal C$. The new probability of not rewinding the verifier is $1/|{\cal C}|$. Thus, $\cal C$ must be polynomially bounded in order to guarantee a polynomial
running time of the simulator.

\section{Verifiable Signcryption}

\subsection{Syntax and model}

A verifiable signcryption scheme consists of the following algorithms/protocols:
\begin{description}

\item \textbf{Setup ($\setup(1^\kappa)$). }This probabilistic algorithm inputs a security parameter $\kappa$, and generates the public parameters $\param$ of the signcryption scheme.
\item \textbf{Key generation ($\keygen_U(1^\kappa,\param), U \in \{S,R\}$). }This probabilistic algorithm inputs the security parameter $\kappa$ and the public parameters $\param$, and outputs a key pair $(\pk_U,\sk_U)$ for the system user $U$ which is either the sender $S$ or the receiver $R$. 
\item \textbf{Signcryption ($\signcrypt(m,\sk_S,\pk_S,\pk_R)$). }This probabilistic algorithm inputs a message $m$, the key pair $(\sk_S,\pk_S)$ of the sender, the public key $\pk_R$ of the receiver, and outputs the signcryption $\mu$ of the message $m$.
\item \textbf{Proof of validity ($\proveValidity(\mu,\pk_S,\pk_R)$). }This is an interactive protocol between the receiver or the sender who has just generated a signcryption $\mu$ on some message, and any verifier: the sender uses the randomness used to create $\mu$ (as private input) and the receiver uses his private key $\sk_R$ in order to convince the verifier that $\mu$ is a valid signcryption on some message. The common input to both the prover and the verifier comprises the signcryption $\mu$, $\pk_S$, and $\pk_R$. At the end, the verifier either accepts or rejects the proof. 
\item \textbf{Unsigncryption ($\unsigncrypt(\mu,\sk_R,\pk_R,\pk_S)$). }This is a deterministic algorithm which inputs a putative signcryption $\mu$ on some message, the key pair $(\sk_R,\pk_R)$ of the receiver, and the public key $\pk_S$ of the sender, and outputs either the message underlying $\mu$ or an error symbol $\perp$.

\item \textbf{Confirmation/Denial ($\{\confirm,\deny\}(\mu,m,\pk_R,\pk_S)$). }These are  interactive protocols between the receiver and any verifier; the receiver uses his private key $\sk_R$ (as private input) to convince any verifier that a signcryption $\mu$ on some message $m$ is/is not valid. The common input comprises the signcryption $\mu$ and the message $m$, in addition to $\pk_R$ and $\pk_S$. At the end, the verifier is either convinced of  the validity/invalidity of $\mu$ w.r.t. $m$ or not.
\item \textbf{Signature extraction ($\sigExtract(\mu,m,\sk_R,\pk_R,\pk_S)$).} This is an algorithm which inputs a signcryption $\mu$, a message $m$, the key pair $(\sk_R,\pk_R)$ of the receiver, and the public key $\pk_S$ of the sender, and outputs either an error symbol $\perp$ if $\mu$ is not a valid signcryption on $m$, or a string $\sigma$ which is a valid digital signature on $m$ w.r.t $\pk_S$ otherwise.

\item \textbf{Signature verification ($\sigVerify(\sigma,m,\pk_S)$).} This is an algorithm for verifying extracted signatures. It inputs the extracted signature $\sigma$, the message $m$ and $\pks$, and outputs either $0$ or $1$.
\end{description}

We require in a signcryption scheme correctness and soundness. Moreover, the protocols $\proveValidity$ and $\{\confirm, \deny\}$ must be complete, sound, and non-transferable. The formal definitions of these notions are similar to the confirmer signatures case (see Subsection \ref{subsec:syntax+model} ), therefore we omit them here due to page limitations.
 
Finally, we require in a signcryption scheme two further properties: unforgeability, which protects the sender's authenticity from \emph{malicious insider} adversaries (i.e. the receiver), and  indistinguishability, which protects the sender's privacy from \emph{outsider adversaries}.


\begin{definition}[Unforgeability]\label{def:sigcry-EUFCMA}
We consider a signcryption scheme $\SC$ given by the algorithms/protocols defined earlier in this section. Let ${\mathcal A}$ be a PPTM. We consider the random
experiment depicted in Experiment $\mathbf{Exp}_{\ensuremath{\SC},{\mathcal A}}^{\mathsf{euf\mbox{-}cma}}(1^\kappa)$.  

\begin{experiment}[Experiment $\mathbf{Exp}_{\ensuremath{\SC},{\mathcal A}}^{\mathsf{euf\mbox{-}cma}}(1^\kappa)$]

\item $\param \leftarrow \ensuremath{\SC}.\setup(1^\kappa)$;
\item $(\pks,\sks) \leftarrow \ensuremath{\SC}.\keygen_S(1^\kappa,\param)$;
\item $(\pkr,\skr) \leftarrow {\cal A}(\pks)$;
\item $\mu^{\star} \leftarrow \mathcal{A}^{\mathfrak{S}}(\pks,\pkr,\skr)$;\\
\phantom{$\mu^{\star} \leftarrow$} $
  \begin{array}{l} 
\mathfrak{S} : m \longmapsto \SC.\signcrypt{\{\sks,\pks,\pkr\}}(m) \\ 
\end{array}$
\item \textsf{return} 1 \textsf{if and only if :} \\
\hspace{3mm} - $\SC.\unsigncrypt_{\{\skr,\pkr,\pks\}}[\mu^\star] = m^\star$ \\
\hspace{3mm} - $m^\star$ was not queried to $\mathfrak{S}$
\end{experiment}

\noindent We define the \emph{success} of $\mathcal{A}$ \emph{via}: 
$$\displaystyle
{\mathbf{Succ}_{\ensuremath{\SC},{\mathcal A}}^{\textsf{euf-cma}}(1^\kappa) = \Pr\left[\mathbf{Exp}_{\ensuremath{\SC},{\mathcal
      A}}^{\textsf{euf-cma}}(1^\kappa)=1\right].} 
$$

\noindent Given $(t,q_s) \in \mathbb{N}^2$ and $\varepsilon \in [0,1]$, $\cal
A$ is called a $(t,\varepsilon,q_s)$-EUF-CMA adversary against
\ensuremath{\SC} if, running in time $t$ and issuing $q_s$
queries to the $\SC.\signcrypt$ oracle, $\cal A$ has
$\mathbf{Succ}_{\ensuremath{\SC},{\mathcal A}}^{\textsf{euf-cma}}(1^\kappa) \geq
\varepsilon$. The scheme \ensuremath{\SC}  is said to be
$(t,\varepsilon,q_s)$-EUF-CMA secure if no
$(t,\varepsilon,q_s)$-EUF-CMA adversary against it exists. 
 \end{definition}

\begin{remark}
Note that $\cal A$ is not given the $\proveValidity$, $\unsigncrypt$, $\sigExtract$, and $\{\confirm,\deny\}$ oracles. In fact, these oracles are useless for him as he has the receiver's private key $\skr$ at his disposal.
\end{remark}


\begin{definition}[Indistinguishability (IND-CCA)]\label{def:sigcryINDCCA}
Let $\SC$ be a signcryption scheme, and let ${\mathcal A}$ be a PPTM. We consider the random experiment for $b \hasard \{0,1\}$ depicted in Experiment $\mathbf{Exp}_{\ensuremath{\SC},{\mathcal
A}}^{\textsf{ind-cca}\mbox{-}b}(1^\kappa)$. 
\begin{center}
\footnotesize
\begin{experiment}[Experiment $\mathbf{Exp}_{\ensuremath{\SC},{\mathcal
A}}^{\textsf{ind-cca}\mbox{-}b}(1^\kappa)$]

\item $\param \leftarrow \ensuremath{\SC}.\setup(1^\kappa)$;
\item $(\sks,\pks) \leftarrow \ensuremath{\SC}.\keygen_S(1^\kappa,\param)$;
\item $(\skr,\pkr) \leftarrow \ensuremath{\SC}.\keygen(1^\kappa,\param)$;

\item $(m_0^{\star},m_1^\star,\mathcal{I}) \leftarrow {\mathcal A}^{\mathfrak{S}, \mathfrak{V}, \mathfrak{U}, \mathfrak{C}}
(\pks,\pkr)$;\\
\phantom{$(m_0^{\star},m_1^{\star},\mathcal{I}) \leftarrow$} $\left\vert
\begin{array}{l} 
\mathfrak{S} : m  \longmapsto \SC.\signcrypt_{\{\sks,\pks,\pkr\}}(m) \\ 
\mathfrak{V} : \mu \longmapsto \SC.\proveValidity(\mu,\pks,\pkr) \\
\mathfrak{U} : \mu \longmapsto \SC.\unsigncrypt_{\skr,\pkr,\pks}(\mu) \\
\mathfrak{C} : (\mu,m) \longmapsto \SC.\{\confirm,\deny\}(\mu,m,\pkr,\pks) \\
\mathfrak{P} : (\mu,m) \longmapsto \SC.\sigExtract(\mu,m,\pkr,\pks) \\
\end{array} \right.$ \\

\item $\mu^{\star} \leftarrow \SC.\signcrypt_{\{\sks,\pks,\pkr\}}(m^\star_b)$;

\item $d\leftarrow {\mathcal A}^{\mathfrak{S}, \mathfrak{V}, \mathfrak{U}, \mathfrak{C}}
(\mathcal{I},\mu^{\star},\pks,\pkc)$;\\
\phantom{$$} $\left\vert
\begin{array}{l} 
\mathfrak{S} : m  \longmapsto \SC.\signcrypt_{\{\sks,\pks,\pkr\}}(m) \\
\mathfrak{V} : \mu \longmapsto \SC.\proveValidity(\mu,\pks,\pkr) \\ 
\mathfrak{U} : \mu (\neq \mu^\star) \longmapsto \SC.\unsigncrypt_{\skr,\pkr,\pks}(\mu) \\
\mathfrak{C} : (\mu,m) (\neq (\mu^\star,m_i^\star), i=0,1) \longmapsto
\SC.\{\confirm,\deny\}(\mu,m) \\
\mathfrak{P} : (\mu,m) (\neq (\mu^\star,m_i^\star), i=0,1) \longmapsto
\SC.\sigExtract(\mu,m) \\
\end{array} \right.$ \\
\item Return $d$;
\end{experiment}
\end{center}

\noindent We define the \emph{advantage} of $\mathcal{A}$ \emph{via}:

$$ 
\mathbf{Adv}_{\ensuremath{\SC},{\mathcal A}}^{\textsf{ind-cca}}(1^\kappa)= \left\vert \Pr\left[\mathbf{Exp}_{\ensuremath{\SC},{\mathcal
A}}^{\mathsf{ind-cca-b}}(1^\kappa)=b\right] - \frac{1}{2} \right\vert.
$$
Given $(t,q_s,q_v,q_{u},q_{cd},q_{e}) \in \mathbb{N}^6$ and $\varepsilon \in [0,1]$, $\cal
A$ is called a $(t,\varepsilon,q_s,q_v,q_{u},q_{cd},q_{e})$-IND-CCA adversary
against \ensuremath{\SC} if, running in time $t$  and issuing $q_s$ queries to the $\signcrypt$ oracle, $q_v$ queries to the
$\proveValidity$ oracle, $q_{u}$ queries to the $\unsigncrypt$
oracle, $q_{cd}$ queries to the $\{\confirm,\deny\}$
oracle, and $q_{e}$ to the $\sigExtract$ oracle, $\cal A$ has\\ $\mathbf{Adv}_{\ensuremath{\SC},{\mathcal
    A}}^{\mathsf{ind-cca}}(1^\kappa) \geq \varepsilon$. The scheme
\ensuremath{\SC} is
$(t,\varepsilon,q_s,q_v,q_{u},q_{cd},q_{e})$-IND-CCA secure if no
$(t,\varepsilon,q_s,q_v,q_{u},q_{cd},q_{e})$-IND-CCA adversary against it
exists.

\end{definition}

\subsection{Classical constructions for verifiable signcryption}
\label{subsec:mainConstructions}
Let $\Sigma$ be a digital signature scheme given by $\Sigma.\keygen$ which generates a key pair ($\Sigma.\sk$, $\Sigma.\pk$), $\Sigma.\sign$, and $\Sigma.\verify$. Let furthermore $\Gamma$ denote a public key encryption scheme described by $\Gamma.\keygen$ that generates the key pair ($\Gamma.\sk$,$\Gamma.\pk$), $\Gamma.\encrypt$, and $\Gamma.\decrypt$. Finally, let $\Omega$ be a commitment scheme given by the algorithms $\Omega.\commit$ and $\Omega.\open$.  The most popular paradigms used to devise signcryption schemes from basic primitives are:

\begin{itemize}
\item The \emph{``sign-then-encrypt'' (StE) paradigm} \cite{AnDodisRabin2002,MatsudaMatsuuraSchuldt2009,ChibaMatsudaSchuldtMatsuura2011}. Given a message $m$, $\signcrypt$ first produces a signature $\sigma$ on the message using $\Sigma.\sk$, then encrypts $m\|\sigma$ under $\Gamma.\pk$. The result forms the signcryption on $m$. To $\unsigncrypt$, one first decrypts the signcryption using $\Gamma.\sk$ in $m\|\sigma$, then checks the validity of $\sigma$, using $\Sigma.\pk$, on $m$. Finally, $\sigExtract$ of a valid signcryption $\mu=\Gamma.\encrypt(m\|\sigma)$ on $m$ outputs $\sigma$.
 
\item The \emph{``encrypt-then-sign'' (EtS) paradigm} \cite{AnDodisRabin2002,MatsudaMatsuuraSchuldt2009}. Given a message $m$, $\signcrypt$ produces an encryption $e$ on $m$ using $\Gamma.\pk$, then produces a signature $\sigma$ on $e$ using $\Sigma.\sk$; the signcryption is the pair $(e,\sigma)$. To $\unsigncrypt$ such a signcryption, one first checks the validity of $\sigma$ w.r.t. $e$ using $\Sigma.\pk$, then decrypts $e$ using $\Gamma.\sk$ to get $m$. Finally, $\sigExtract$ outputs a zero knowledge non-interactive (NIZK) proof that $m$ is the decryption of $e$; such a proof is possible since the statement in question is in NP (\cite{GoldreichMicaliWigderson1986} and \cite{BlumFeldmanMicali1988}). This paradigm naturally requires the presence of a trusted authority in order to generate the common reference string needed for the NIZK proofs.

\item The \emph{``commit-then-encrypt-and-sign'' (CtEaS) paradigm} \cite{AnDodisRabin2002}. This construction has the advantage of performing the signature and the encryption \emph{in parallel} in contrast to the previous sequential compositions. Given a message $m$, one first produces a commitment $c$ on it using some random nonce $r$, then encrypts $m\|r$ under $\Gamma.\pk$, \emph{and} produces a signature $\sigma$ on $c$ using $\Sigma.\sk$. The signcryption is the triple $(e,c,\sigma)$. To $\unsigncrypt$ such a signcryption, one first checks the validity of $\sigma$ w.r.t. $c$, then decrypts $e$ to get $m\|r$, and finally checks the validity of the commitment $c$ w.r.t $(m,r)$. $\sigExtract$ is achieved by releasing the decryption of $e$, namely $m\|r$.
 
\end{itemize}

The proofs of  well (mal) formed-ness, namely $\sf{prove}$-$\sf{Validity}$ and $\{\confirm,\deny\}$ can be carried out since the languages in question are in NP and thus accept zero knowledge proof systems \cite{GoldreichMicaliWigderson1986}. 

\subsection{Negative results for StE and CtEaS}
We proceed in this subsection as we did in confirmer signatures. First, we prove that OW-CCA and NM-CPA encryption are insufficient to yield IND-CCA constructions from StE or CtEaS. We first prove this result for \emph{key-preserving reductions}, then we generalize it to arbitrary reductions assuming further properties on the underlying encryption. Next, we rule out OW-CPA, IND-CPA, and OW-PCA by remarking that ElGamal's \cite{ElGamal1985} encryption  meets all those notions but 
leads to a simple attack against IND-CCA, when employed in constructions from StE or CtEaS.

\begin{lemma}
\label{lemma:OW-CCA-insufficient}
Assume there exists a key-preserving reduction $\cal R$ that converts an IND-CCA adversary $\cal A$ against signcryptions from the StE (CtEaS) paradigm to a OW-CCA adversary against the underlying encryption scheme. Then, there exists a meta-reduction $\cal M$ that OW-CCA breaks the encryption scheme in question. 
\end{lemma}

\begin{lemma}
\label{lemma:NM-CPA-insufficient}
Assume there exists a key-preserving reduction $\cal R$ that converts an IND-CCA adversary $\cal A$ against signcryptions from the StE (CtEaS) paradigm to an NM-CPA adversary against the underlying encryption scheme. Then, there exists a meta-reduction $\cal M$ that NM-CPA breaks the encryption scheme in question.
\end{lemma}

The proofs are similar to those of Lemma \ref{lemma:OW-CCA-CS} and Lemma \ref{lemma:NM-CPA-CS} respectively.

\noindent We similarly generalize the previous results to arbitrary reductions as in Subsection \ref{subsec:arbitraryReductions} if the encryption scheme has a \emph{non-malleable key generator}, which informally means that OW-CCA (NM-CPA) breaking the encryption, w.r.t. a public key $\pk$, is no easier when given access to a decryption oracle w.r.t. any key $\pk'$ different from $\pk$.

\noindent Moreover, we can rule out the OW-CPA, OW-PCA, and IND-CPA notions by remarking that ElGamal's encryption meets all those notions (under different assumptions), but cannot be employed in StE and CtEaS as it is malleable. In fact, the indistinguishability adversary can create a new signcryption (by re-encrypting the ElGamal encryption) on the challenge message, and query it for unsigncryption. The answer to such a query is sufficient to conclude. 

In consequence of the above analysis, the used encryption scheme has to satisfy at least IND-PCA security in order to lead to secure signcryption from StE or CtEaS. This translates in expensive operations, especially if verifiability is further required for the resulting signcryption.

\subsection{Positive results for signcryption schemes}

\subsubsection{The new StE and CtEaS paradigms} Signcryptions from StE or CtEaS suffer the strong forgeability: given a signcryption on some message, one can create another signcryption on the same message without the sender's help. To circumvent this problem, we can apply the same techniques used previously in confirmer signatures, namely bind the digital signature to its corresponding signcryption. This translates for CtEaS in producing the digital signature on both the commitment and the encryption. Similarly to confirmer signatures, the new CtEaS looses the parallelism of the original one, i.e. encryption and signature can no longer be carried out in parallel, however it has the advantage of resting on cheap encryption compared to the early one. The new StE uses similarly an encryption scheme from the KEM-DEM paradigm, and the digital signature is produced on both the encapsulation of the key (used later for encryption) and the message .

Unfortunately,  verifiability turns out to be a hurdle in both StE and CtEaS paradigms; the new (and old) StE paradigm encrypts the message to be signcrypted concatenated with a digital signature. As we are interested in proving the validity of the produced signcryption, we  need to exploit the homomorphic properties of the signature and of the encryption schemes in order to provide proofs of knowledge of the encrypted signature and message. As a consequence, the used encryption and signature need to operate on elements from a set with a known algebraic structure rather than on bit-strings. The same thing applies to the new (and old) CtEaS paradigm as encryption is performed on the concatenation of the message to be signcrypted and the opening value of the commitment scheme.

This leaves us with only the EtS paradigm to get efficient verifiable signcryption. In fact, the sender needs simply to prove knowledge of the decryption of a given ciphertext. Also, the receiver has to prove that a message is/isn't the decryption of a given ciphertext. Such proofs are easy to carry out if one considers the already mentioned class $\bbbe$. Moreover, $\sigExtract$ (similar to conversion in confirmer signatures) can be made efficient for many encryption schemes from the class $\bbbe$. Unfortunately, signcryptions from EtS are not anonymous, i.e. disclose the identity of the sender (anyone can check the validity of the digital signature on the ciphertext w.r.t. the sender's public key).

To  sum-up, EtS provides efficient verifiability but at the expense of the sender's anonymity. StE achieves better privacy  but at the expense of verifiability. It would be nice to have a technique that combines the merits of both paradigms while avoiding their drawbacks. This is the main contribution in the next paragraph.

\subsubsection{A new paradigm for efficient verifiable signcryption}

The core of the idea consists in first encrypting the message to be signcrypted using a public key encryption scheme, then applying the new StE to the produced encryption. The result of this operation in addition to the encrypted message form the new signcryption of the message in question. In other terms, this technique can be seen as a merge between EtS and StE; thus we can term it the "encrypt-then-sign-then-encrypt" paradigm (EtStE).

Consider a signature scheme $\Sigma$, an encryption scheme $\Gamma$, and another encryption $(\kem,\dem)$ derived from the KEM/DEM paradigm. On input the security parameter $\kappa$, generate the parameters $\param$ of these schemes. Note that a trusted authority is needed to generate the common reference strings for the NIZK proofs. We assume that signatures issued with $\Sigma$ can be written as $(r,s)$, where $r$ reveals no information about the signed message nor about the public signing key, and $s$ represents the ``significant'' part of the signature. Signcryptions from EtStE are as follows:
\begin{description}
\item \textbf{Key generation. }Invoke the key generation algorithms of the building blocks and set the sender's key pair to $(\Sigma.\pk,\Sigma.\sk)$, and the receiver's key pair to $(\{\Gamma.\pk,\kem.\pk\},\{\Gamma.\sk,\kem.\sk\})$.
\item \textbf{Signcrypt. }On a message $m$, produce an encryption $e=\Gamma.\encrypt_{\Gamma.\pk}(m)$ of $m$. Then fix a key $k$ along with its encapsulation $c$ using $\kem.\encrypt_{\kem.\pk}$, produce a signature $(r,s)$ on $c\|e$, and finally encrypt $s$ with $k$ using $\dem.\encrypt$. The signcryption of $m$ is the tuple $(e,c,\dem.\encrypt_k(s),r)$.
\item \textbf{Prove Validity. }Given a signcryption $\mu=(\mu_1,\mu_2,\mu_3,\mu_4)$ on a message $m$, the prover proves knowledge of the decryption of $\mu_1$, and of the decryption of $(\mu_2,\mu_3)$, which together with $\mu_4$ forms a valid digital signature on $\mu_2\|\mu_1$. The private input is either the randomness used to create $\mu$ or $\{\Gamma.\sk,\kem.\sk\}$.
\item \textbf{Unsigncrypt. }On a signcryption a $(\mu_1,\mu_2,\mu_3,\mu_4)$, compute $m=\Gamma.\decrypt_{\Gamma.\sk}(\mu_1)$ and $k=\kem.\decap_{\kem.\sk}(\mu_2)$. Check whether $(\dem.\decrypt_k(\mu_3),\mu_4)$ is valid signature on $\mu_2\|\mu_1$; if yes then output $m$, otherwise output $\perp$.
\item \textbf{Confirm/Deny. }On input a putative signcryption $\mu=(\mu_1,\mu_2,\mu_3,\mu_4)$ on a message $m$, use the receiver's private key to prove that $m$ is/isn't the decryption of $\mu_1$, and prove knowledge of the decryption of $(\mu_2,\mu_3)$, which together with $\mu_4$ forms a valid/invalid digital signature on $\mu_2\|\mu_1$. 
\item \textbf{Signature extraction. }On a valid signcryption $\mu=(\mu_1,\mu_2,\mu_3,\mu_4)$ on a message $m$, output a NIZK proof that $\mu_1$ encrypts $m$, in addition to \\$(\dem.\decrypt_{\kem.\decap{(\mu_2)}}(\mu_3),\mu_4)$.
\end{description}

Signcryptions from EtStE meet the following strong indistinguishability notion, which captures both the anonymity of the sender and the indistinguishability of the signcryptions. The notion informally denotes the difficulty to distinguish signcryptions on an adversarially chosen message from random elements in the signcryption space.

\begin{definition}[String Indistinguishability (SIND-CCA)]\label{def:sigcryINVCCA}
Let $\SC$ be a signcryption scheme, and let ${\mathcal A}$ be a PPTM. We consider the random experiment for $b \hasard \{0,1\}$ depicted in Experiment $\mathbf{Exp}_{\ensuremath{\SC},{\mathcal
A}}^{\textsf{sind-cca}\mbox{-}b}(1^\kappa)$. 
\begin{center}
\footnotesize
\begin{experiment}[Experiment $\mathbf{Exp}_{\ensuremath{\SC},{\mathcal
A}}^{\textsf{sind-cca}\mbox{-}b}(1^\kappa)$]

\item $\param \leftarrow \ensuremath{\SC}.\setup(1^\kappa)$;
\item $(\sks,\pks) \leftarrow \ensuremath{\SC}.\keygen_S(1^\kappa,\param)$;
\item $(\skr,\pkr) \leftarrow \ensuremath{\SC}.\keygen(1^\kappa,\param)$;

\item $(m^{\star},\mathcal{I}) \leftarrow {\mathcal A}^{\mathfrak{S}, \mathfrak{V}, \mathfrak{U}, \mathfrak{C}}
(\pks,\pkr)$;\\
\phantom{$(m^{\star},\mathcal{I}) \leftarrow$} $\left\vert
\begin{array}{l} 
\mathfrak{S} : m  \longmapsto \SC.\signcrypt_{\{\sks,\pks,\pkr\}}(m) \\ 
\mathfrak{V} : \mu \longmapsto \SC.\proveValidity(\mu,\pks,\pkr) \\
\mathfrak{U} : \mu \longmapsto \SC.\unsigncrypt_{\skr,\pkr,\pks}(\mu) \\
\mathfrak{C} : (\mu,m) \longmapsto \SC.\{\confirm,\deny\}(\mu,m,\pkr,\pks) \\
\mathfrak{P} : (\mu,m) \longmapsto \SC.\sigExtract(\mu,m,\pkr,\pks) \\
\end{array} \right.$ \\

\item $\mu_1^{\star} \leftarrow \SC.\signcrypt_{\{\sks,\pks,\pkr\}}(m^\star)$;

\item $\mu_0^{\star} \hasard \SC.\sf{space}$; $b \hasard \{0,1\}$

\item $d\leftarrow {\mathcal A}^{\mathfrak{S}, \mathfrak{V}, \mathfrak{U}, \mathfrak{C}}
(\mathcal{I},\mu_b^{\star},\pks,\pkc)$;\\
\phantom{$$} $\left\vert
\begin{array}{l} 
\mathfrak{S} : m  \longmapsto \SC.\signcrypt_{\{\sks,\pks,\pkr\}}(m) \\
\mathfrak{V} : \mu \longmapsto \SC.\proveValidity(\mu,\pks,\pkr) \\ 
\mathfrak{U} : \mu (\neq \mu^\star) \longmapsto \SC.\unsigncrypt_{\skr,\pkr,\pks}(\mu) \\
\mathfrak{C} : (\mu,m) (\neq (\mu^\star,m^\star)) \longmapsto
\SC.\{\confirm,\deny\}(\mu,m) \\
\mathfrak{P} : (\mu,m) (\neq (\mu^\star,m^\star)) \longmapsto
\SC.\sigExtract(\mu,m) \\
\end{array} \right.$ \\
\item Return $d$;
\end{experiment}
\end{center}

\noindent We define the \emph{advantage} of $\mathcal{A}$ \emph{via}:

$$ 
\mathbf{Adv}_{\ensuremath{\SC},{\mathcal A}}^{\textsf{sind-cca}}(1^\kappa)= \left\vert \Pr\left[\mathbf{Exp}_{\ensuremath{\SC},{\mathcal
A}}^{\mathsf{sind-cca-b}}(1^\kappa)=b\right] - \frac{1}{2} \right\vert.
$$
Given $(t,q_s,q_v,q_{u},q_{cd},q_{e}) \in \mathbb{N}^6$ and $\varepsilon \in [0,1]$, $\cal
A$ is called a $(t,\varepsilon,q_s,q_v,q_{u},q_{cd},q_{e})$-SIND-CCA adversary
against \ensuremath{\SC} if, running in time $t$  and issuing $q_s$ queries to the $\signcrypt$ oracle, $q_v$ queries to the
$\proveValidity$ oracle, $q_{u}$ queries to the $\unsigncrypt$
oracle, $q_{cd}$ queries to the $\{\confirm,\deny\}$
oracle, and $q_{e}$ to the $\sigExtract$ oracle, $\cal A$ has\\ $\mathbf{Adv}_{\ensuremath{\SC},{\mathcal
    A}}^{\mathsf{sind-cca}}(1^\kappa) \geq \varepsilon$. The scheme
\ensuremath{\SC} is
$(t,\varepsilon,q_s,q_v,q_{u},q_{cd},q_{e})$-SIND-CCA secure if no
$(t,\varepsilon,q_s,q_v,q_{u},q_{cd},q_{e})$-SIND-CCA adversary against it
exists.

\end{definition}

\begin{thm}
\label{thm:EtStE-forgeryGeneric}
Given $(t,q_s) \in \mathbb{N}^2$ and $\varepsilon \in [0,1]$, the above construction is ($t,\epsilon,q_s$)-EUF-CMA secure if the underlying digital signature scheme is ($t,\epsilon,q_s$)-EUF-CMA secure.
\qed \end{thm}

\begin{thm}
\label{thm:EtStE-indistinguishability}
Given $(t,q_s,q_v,q_{u},q_{cd},q_{e}) \in \mathbb{N}^6$ and $(\varepsilon,\epsilon') \in
[0,1]^2$, the construction proposed above is ($t,\epsilon,q_s,q_v,q_u,\-q_{cd},q_{e}$)-SIND-CCA
secure if it uses a $(t,\epsilon_s,q_s)$-SEUF-CMA secure digital signature, a $(t,\epsilon_e)$-INV-CPA secure encryption, a ($t,\epsilon_d$)-INV-OT secure DEM with injective encryption, and a ($t+q_s(q_v+q_u+q_{cd}+q_{e}),\epsilon(1-\epsilon_e)(1-\epsilon_d)(1-\epsilon_s)^{q_v+q_{cd}+q_u+q_{pv}}$)-IND-CPA secure KEM.
\end{thm}

\begin{proof}(Sketch)
From an SIND-CCA adversary $\A$ against the construction, we construct an algorithm $\R$ that IND-CPA break the KEM underlying the construction. $\R$ gets the public parameters of the KEM from her challenger and chooses further the remaining building blocks, i.e. the DEM, the signature, and the encryption scheme. Simulation of $\A$'s environment is done using the key pairs of the used signature and encryption schemes, in addition to a list in which $\R$ maintains the queries, their responses and the intermediate values used to generate these responses.

Eventually, $\cal A$ outputs a challenge messages $m$. $\cal R$ will encrypt, in $e$, the message $m$. Next, she produces a signature $(s,r)$ on $c\|e$, where $(c,k)$ is her challenge. Finally, $\cal R$ encrypts $s$ in $e_\dem$ using $k$, and outputs $\mu=(e,c,e_\dem,r)$ as a challenge signcryption. Since the used encryption is INV-CPA secure by assumption, then information about $m$ can only leak from $(c,e_\dem,r)$. If $k$ is the decapsulation of $c$, then $\mu$ is a valid signcryption of $m$, otherwise it is a random element from the signcryption space due to the assumptions on the used components (encryption scheme is INV-CPA, the DEM is INV-OT, and finally $r$ reveals no information about $e$ nor about sender's key). The rest follows as in the proof of Theorem \ref{thm:KEM-invisibility}.
\qed \end{proof}

\paragraph{Instantiations}The $\proveValidity$ and $\{\confirm,\deny\}$ protocols comprise the following sub-protocols:

\begin{enumerate}
\item Proving knowledge of the decryption of a ciphertext produced using the encryption scheme $\Gamma$.
\item Proving that a message is/isn't the decryption of a certain ciphertext produced using $\Gamma$.
\item Proving knowledge of the decryption of a ciphertext produced using $(\kem,\dem)$, and that this decryption forms a valid/invalid digital signature, issued using $\Sigma$, on some known string.
\end{enumerate}

It is natural to instantiate the encryption $\Gamma$ from the class $\bbbe$ described in Definition \ref{def:cryptosystemClass-EoS}. With this choice, the first two sub-protocols can be efficiently carried out as depicted in Figure \ref{fig:PKcryptosystemEoS} and Figure \ref{fig:conf-denySOC} respectively. Moreover, one can consider encryptions  from the class $\bbbe$ that are derived from the KEM/DEM paradigm, in addition to signatures from the class $\bbbs$ described in Definition \ref{def:signatureClass}. The last sub-protocol boils down then to the protocol depicted in Figure \ref{fig:conf-deny}.

Finally, for the $\sigExtract$ algorithm, we refer to the solutions adopted in confirmer signatures (described in Paragraph \ref{subsubsection:conversion-CDCS}) when it comes to producing a NIZK proof of the correctness of a decryption.

\subsection{Extension to multi-user signcryption}
So far, we considered signcryption schemes in the two-user setting, i.e. a single sender interacts with a single receiver. A signcryption scheme secure in the two-user setting does not necessarily mean that it conserves this security in the multi-user setting. In fact, the unforgeability adversary in the latter mode is allowed to  return a forgery on a message $m^\star$ that may have been queried before but w.r.t. a receiver's key different from the target receiver's key $\pkr^\star$. Moreover, the indistinguishability adversary is allowed to ask the unsigncryption or (public) verification of the challenge w.r.t. any receiver's key except that of the target receiver. However many works \cite{AnDodisRabin2002,MatsudaMatsuuraSchuldt2009} have proposed simple tweaks in order to derive multi-user security from two-user security. These techniques apply also to our constructions in order to guarantee security in the multi-user setting. For instance, the EtS paradigm in the multi-user setting departs from that of the two-user setting in the following elements:

\begin{enumerate}
\item It considers a tag-based encryption scheme where the tag is set to the public key of the sender $\pks$.
\item The digital signature is produced on the resulting ciphertext and on the public key of the receiver.
\end{enumerate}

Similarly, the EtStE paradigm in the multi-user setting deviates from that of the two-user one as follows:

\begin{enumerate}
\item It considers a tag-based KEM where the tag is set to the public key of the sender $\pks$.
\item The digital signature is produced on the resulting ciphertext and on the public key of the receiver.
\end{enumerate}

\section{Security Enhancement}
\label{sec:extensions}
In this section, we present two efficient transforms that upgrade the security in confirmer signatures/signcryption, i.e., allow to obtain online non-transferability and  insider invisibility/indistinguishability.

\subsection{Online non-transferability}
\label{subsec:onlineVsOffline-NT}
Online non-transferability as previously mentioned allows to avoid some attacks in which the
intended verifier, say $V$, interacts \emph{concurrently} with the genuine prover and a hidden malicious verifier $\widetilde{V}$ such that this latter gets convinced of the proven statement (validity or invalidity of a signature w.r.t. a given message). 

\noindent One way to circumvent this problem consists in using designated verifier proofs \cite{JakobssonSakoImpagliazzo1996}. In fact, these proofs can be conducted by both the prover and the verifier. When a verifier receives such a proof, he will be convinced of the validity of the underlying statement since he has not proved it himself. However, he cannot convince a third party of the validity of the statement as he can himself perfectly simulate the answers sent by the prover.

\noindent A generic construction of designated verifier proofs from $\Sigma$ protocols was given in \cite{ShahandashtiSafavi-Naini2008}. The idea consists in proving either the statement in question or proving knowledge of the verifier's private key. This is achieved using \emph{proofs of disjunctive knowledge} if the proof of the statement in question and the proof of knowledge of the verifier's private key are both $\Sigma$ protocols. We refer to \cite{ShahandashtiSafavi-Naini2008} for the details.

\indent Getting back to our problem, our already mentioned confirmation/denial protocols can be shown to be Sigma protocols if the proof $\zkp$ in the last round is non-interactive (this would necessitate the presence of a trusted authority). In this case, they can be efficiently transformed into designated verifier proofs,  providing therefore the required online non-transferability for the resulting signatures. 

\subsection{Insider invisibility/indistinguishability}
Insider invisibility/indistinguishability does not seem plausible without IND-CCA encryption. In fact, the reduction should be able to answer any query submitted by the adversary; this latter who can create, using his signing key, valid queries (confirmer signatures or signcryptions) without the help of the reduction. A suitable candidate for IND-CCA encryption that fits nicely within our framework is encryption obtained from the Canetti-Halevi-Katz like transformation \cite{CanettiHaleviKatz2004,Kiltz2006}. This encryption is obtained by combining a weakly secure tag-based encryption (indistinguishable under selective-tag weak chosen-ciphertext attacks or IND-st-wCCA) and a strongly unforgeable one-time signature; the combination enjoys the required IND-CCA security while proffering good verifiablity properties (the weakly secure encryption ought to be homomorphic, e.g \cite{CashKiltzShoup2008}). 

Therefore, StE, EtS, CtEaS, or EtStE will be used as follows. First generate a pair of public/private keys for the one-time signature. Then proceed as dictated by the  paradigms with the exception of producing the required encryption using an IND-st-wCCA tag-based encryption with the public key of the one-time signature as a tag, and finally signing all the produced quantities (that form the confirmer signature/signcryption) using the private key (and the signing algorithm) of the one time signature. The new confirmer signature/\-signcryption is increased by the verification key and the one-time signature, which amounts to four group elements when using Boneh-Boyen's one-time signature, however it enjoys a full insider invisibility while remaining efficiently verifiable. Note also that the produced signature/signcryption remains secure even in the presence of an  insider invisibility/indistinguishability adversary who is only restricted from querying the challenge for verification/decryption. This is definitely a stronger attack model than that adopted in general for the EtS/CtEaS paradigm, e.g. in \cite{GentryMolnarRamzan2005,AnDodisRabin2002}.

\section{Perspectives}
\label{sec:summary}
 In this paper, we studied the classical paradigms used to build many opaque signatures, that are StE, EtS, and CtEtS. We showed  using an increasingly popular tool, namely meta-reductions, that StE and CtEaS require expensive encryption in order to provide a reasonable security level for the resulting construction. This is due to an intrinsic weakness in those paradigms which consists in the possibility of obtaining the opaque signature without the help of the signer. Next, we proposed some adjustments to these paradigms which circumvent this weakness and allow to rest on cheap encryption without compromising the security level of the result. We further gave many practical instantiations of these paradigms which efficiently implement the verifiability feature in the constructions, i.e. the possibility to prove the validity of the opaque signature.

Our analysis accepts many possible extensions. We note in the following the most immediate ones:

\paragraph{Verifiably encrypted signatures} A verifiably encrypted signature (VES) allows a signer to encrypt a
  signature under the public key of a trusted party (the adjudicator), while
  maintaining public signature verifiability without interactive proofs. Actually, verifiability is usually achieved by considering special classes of signature/encryption schemes. For instance, the class of encryptions includes schemes where any pair of message and corresponding ciphertext, under a given key, satisfies a  relation confined by some efficiently computed map, say $f$. It is obvious that such encryption schemes  cannot be NM-CPA nor IND-CPA secure  due to the map $f$ which allows to efficiently check whether a ciphertext encrypts a given message under some given key. To rule out OW-CCA encryption, one could similarly consider a meta-reduction $\cal M$ which forbids existence of key-preserving reductions from OW-CCA security of the encryption to the opacity of the VES: $\cal M$ can ask the reduction for a VES on an arbitrary message, say $m$, then queries the CCA oracle for the decryption of this VES (it is possible to make this query as it is different from the challenge ciphertext). The result of this query, along with $m$, forms a valid answer of the opacity adversary. Again, the interpretation of these impossibility results is that the opacity adversary can create VES without the help of the signer by simply re-encrypting the extracted signatures. It would be interesting to envisage the previously presented solutions in order to make the opacity in VES rest on the CPA security of the underlying encryption.

\paragraph{Group signatures}
Group signatures, introduced by Chaum and Van Heyst \cite{ChaumvanHeyst1991}, allow members of a group to anonymously sign messages on behalf of the whole group. However, to prevent abuses, the group is controlled by a group manager that has the ability to \textit{open} the group signature, \textit{i.e.} to identify the signer of a message. A generic construction of group signatures from the StE paradigm \cite{BonehShacham2004} consists in encrypting the identity of the user in the public key
of the group manager, then providing a signature of knowledge (NIZK of the plaintext underlying the encryption and on the SDH solution) of the message to be signed. The used encryption scheme has to be CCA secure in order to provide full anonymity of the group signature. There exists also a weaker notion of anonymity (than the full anonymity), called selfless anonymity, where the adversary does not have the signing key of the target users (similar to outsider security in CDCS/signcryption). This suggests to carry out the same analysis (provided earlier) in order to study the exact security needed for the encryption scheme in order to derive fully/selfless anonymous group signatures.  Note that we gave in \cite{ElAimaniSanders2012} a generic construction of fully anonymous group signatures using IND-st-wCCA tag based encryption combined with strongly unforgeable one-time signatures. Our construction, which uses many ideas presented earlier in this text, generalizes a well known group signature \cite{Groth2007}, and served as a basis for a recent proposal by \cite{Ghadafi2014} of the same primitive which reduces the trust on the group manager by distributing the opening procedure.
 


\bibliographystyle{amsplain}

\providecommand{\bysame}{\leavevmode\hbox to3em{\hrulefill}\thinspace}
\providecommand{\MR}{\relax\ifhmode\unskip\space\fi MR }
\providecommand{\MRhref}[2]{%
  \href{http://www.ams.org/mathscinet-getitem?mr=#1}{#2}
}
\providecommand{\href}[2]{#2}

\end{document}